\documentclass[14pt,a4paper]{extarticle} 

\usepackage[margin=1in]{geometry}
\usepackage{pdflscape}
\usepackage{everypage}
\usepackage{lipsum}

\newcommand{\Lpagenumber}{\ifdim\textwidth=\linewidth\else\bgroup
	\dimendef\margin=0 
	\ifodd\value{page}\margin=\oddsidemargin
	\else\margin=\evensidemargin
	\fi
	\raisebox{\dimexpr -\topmargin-\headheight-\headsep-0.5\linewidth}[0pt][0pt]{%
		\rlap{\hspace{\dimexpr \margin+\textheight+\footskip}%
			\llap{\rotatebox{90}{\thepage}}}}%
	\egroup\fi}
\AddEverypageHook{\Lpagenumber}%
\usepackage[contents={}]{background}

\usepackage{amsmath}
\usepackage{amsthm}
\usepackage{amssymb}
\usepackage{algorithm}
\usepackage{graphicx}
\usepackage{graphicx,psfrag,epsf}
\usepackage{subcaption}
\usepackage{enumitem}
\usepackage{thmtools}
\usepackage{amsmath,siunitx,booktabs}
\usepackage{multirow}
\usepackage{float}
\usepackage{array} 
\usepackage[T1]{fontenc}
\usepackage{booktabs,graphicx,makecell,siunitx,eqparbox,ragged2e}

\usepackage{tikz}
\usetikzlibrary{arrows.meta}

\usepackage{comment}

\usepackage[utf8]{inputenc}
\usepackage[hidelinks]{hyperref}
\hypersetup{
	unicode,
	colorlinks=true,
	linktoc=none,
	urlcolor=blue, 
	linkcolor=blue, 
	citecolor=blue,
	pdfauthor={Author One, Author Two, Author Three},
	pdftitle={A simple article template},
	pdfsubject={A simple article template},
	pdfkeywords={article, template, simple},
	pdfproducer={LaTeX},
	pdfcreator={pdflatex}
}

\usepackage{url} 

\usepackage{natbib}

\usepackage{etoolbox}

\makeatletter

\patchcmd{\NAT@citex}
{\@citea\NAT@hyper@{%
		\NAT@nmfmt{\NAT@nm}%
		\hyper@natlinkbreak{\NAT@aysep\NAT@spacechar}{\@citeb\@extra@b@citeb}%
		\NAT@date}}
{\@citea\NAT@nmfmt{\NAT@nm}%
	\NAT@aysep\NAT@spacechar\NAT@hyper@{\NAT@date}}{}{}

\patchcmd{\NAT@citex}
{\@citea\NAT@hyper@{%
		\NAT@nmfmt{\NAT@nm}%
		\hyper@natlinkbreak{\NAT@spacechar\NAT@@open\if*#1*\else#1\NAT@spacechar\fi}%
		{\@citeb\@extra@b@citeb}%
		\NAT@date}}
{\@citea\NAT@nmfmt{\NAT@nm}%
	\NAT@spacechar\NAT@@open\if*#1*\else#1\NAT@spacechar\fi\NAT@hyper@{\NAT@date}}
{}{}

\makeatother

\usepackage{cleveref} 
\theoremstyle{plain}

\newtheorem{proposition}{Proposition}

\theoremstyle{definition}
\newtheorem{definition}{Definition}

\usepackage{graphicx, color}
\graphicspath{{fig/}}

\usepackage{algorithm, algpseudocode} 
\usepackage{mathrsfs} 

\usepackage{lipsum}

\usepackage{dashrule}

\title{Stationary birth-death processes generating inflation-deflation distributions: Avoiding the issue of dominance}
\author{Wanrudee Skulpakdee$^1$ \thanks{Corresponding author.} \and Mongkol Hunkrajok$^2$}

\date{
	$^1$Graduate School of Applied Statistics, National Institute of Development Administration, Bangkok, Thailand.\\ \texttt{wanrudee.sku@nida.ac.th}\\%
	$^2$Independent Researcher, Bangkok, Thailand.\\ 
	\texttt{hunkrajokmongkol@gmail.com}\\[2ex]%
}

\begin{document}
	\addtocontents{toc}{\protect\iffalse}
	\maketitle
	\sloppy
	\renewcommand\thmcontinues[1]{Continued}
	\begin{abstract}
	A mixture of two or more count distributions has become deeply embedded in the analysis of excess counts, often relative to the stationary (equilibrium) distributions of birth-death processes such as the geometric, Poisson, Poisson-Lindley (PL), negative binomial (NB), hyper-Poisson (HP), and Conway-Maxwell-Poisson (CMP) distributions. However, the mechanism by which excess counts arise—namely, through modifications of the birth and death rates in the base distributions—has not yet been directly examined in the research literature. All well-known inflation mixture distributions are, in fact, parameterizations of the stationary distributions of birth-death processes. Thus, although the resulting distributions share the same shapes, they arise from distinct mechanisms and are not equivalent in regression analyses. This paper focuses on inflation-deflation stationary distributions arising from modified birth-death processes that form an exponential family and introduces two types of such distributions.\\
		
	\end{abstract}

	
	\linespread{1.125}\selectfont
	
	\section{Introduction}
	\label{sec:intro}
	Count data often exhibit excess counts relative to the stationary distributions of birth-death processes, including the geometric, Poisson, PL, NB, HP, and CMP. This phenomenon has motivated the development of numerous alternative distributions for use in applied statistics. Such patterns may result from unobserved heterogeneity, which arises when responses are drawn from a population composed of multiple sub-populations. To address this, researchers frequently use inflated models based on mixture mechanisms. For example, a Google Scholar search from 2020 to 2025 identified 17,000 articles containing the phrase ``inflated count data model'' anywhere in the articles.

	The hurdle model (\citealt{Cragg71}, \citealt{Mullahy86}) is a form of zero-inflated mixture model that separates between zero and non-zero counts. A degenerate distribution (not a Bernoulli distribution) with probability $\pi$ represents zero counts, whereas non-zero counts are modeled using a zero-truncated distribution with probability $1-\pi$. This mixture model can be expressed as
	\begin{equation}
		p(n,\boldsymbol{\theta},\pi) =\left\{\begin{array}{cll}
			\displaystyle \pi & , & n=0 \\
			\displaystyle \frac{(1-\pi)b(n,\boldsymbol{\theta})}{1-b(0,\boldsymbol{\theta})} & , & n>0, \\
		\end{array}\right.
		\label{eq:hurdle}
	\end{equation}
	where $\pi$ is the probability of zero in a Bernoulli distribution, $\boldsymbol{\theta}=[\theta_1,\theta_2,...,\theta_l] \in R^l$, and $b(n,\boldsymbol{\theta})$ is a base distribution. Since the hurdle is set at zero, $p(0,\boldsymbol{\theta},\pi)$ is independent of $b(0,\boldsymbol{\theta})$. Additionally, $b(n,\boldsymbol{\theta})/(1-b(0,\boldsymbol{\theta}))$ refers to the zero-truncated form of a standard discrete distribution, which includes the six previously mentioned stationary distributions. \citet[p.~256]{Puig24} recently introduced a stationary birth-death process that generates the hurdle model \eqref{eq:hurdle} with a Poisson base distribution.
	
	This model can be expanded to include multiple hurdles. The $(l+q+1)$-parameter family of multiple hurdle mixture models may have a probability function
	\begin{equation}
		p(n,\boldsymbol{\theta},\boldsymbol{\pi}) =\left\{\begin{array}{cll}
			\displaystyle \pi_n & , & n\le q \\
			\displaystyle \frac{\left( 1-\sum_{k=0}^{q} \pi_k\right)b(n,\boldsymbol{\theta})}{ 1-\sum_{k=0}^{q} b(k,\boldsymbol{\theta})} & , & n>q, \\
		\end{array}\right.
	\end{equation}
	where $\boldsymbol{\pi}=\left[ \pi_0,\pi_1,...,\pi_q \right]$, $0<\pi_k<1$ for $k=0,1,...,q$ and $\sum_{k=0}^{q} \pi_k<1$. These extended distributions have received limited attention due to a lack of practical applications \citep[p.~181]{Winkelmann08}.
	
	If we let $\pi$ depend on $b(0,\boldsymbol{\theta})$, that is, $\pi=\omega+(1-\omega)b(0,\boldsymbol{\theta})$, then the mixture model \eqref{eq:hurdle} becomes the zero-inflated mixture model (\citealt{Singh63}, \citealt{Johnson69})
	\begin{equation}
		p(n,\boldsymbol{\theta},\omega) =\left\{\begin{array}{cll}
			\displaystyle \omega+(1-\omega)b(0,\boldsymbol{\theta}) & , & n=0 \\
			\displaystyle (1-\omega)b(n,\boldsymbol{\theta}) & , & n>0, \\
		\end{array}\right.
		\label{eq:zip}
	\end{equation}
	where $\omega$ is a probability at zero in a Bernoulli distribution, and $p(0,\omega ,\boldsymbol{\theta})$ depends on $b(0,\boldsymbol{\theta})$. Its zero probability is dominated by $\omega$ when $b(0,\boldsymbol{\theta})$ is small. One way to address this issue is to include covariates in $\omega$, but this may introduce unnecessary additional parameters. This distribution is a mixture of a degenerate distribution at zero and an untruncated count distribution, unlike \eqref{eq:hurdle}. Both models have the same shape, log-likelihood, and fitted probabilities when no regressors are included.
	
	Most applications use the untruncated Poisson distribution as the base for \eqref{eq:zip}, though other untruncated distributions, such as NB and CMP, are also considered. Each serves as the stationary distribution of a birth-death process. The Poisson and NB distributions also represent pure-birth processes (\citealt[p.~13-14]{Boswell70}, \citealt[p.~433]{Faddy97}, \citealt[p.~587]{Janardan05}). Since $\omega$ ranges from 0 to 1, the mixture concept applies to \eqref{eq:zip}. If $\omega$ is not restricted to values greater than or equal to 0, the non-mixture model \eqref{eq:zip} can address zero-deflated count data. Although the mixture mechanism does not produce a zero-deflated distribution, a stationary birth-death process can account for this phenomenon.
	
	Distinct parameterizations of $\pi$ yield different zero-inflation models but may remove the interpretation as a two-distribution mixture. \citet[p.~226]{Haslett22} replace $\pi$ with $e^{\psi}b(0,\boldsymbol{\theta})/(1+(e^{\psi}-1)b(0,\boldsymbol{\theta}))$ in \eqref{eq:hurdle}, yielding a new zero-inflation distribution
	\begin{equation}
		p(n,\boldsymbol{\theta},\psi)=\left\{\begin{array}{cll}
			\displaystyle \frac{e^{\psi}b(0,\boldsymbol{\theta})}{1+(e^{\psi}-1)b(0,\boldsymbol{\theta})} & , & n=0 \\
			\displaystyle \frac{b(n,\boldsymbol{\theta})}{1+(e^{\psi}-1)b(0,\boldsymbol{\theta})}  & , & n>0, \\
		\end{array}\right.
		\label{eq:haslett}
	\end{equation}
	 where $\psi$ is a real number. In contrast to \eqref{eq:zip}, the additional parameter $\psi$ does not dominate the zero probability when $b(0,\boldsymbol{\theta})$ is small. We can transform \eqref{eq:haslett} to demonstrate that it does not represent a mixture model. \citet[p.~226]{Haslett22} interpret \eqref{eq:haslett} as a multiplicative modification of the odds ratio for a zero count, that is, $p(0,\boldsymbol{\theta},\psi)/(1-p(0,\boldsymbol{\theta},\psi))=e^\psi b(0,\boldsymbol{\theta})/(1-b(0,\boldsymbol{\theta}))$. This approach merely rearranges \eqref{eq:haslett} to describe the model, leaving the underlying mechanism that gives rise to it unclear. Nevertheless, a birth-death process at equilibrium generates this distribution, as detailed in Section \ref{sec:IDS_model}. Since the shape distributions of the three models discussed above are identical, as indicated by their reparameterizations, their maximum likelihoods must also be equal. Consequently, it is not possible to determine, using maximum likelihood, from which of the three models the sample originated.

	\begin{figure}[!t]
		\centering
		\begin{tikzpicture}
			\draw[scale=1, thick] (-7.4, 0) -- (-2.8, 0) node[below=0.1cm] {$w$};
			\draw[ultra thick] (-5.1, 0) -- (-3.6, 0);
			\draw[thick] (-5.1, 0.1) -- (-5.1, -0.1) node[below] {$0$};
			\draw[thick] (-3.6, 0.1) -- (-3.6, -0.1) node[below] {$1$};
			\filldraw[black,fill=white] (-5.1,0) circle (3pt);
			\filldraw[black,fill=white] (-3.6,0) circle (3pt);
			
			\draw[scale=1, thick] (-2.3, 0) -- (2.3, 0) node[below=0.1cm] {$\alpha$};
			\draw[ultra thick] (-1, 0) -- (0.5, 0);
			\draw[thick] (-1, 0.1) -- (-1, -0.1) node[below] {$0$};
			\draw[thick] (0.5, 0.1) -- (0.5, -0.1) node[below] {$1$};
			\filldraw[black,fill=white] (-1,0) circle (3pt);
			\filldraw[black,fill=white] (0.5,0) circle (3pt);
			
			\draw[scale=1, ultra thick] (2.8, 0) -- (7.4, 0) node[below=0.05cm] {$\psi$};
			\draw[thick] (5.1, 0.1) -- (5.1, -0.1) node[below] {$0$};
			
			\draw[thick,-{Latex[length=3mm]}] (-5,0.5) arc (120:60:9.8);
			\node[fill=white] at (0.2,2) {$\psi = \log \left ( \frac{w}{1-w} \right )$};
			\draw[thick,-{Latex[length=3mm]}] (-4.9,-1) arc (230:310:3.5);
			\node[fill=white] at (-2.8,-2) {$\alpha=w$};
			\draw[thick,-{Latex[length=3mm]}] (0.4,-1) arc (230:310:3.5);
			\node[fill=white] at (2.8,-2) {$\psi=\log \left ( \frac{\alpha}{1-\alpha} \right )$};
		\end{tikzpicture}		
		\begin{tikzpicture}
			\draw[scale=1, thick] (-7.4, 0) -- (-2.8, 0) node[below=0.1cm] {$w$};
			\draw[ultra thick] (-6.6, 0) -- (-3.6, 0);
			\draw[thick] (-6.6, 0.1) -- (-6.6, -0.1) node[below]
			{$\frac{b(0,\boldsymbol{\theta})}{b(0,\boldsymbol{\theta})-1}$};
			\draw[thick] (-5.1, 0.1) -- (-5.1, -0.1) node[below] {$0$};
			\draw[thick] (-3.6, 0.1) -- (-3.6, -0.1) node[below] {$1$};
			\filldraw[black,fill=white] (-6.6,0) circle (3pt);
			\filldraw[black,fill=white] (-3.6,0) circle (3pt);
			
			\draw[scale=1, thick] (-2.3, 0) -- (2.3, 0) node[below=0.1cm] {$\alpha$};
			\draw[scale=1, ultra thick] (-1, 0) -- (2.3, 0);
			\draw[thick] (-1, 0.1) -- (-1, -0.1) node[below] {$0$};
			\draw[thick] (0.5, 0.1) -- (0.5, -0.1) node[below] {$1$};
			\filldraw[black,fill=white] (-1,0) circle (3pt);
			
			\draw[scale=1, ultra thick] (2.8, 0) -- (7.4, 0) node[below=0.05cm] {$\psi$};
			\draw[thick] (5.1, 0.1) -- (5.1, -0.1) node[below] {$0$};
			
			\draw[thick,-{Latex[length=3mm]}] (-5,0.5) arc (120:60:9.8);
			\node[fill=white] at (0.2,2) {$\psi = \log \left ( \frac{\omega+(1-\omega)b(0,\boldsymbol{\theta})}{(1-\omega)b(0,\boldsymbol{\theta})} \right )$};
			\draw[thick,-{Latex[length=3mm]}] (-4.9,-1) arc (230:310:3.5);
			\node[fill=white] at (-2.8,-2) {$\alpha =\frac{\omega+(1-\omega)b(0,\boldsymbol{\theta})}{(1-\omega)b(0,\boldsymbol{\theta})}$};
			\draw[thick,-{Latex[length=3mm]}] (0.4,-1) arc (230:310:3.5);
			\node[fill=white] at (2.8,-2) {$\psi=\log (\alpha)$};
		\end{tikzpicture}
		\caption{Mapping by the common (top) and new (bottom) link functions of \eqref{eq:zip}}
		\label{fig:w_map}
	\end{figure}
	
	A commonly used link function for \eqref{eq:zip} that connects $\omega$ to $\psi$ is the logit function, defined as $\psi=\log (\omega/(1-\omega))$, as shown in Figure \ref{fig:w_map} (top) (\citealt[p.~181]{Ridout98}, \citealt[p.~189]{Winkelmann08}, \citealt[p.~141]{Cameron13}, \citealt[p.~3]{Feng21}). This is just a parameterization of \eqref{eq:zip} with $\omega=e^\psi/(1+e^\psi)$, so
	\begin{equation}
		p(n,\boldsymbol{\theta},\psi) =\left\{\begin{array}{cll}
			\displaystyle \frac{e^\psi+b(0,\boldsymbol{\theta})}{1+e^\psi} & , & n=0 \\
			\displaystyle \frac{b(n,\boldsymbol{\theta})}{1+e^\psi} & , & n>0. \\
		\end{array}\right.
		\label{eq:zip_r}
	\end{equation} 
	The model \eqref{eq:zip} with $\omega=0$, representing the base distribution, is not defined because zero is not a possible argument for the logit function and cannot be included in \eqref{eq:zip_r}. The logit link function maps only the inflation parameter space $\omega \in
	\left(0, 1\right)$ onto $\psi \in (-\infty,\infty )$. If $\omega$ in \eqref{eq:zip} takes negative values, it results in a zero-deflation model; in this case, the distribution cannot be represented as a mixture distribution. We may map $\omega \in (0,1)$ by $\alpha=\omega$ onto $\alpha \in (0,1)$ and then onto $\psi \in (-\infty,\infty )$ by $\psi=\log (\alpha/(1-\alpha))$(see Figure \ref{fig:w_map}, top). If we set $\alpha =(\omega+(1-\omega)b(0,\boldsymbol{\theta}))/(1-\omega)b(0,\boldsymbol{\theta})$, the function maps $\omega \in (b(0,\boldsymbol{\theta})/(b(0,\boldsymbol{\theta})-1),0)$ and $\omega \in [0,1)$ onto $\alpha\in (0,1)$ and $\alpha\in [1,\infty)$, respectively (see Figure \ref{fig:w_map}, bottom). Therefore, the link function that maps both positive and negative parameter spaces of $\omega$ onto $\psi \in (-\infty,\infty )$ is
	\begin{equation}
		\psi = \log \left (\frac{\omega+(1-\omega)b(0,\boldsymbol{\theta})}{(1-\omega)b(0,\boldsymbol{\theta})}  \right ).
		\label{eq:ww_link}
	\end{equation}
	This is an immediate application of \Cref{equi_shape} in Section \ref{sec:IDS_model}. Reparameterizing $\omega$ as $(e^\psi-1)b(0,\boldsymbol{\theta})/(1+(e^\psi-1)b(0,\boldsymbol{\theta}))$ transforms \eqref{eq:zip} into \eqref{eq:haslett}.  In contrast to \eqref{eq:zip_r}, the zero-inflated-deflated model \eqref{eq:haslett} of \citet{Haslett22} includes the base distribution of \eqref{eq:zip}. Figure \ref{fig:w_map} (bottom) illustrates that the deflation $\omega \in (b(0,\boldsymbol{\theta})/(b(0,\boldsymbol{\theta})-1),0)$ and  inflation $\omega \in (0,1)$ intervals are mapped to the left and right real lines, respectively.
	
	Extensions of the mixture model \eqref{eq:zip} are relatively straightforward. An obvious approach is to include more than a single inflation. For example, the multiple-inflation mixture model (\citealt{Su13}, \citealt{Böh25}) specifies
	\begin{equation}
		p(n,\boldsymbol{\theta},\boldsymbol{\omega})=\left\{\begin{array}{cll}
			\displaystyle \omega_n+(1-\sum_{i=0}^{m}\omega_{n_i})b(n,\boldsymbol{\theta})& , &n \in \mathcal{F}  \\
			\displaystyle (1-\sum_{i=0}^{m}\omega_{n_i})b(n,\boldsymbol{\theta})& , &\text{otherwise},\\ 
		\end{array}\right.
		\label{eq:mix_zk}
	\end{equation}
	where $\boldsymbol{\omega}=[\omega_{n_0},\omega_{n_1},...,\omega_{n_m}]$, and $\mathcal{F}= \left\{n_0,n_1,...,n_m\right\} \subset \left\{ 0,1,2,...\right\}$. In \eqref{eq:mix_zk}, we scale the base distribution $b(n,\boldsymbol{\theta})$ by the mixing proportion $(1-\sum_{i=0}^{m}\omega_{n_i})$, and we must add the mixing proportions $\omega_{n_0},\omega_{n_1},...,$and $\omega_{n_m}$ of $n_0,n_1,...,$ and $n_m$ counts to the scaled base distribution to ensure that the probabilities sum to one. The most common cases are $\mathcal{F}= \left\{ 0\right\}$, $\mathcal{F}= \left\{ 0,1\right\}$, $\mathcal{F}= \left\{ 0,2\right\}$, and $\mathcal{F}= \left\{ 0,q\right\}$. Each one leads to a different model: zero-inflation \citep{Lambert92}, zero-and-one-inflation \citep{Zhang16}, zero-and-two-inflation \citep{Melkersson00}, and zero-and-$q$-inflation (\citealt{Lin13}, \citealt{Arora21}), in that order. The widespread popularity of the mixture mechanism is due, in part, to its ability to easily generate distributions for inflation count data, which, in turn, has led to an abundance of studies on these models over many years.
	
	\begin{proposition}
		\label{domination}
		(Domination) Suppose the base distribution $b(n,\boldsymbol{\theta})$ of \eqref{eq:mix_zk} has $\lambda$ as a main parameter and 
		\begin{equation*}
			\lim_{\lambda  \to \infty} b(n,\boldsymbol{\theta})=0 \:\:\: \text{for all}\:\: n.
		\end{equation*}
		Then
		\begin{equation*}
			\lim_{\lambda  \to \infty}p(n,\boldsymbol{\theta},\boldsymbol{\omega})=\left\{\begin{array}{cll}
				\displaystyle \omega_n& , &n\in \left\{n_0,n_1,...,n_m  \right\}  \\
				\displaystyle 0& , &\text{otherwise}.\\ 
			\end{array}\right.
		\end{equation*}
	\end{proposition}
	\noindent We omit the proof as it is obvious. An immediate consequence of the proposition is that the model \eqref{eq:mix_zk} encounters a dominating issue, which is exactly the same as \eqref{eq:zip}. Varying all $\omega_n$ with covariates can lead to another issue, that is, non-parsimonious.
	
	This extension model has become fairly popular in the research literature; however, it may not correctly describe inflation count data in general. \citet{Lin13} analyze health survey count data to demonstrate the effectiveness of their zero-and-$q$-inflation Poisson model. We can view the model as a mixture of three count distributions: the degenerate distributions at zero and $q$, and the Poisson distribution. The health survey count data represent the number of Pap smear tests performed on women over a six-year period, sourced from the National Health Interview Survey (NHIS). The survey has two questions: `Have you EVER HAD a Pap test?' and `How many Pap tests have you had in the LAST 6 YEARS?' We can use the first question to divide the women into two groups: the first group has only a zero count, while the second group has both zero and non-zero counts. These health survey data have only two groups, and the second group contains an excess of sixes. Consequently, the zero-and-six-inflated Poisson model is inappropriate for modeling these count data. We suggest reverting to a zero-inflated model with a single-unusual-event (SUE) base distribution (\citealt{Skul22}, \citealt{Hun25}). The SUE distribution should have an unusual event at seven because the women in the second group may not want to have the seventh Pap test (i.e., an appearance of many sixes). In-depth experiments are required to prove this conjecture. The mechanism generating this inflated (SUE) distribution is a pure-birth process with a non-monotonic birth rate sequence, which is not the focus of this paper.
	
	A birth-death process at equilibrium is an effective mechanism for generating count data distributions. With the birth rate $\gamma_n=\gamma, n=0,1,...,$ and the death rate $\mu_n = n \mu, n=1,2,...,$, this process yields the widely used Poisson distribution for analyzing count data. In this case, the birth and death rates lead to equidispersion, where the variance equals the mean. If count data exhibit overdispersion (variance > mean) or underdispersion (variance < mean), using stationary distributions with modified Poisson birth and death rates, such as the geometric, PL, NB, HP, CMP, and weighted Poisson (WP) (\citealt[p.~569]{Castillo98}), is logical (see Tables \ref{table:table_ratio} and \ref{table:table_modify}). The birth-death rate sequences for the NB and CMP are $\left( \gamma_n=\left( n/r+1 \right)\gamma, \:\mu_n=n\mu \right)$ (\citealt[p.~14]{Boswell70}) and $\left( \gamma_n=\gamma, \:\mu_n=n^\nu\mu \right)$ (\citealt[p.~134]{Conway62}), respectively, but they are not uniquely determined. The Poisson distribution is a pure-birth process with $\gamma_n=\gamma$. This birth rate sequence results in equidispersion. Therefore, using pure-birth distributions with modified Poisson birth rates, such as geometric, NB, and SUE, is also logical. The birth-rate sequences are $\gamma_n=-(n/r+1) \log (1-\gamma/r)^r$ for the NB (\citealt[p.~433]{Faddy97}) and $\gamma_n= \alpha \gamma$ if $n=q$, and $\gamma$ otherwise for the SUE (\citealt[p.~44]{Skul22}).
	
	Equi-, over-, and underdispersed distributions can also display inflation or deflation. Figure \ref{fig:his_equi} in Section \ref{subsec:equidispersion} shows five equidispersed distributions: the first two exhibit zero-to-two inflation, and the last two show zero-to-two deflation compared to the Poisson distributions with the same ratio parameters. The parameter $\lambda$ is called the ratio parameter for the stationary birth-death process ($\lambda=\gamma/\mu$) and the rate parameter for the Poisson pure-birth process ($\lambda=\gamma$). Figure \ref{fig:wfig8} illustrates how these distributions represent zero-to-two inflation and deflation.
	
	\afterpage{
		\begin{table}[H]
			\setlength{\tabcolsep}{0.2em}	
			\small 
			\centering
			\caption{Birth-death-rate ratio sequences correspond to various examples of type 1 and type 2 stationary distributions with $\mathcal{F}= \left\{ 0,2\right\}$. The geometric, Poisson, PL, NB, HP, and CMP distributions are special cases within these families, so we can recognize them as stationary birth-death processes. Note that $\lambda>0$ except for the geometric and PL distribution, where $0<\lambda<1$.}    
			\label{table:table_ratio}	
			\makebox[\textwidth]{\begin{tabular}{ >{\raggedright}b{3cm} 	>{\raggedright\arraybackslash}b{12cm}}	
					
					\toprule
					Distribution & Ratio sequence; $\lambda=\gamma/\mu$\\[0.5ex]				
					\midrule
					Geometric & $\displaystyle \quad\,\lambda,\, \quad\lambda,\, \quad\;\lambda,\, \lambda,\,\lambda,\, ... $ \\[1ex]	
					
					Type 1 & $\displaystyle \frac{1}{\alpha_0} \lambda,\,\alpha_2 \lambda,\, \frac{1}{\alpha_2} \lambda,\, \lambda,\,\lambda,\, ... $ \\[1.5ex]			
					
					Type 2 & $\displaystyle \frac{1}{\varphi_0} \lambda,\,\quad\lambda,\, \frac{1}{\varphi_2}\lambda,\, \lambda,\,\lambda,\, ... $ \\[2.5ex]	
					
					Poisson & $\displaystyle \quad\,\lambda,\, \quad\frac{\lambda}{2},\, \quad\;\frac{\lambda}{3},\, \frac{\lambda}{4},\,\frac{\lambda}{5},\, ... $ \\[1.5ex]	
					
					Type 1 & $\displaystyle \frac{1}{\alpha_0} \lambda,\,\alpha_2 \frac{\lambda}{2},\, \frac{1}{\alpha_2} \frac{\lambda}{3},\, \frac{\lambda}{4},\,\frac{\lambda}{5},\, ... $ \\[1.5ex]			
					
					Type 2 & $\displaystyle \frac{1}{\varphi_0} \lambda,\,\quad\frac{\lambda}{2},\, \frac{1}{\varphi_2}\frac{\lambda}{3},\, \frac{\lambda}{4},\,\frac{\lambda}{5},\, ... $ \\[2.5ex]	
					
					PL & $\displaystyle \quad\,\frac{(1+2\lambda)\lambda}{1+\lambda},\, \quad\frac{(1+3\lambda)\lambda}{1+2\lambda},\, \quad\;\frac{(1+4\lambda)\lambda}{1+3\lambda},\, \frac{(1+5\lambda)\lambda}{1+4\lambda},\,\frac{(1+6\lambda)\lambda}{1+5\lambda},\, ... $ \\[1.5ex]	
					
					Type 1 & $\displaystyle \frac{1}{\alpha_0} \frac{(1+2\lambda)\lambda}{1+\lambda},\,\alpha_2 \frac{(1+3\lambda)\lambda}{1+2\lambda},\, \frac{1}{\alpha_2} \frac{(1+4\lambda)\lambda}{1+3\lambda},\, \frac{(1+5\lambda)\lambda}{1+4\lambda},\,\frac{(1+6\lambda)\lambda}{1+5\lambda},\, ... $ \\[1.5ex]			
					
					Type 2 & $\displaystyle \frac{1}{\varphi_0} \frac{(1+2\lambda)\lambda}{1+\lambda},\,\quad\frac{(1+3\lambda)\lambda}{1+2\lambda},\, \frac{1}{\varphi_2}\frac{(1+4\lambda)\lambda}{1+3\lambda},\, \frac{(1+5\lambda)\lambda}{1+4\lambda},\,\frac{(1+6\lambda)\lambda}{1+5\lambda},\, ... $ \\[2.5ex]				
					
					NB & $\displaystyle \quad\,\lambda,\,\quad\, \left ( \frac{1}{r}+1 \right ) \frac{\lambda}{2},\,\quad\;\; \left ( \frac{2}{r}+1 \right ) \frac{\lambda}{3},\, \left ( \frac{3}{r}+1 \right ) \frac{\lambda}{4},\,\left ( \frac{4}{r}+1 \right ) \frac{\lambda}{5},\, ... $ \\[1.5ex]
					
					Type 1 & $\displaystyle \frac{1}{\alpha_0} \lambda,\, \alpha_2 \left ( \frac{1}{r}+1 \right ) \frac{\lambda}{2},\, \frac{1}{\alpha_2} \left ( \frac{2}{r}+1 \right ) \frac{\lambda}{3},\, \left ( \frac{3}{r}+1 \right ) \frac{\lambda}{4},\,\left ( \frac{4}{r}+1 \right ) \frac{\lambda}{5},\, ... $	  \\[1.5ex]	
					
					Type 2 & $\displaystyle \frac{1}{\varphi_0} \lambda,\,\quad\, \left ( \frac{1}{r}+1 \right ) \frac{\lambda}{2},\, \frac{1}{\varphi_2} \left ( \frac{2}{r}+1 \right ) \frac{\lambda}{3},\, \left ( \frac{3}{r}+1 \right ) \frac{\lambda}{4},\,\left ( \frac{4}{r}+1 \right ) \frac{\lambda}{5},\, ... $ \\[2.5ex]	
					
					HP & $\displaystyle \quad\,\frac{\lambda}{\tau},\, \quad\frac{\lambda}{\tau+1},\, \quad\;\frac{\lambda}{\tau+2},\, \frac{\lambda}{\tau+3},\,\frac{\lambda}{\tau+4},\, ... $ \\[1.5ex]	
					
					Type 1 & $\displaystyle \frac{1}{\alpha_0} \frac{\lambda}{\tau},\,\alpha_2 \frac{\lambda}{\tau+1},\, \frac{1}{\alpha_2} \frac{\lambda}{\tau+2},\, \frac{\lambda}{\tau+3},\,\frac{\lambda}{\tau+4},\, ... $ \\[1.5ex]			
					
					Type 2 & $\displaystyle \frac{1}{\varphi_0} \frac{\lambda}{\tau},\,\quad\frac{\lambda}{\tau+1},\, \frac{1}{\varphi_2}\frac{\lambda}{\tau+2},\, \frac{\lambda}{\tau+3},\,\frac{\lambda}{\tau+4},\, ... $ \\[2.5ex]	
					
					CMP & $\displaystyle \quad\,\,\lambda,\,\quad\, \frac{\lambda}{2^\nu},\quad\;\frac{\lambda}{3^\nu},\, \frac{\lambda}{4^\nu},\,\frac{\lambda}{5^\nu},\, ... $ \\[1.5ex]
					
					Type 1 & $\displaystyle \,\frac{1}{\alpha_0} \lambda,\, \alpha_2 \frac{\lambda}{2^\nu},\,  \frac{1}{\alpha_2} \frac{\lambda}{3^\nu},\, \frac{\lambda}{4^\nu},\,\frac{\lambda}{5^\nu},\, ... $ \\[1.5ex]	
					
					Type 2  & $\displaystyle \,\frac{1}{\varphi_0} \lambda,\,\quad\,\frac{\lambda}{2^\nu},\,  \frac{1}{\varphi_2} \frac{\lambda}{3^\nu},\, \frac{\lambda}{4^\nu},\,\frac{\lambda}{5^\nu},\, ... $ \\[2ex]	
					\bottomrule	
			\end{tabular}}	
		\end{table}
	}

	This article presents two new families of inflation-deflation stationary distributions generated by stationary birth-death processes. The first family of inflation-deflation stationary distributions share the same shape as the mixture model \eqref{eq:mix_zk} but avoids the dominance issue associated with inflation-deflation parameters. The paper contrasts both mechanisms to highlight their distinct properties. Overall, these new models, not previously proposed or used, offer greater flexibility for modeling inflation and deflation in count data.

\section{Inflation-deflation stationary models}
\label{sec:IDS_model}

Let $X(t)$ be a discrete random variable representing the total number of particles in a system at time $t$. $\left\{ X(t): t\geq 0\right\}$ is a birth-death process with $X(0)=0$, birth rates $\gamma_n$ for $n=0,1,2,...$ and death rates $\mu_n$ for $n=1,2,...$. The three possible events that can lead the birth-death process to state $n$ at the current time $t$ (i.e., $X(t)=n$) are: no change from  state $n$, a birth from state $n-1$, or a death from state $n+1$. The pure birth process (i.e., $\mu_n=0$) has only two possible events: that is, no change from $n$, or a birth from $n-1$. The transition probabilities $p_n (t) = P \left\{ X(t)=n | X(0)=0\right\}$ from state $X(0)=0$ to $X(t)=n$ satisfy an infinite system of differential equations known as the kolmogorov forward differential equations. These are given by  
\begin{align} 
	p'_{0}(t)&= \mu_1 p_1(t)-\gamma_0 p_0(t), \nonumber \\ 
	p'_{n}(t)&=\gamma_{n-1} p_{n-1}(t) + \mu_{n+1} p_ {n+1}(t) -(\gamma_{n}+\mu_{n})p_{n}(t),\;\;\;\; n > 0,\label{eq:fde} 
\end{align}
with boundary conditions $p_{0}(0)=1\;\text{and}\; p_{n}(0)=0,\;n>0.$

\subsection{Stationary birth-death processes}
\label{subsec:SBD_process}
The system of differential equations \eqref{eq:fde} is extremely hard to solve in its generality (\citealt[p.~4]{Boswell70}, \citealt[p.~4]{Crawford18}, \citealt[p.~251]{Puig24}), but the asymptotic behavior of the process at $t\to \infty$ is an easier matter to determine. The stationary (limiting) distribution for a birth-death process is given by
\begin{equation}
	p_0 = \left( 1+\sum_{i=0}^{\infty}\lambda_0\lambda_1...\lambda_{i}\right)^{-1},
	\label{eq:p0}
\end{equation}
and
\begin{equation}
	p_n = \lambda_0\lambda_1...\lambda_{n-1}p_0 \:\:\:\:\:\:\:\:\:\:\:,n>0
	\label{eq:pn}
\end{equation}
if and only if the summation
\begin{equation}
	1+\sum_{i=0}^{\infty}\lambda_0\lambda_1...\lambda_{i}< \infty
	\label{eq:pcon}
\end{equation}
is finite, where $p_n= \lim_{t \to \infty} p_n(t)$, and $\lambda_n=\gamma_n/\mu_{n+1}$ (\citealt[p.~176]{Cox65}, \citealt[p.~269]{Grim01}). $\gamma_n$ is the rate at which a birth occurs when the process is in state $n$, and $\mu_{n+1}$ is the rate at which a death occurs when the process is in state $n+1$. Thus, if \eqref{eq:pcon} is divergent, a stationary distribution does not exist. A stationary distribution is time-invariant. Since $\lambda_n=p_{n+1}/p_n$, \eqref{eq:pcon} becomes $1+\sum_{i=0}^{\infty} p_{i+1}/p_0 <\infty$. The geometric, Poisson, PL, NB, HP, and CMP distributions are stationary solutions to system \eqref{eq:fde}. Table \ref{table:table_ratio} shows the birth-death-rate ratios $\lambda_n$ for these models, all of which satisfy condition \eqref{eq:pcon}. Observed data may display higher or lower relative frequencies of certain counts than those predicted by these models. In such cases, we can adjust their ratio sequences to better fit inflation-deflation count data.

\subsection{Type 1 model specification}
\label{subsec:type1_model}
The multiple-inflation mixture model \eqref{eq:mix_zk} with a stationary base distribution, such as geometric, Poisson, PL, NB, HP, and CMP, matches the shape of a modified version of that base distribution. This equivalence is achieved by carefully adjusting the birth-death-rate ratios of the base distribution. However, the details of this relationship may not be straightforward or intuitive. Let $\lambda^{b}_n$ denote the birth-death-rate ratios of a stationary base distribution that satisfy condition \eqref{eq:pcon}. The mixing proportions $ \omega_{n_0},\omega_{n_1},...,$ and $\omega_{n_m}$ in \eqref{eq:mix_zk} increase the probabilities only at $n_0,n_1,...,$ and $n_m$  for the scaled base distribution. For a stationary distribution with shape \eqref{eq:mix_zk}, the birth-death-rate ratios may satisfy 
\begin{equation}
	\lambda_{n}=\left\{\begin{array}{cll}
		\displaystyle \frac{\alpha_{n+1}^{\textbf{1}_{\mathcal{F}}(n+1)}}{\alpha_n^{\textbf{1}_{\mathcal{F}}(n)}}\lambda_n^b & , & n = 0,1,...,q-1 \\
		\displaystyle \frac{1}{\alpha_q^{\textbf{1}_{\mathcal{F}}(q)}}\lambda_q^b & , & n = q \\
		\displaystyle \lambda_n^b & , & n > q, \\
	\end{array}\right.
	\label{eq:ratios}
\end{equation}
where $0<\alpha_n<\infty$, $\mathcal{F}=\left\{ n_0,n_1,...,n_m \right\}\subseteq \left\{ 0,1,...,q \right\} \subset \left\{ 0,1,... \right\}$, and $\textbf{1}_{\mathcal{F}}(n)$ is the indicator function defined as $\textbf{1}_{\mathcal{F}}(n)=1$ if $n \in \mathcal{F}$ and $0$ otherwise.

\begin{proposition}
	\label{same ratio}
	If the series \eqref{eq:pcon} with $\lambda^{b}_n$ converges, then the series \eqref{eq:pcon} with the ratio sequence $\lambda_n$ \eqref{eq:ratios} converges.
\end{proposition}

\Cref{same ratio} is valid because the ratio test confirms that both series have the same limit, $\lim_{n \to \infty } \lambda_n=\lim_{n \to \infty } \lambda^{b}_n<1$. As a result, the equilibrium distribution for the birth-death process with \eqref{eq:ratios} exists. By substituting \eqref{eq:ratios} into \eqref{eq:p0} and \eqref{eq:pn}, the equilibrium birth-death process produces a type 1 inflation-deflation stationary model, 
\begin{equation}
	p(n,\boldsymbol{\theta},\boldsymbol{\alpha})=\frac{f(n,\boldsymbol{\alpha})b(n,\boldsymbol{\theta})}{z(\boldsymbol{\theta},\boldsymbol{\alpha})},
	\label{eq:new_inflat}
\end{equation}
where 
\begin{equation}
	f(n,\boldsymbol{\alpha})=\prod_{i=0}^{m}\alpha_{n_i}^{\textbf{1}_{n_i}(n)},
	\label{eq:f1}
\end{equation}
\begin{equation}
	z(\boldsymbol{\theta},\boldsymbol{\alpha})=1+\sum_{i=0}^{m} \left( \alpha_{n_i}-1 \right)b(n_i,\boldsymbol{\theta}),
	\label{eq:z}
\end{equation}
$\boldsymbol{\alpha}=[\alpha_{n_0},\alpha_{n_1},...,\alpha_{n_m}]$, and $z(\boldsymbol{\theta},\boldsymbol{\alpha})$ is a normalizing constant ensuring $p(n,\boldsymbol{\theta},\boldsymbol{\alpha})$ sums to one. The function $f(n,\boldsymbol{\alpha})$ is a weight function of $m+1$ variables and reflects the inflation-deflation behavior of the type 1 model relative to its base distribution (see Figure \ref{fig:weight_functions}, left part). $\textbf{1}_{n_i}(n)$ is the indicator function, defined as $\textbf{1}_{n_i}(n)=1$ if $n=n_i$ and $0$ otherwise. If $\mathcal{F}=\left\{ 0 \right\}$ and $\alpha_0=e^{\psi}$, then \eqref{eq:new_inflat} simplifies to the distribution of Haslett et al. \eqref{eq:haslett}. The supplementary material provides the derivation of \eqref{eq:new_inflat} and the state-transition-rate ratio diagram for this type 1 process.

When $b(n,\lambda)=(1-\lambda)\lambda^n$ (geometric), where $\lambda$ is a ratio parameter, \eqref{eq:new_inflat} becomes a well-known class of power series distributions \citep[p.~75, eq.(2.1)]{Johnson05}; that is, $p(n,\lambda,\cdot)=f(n,\cdot)\lambda^n/\sum_{i=0}^{\infty} f(i,\cdot)\lambda^i$. We call $\lambda$ the ratio parameter of a stationary birth-death distribution to distinguish it from the Poisson pure-birth distribution. All distributions in Table \ref{table:table_modify}, except the PL distribution, fall into this class. Thus, these power series distributions can be viewed simply as modified geometric distributions.

 \afterpage{
	\begin{table}[H]
		\setlength{\tabcolsep}{0.2em}	
		\small 
		\centering
		\caption{Characteristics of stationary count distributions that modify the geometric and Poisson ratio sequences once. The propability mass function and birth-death-rate ratio sequence are $p(n,\lambda, \cdot)=f(n,\cdot)b(n,\lambda)/z(\lambda,\cdot)$, and $\lambda_n=g(n,\cdot)\lambda_n^b=f(n+1,\cdot)\lambda_n^b/f(n,\cdot)$, respectively.} 
		\label{table:table_modify}	
		\makebox[\textwidth]{\begin{tabular}{ >{\raggedright}b{2cm} >{\centering}p{3.2cm} >{\centering}p{3.2cm} >{\centering}p{3.2cm}	>{\centering\arraybackslash}b{3.2cm}}	
						\toprule
						$b(n,\lambda)$ & \multicolumn{2}{c}{Geometric} & \multicolumn{2}{c}{Poisson} \\
						
						$p(n,\lambda,\cdot)$ & \multicolumn{1}{c}{$f(n,\cdot)$} & \multicolumn{1}{c}{$g(n,\cdot)$} & \multicolumn{1}{c}{$f(n,\cdot)$} & \multicolumn{1}{c}{$g(n,\cdot)$}\\
						\midrule
						Geometric & $\displaystyle  1 $ & $\displaystyle 1 $ & $\displaystyle n!$ & $\displaystyle  n+1$ \\[2ex]
						
						Poisson & $\displaystyle  \frac{1}{n!} $ & $\displaystyle \frac{1}{n+1} $ & $\displaystyle  1 $ & $\displaystyle 1 $\\[2ex]
						
						PL & $\displaystyle 1+(n+1)\lambda$ & $\displaystyle  \frac{1+ (n+2)\lambda}{1+(n+1)\lambda}$ & $\displaystyle (1+(n+1)\lambda) n!$ & $\displaystyle  \frac{(1+(n+2)\lambda)}{(n+1)^{-1}+\lambda}$ \\[2ex]
						
						NB & $\displaystyle  \frac{\left ( r \right )_n}{r^nn!}$ & $\displaystyle \frac{n/r+1}{n+1} $ & $\displaystyle  \frac{\left ( r \right )_n}{r^n}$ & $\displaystyle n/r+1 $ \\[2ex]
						
						HP & $\displaystyle  \frac{1}{\left ( \tau \right )_n}$ & $\displaystyle  \frac{1}{n+\tau}$ & $\displaystyle  \frac{n!}{\left ( \tau \right )_n}$ & $\displaystyle  \frac{n+1}{n+\tau}$  \\[2ex]
						
						CMP & $\displaystyle \frac{1}{(n!)^{\nu}}$ & $\displaystyle \frac{1}{(n+1)^{\nu}}$ & $\displaystyle \frac{1}{(n!)^{\nu-1}}$ & $\displaystyle \frac{1}{(n+1)^{\nu-1}}$ \\[2ex]	
						
						WP & $\displaystyle \frac{(n+\tau)^r}{n!}$ & $\displaystyle \frac{(n+1+\tau)^r}{(n+\tau)^r(n+1)}$ & $\displaystyle (n+\tau)^r$ & $\displaystyle \frac{(n+1+\tau)^r}{(n+\tau)^r}$ \\[2ex]	
						
						Puig et al. & $\displaystyle \frac{e^{-n^2\tau}}{n!}$ & $\displaystyle  \frac{e^{-(2n+1)\tau }}{n+1}$ & $\displaystyle e^{-n^2\tau}$ & $\displaystyle  e^{- (2n+1)\tau}$ \\[2ex]	
						
						Böhning & $\displaystyle \frac{e^{-n^2\tau}}{(n!)^{\nu}}$ & $\displaystyle  \frac{e^{-(2n+1)\tau }}{(n+1)^{\nu}}$ & $\displaystyle \frac{e^{-n^2\tau}}{(n!)^{\nu-1}}$ & $\displaystyle  \frac{e^{-(2n+1)\tau }}{(n+1)^{\nu-1}}$ \\[2ex]	
						\bottomrule		
						\multicolumn{5}{l}{\rule{0pt}{3ex}Note that $\left( r\right)_n=r\left( r+1\right)...\left ( r+n-1 \right )$ is the rising factorial, $\left( 1\right)_n=n!$, and $\left( r\right)_0=1$.}\\
				\end{tabular}}	
			\end{table}
		}

The type 1 stationary model \eqref{eq:new_inflat} does not retain the same structure as the base for successive terms. The probability ratio of two consecutive values is
\begin{equation}
	\lambda_{n}=g(n,\boldsymbol{\alpha})\lambda_{n}^b,
	\label{eq:f_type_I}
\end{equation}
where
\begin{equation}
	g(n,\boldsymbol{\alpha})=\prod_{i=0}^{m} \alpha_{n_i}^{\textbf{1}_{n_i}(n+1)-\textbf{1}_{n_i}(n)},
	\label{eq:g_type_I}
\end{equation}
 $\lambda_n=p(n+1,\boldsymbol{\theta},\boldsymbol{\alpha})/p(n,\boldsymbol{\theta},\boldsymbol{\alpha})$, and  $\lambda_{n}^b=b(n+1,\boldsymbol{\theta})/b(n,\boldsymbol{\theta})$. $g(n,\boldsymbol{\alpha})$ is a weight function of $m+1$ variables, and \eqref{eq:f_type_I} is just another way to express \eqref{eq:ratios}. The relationship simplifies to $\lambda_{n}^b$ when $\alpha_{n_0}=...=\alpha_{n_m}=1$. If $g(n,\alpha)=\alpha^{-\textbf{1}_{0}(n)}$, and $\lambda_{n}^b=\lambda/(n+1)$ (Poisson), then \eqref{eq:f_type_I} reduces to the ratio sequence of \citet[p.~256]{Puig24}. 
 
Because $g(n,\cdot)=f(n+1,\cdot)/f(n,\cdot)$, many common count distributions are stationary birth-death processes. For example, if $f\left ( n,\tau \right )=  n!/ \left ( \tau  \right )_{n}$ and $\lambda_n^b=\lambda /\left ( n+1 \right )$ (Poisson), where $\left( \tau \right)_n=\tau \left(  \tau+1\right)... \left( \tau+n-1 \right)$ is the rising factorial and $\left ( \tau  \right )_{0}=1$, then $\lambda_0=\lambda_0^b /\tau , \lambda_1=2 \lambda_1^b /\left ( \tau+1 \right ), \lambda_2=3 \lambda_2^b /\left ( \tau+2 \right ),...$ forms a birth-death-rate ratio sequence for the HP distribution \citep[p.~134]{Bardwell64}, a single-modified Poisson ratio sequence. If $f(n,\tau,\alpha)=f(n,\tau)f(n,\alpha)$ with $f(n,\alpha)=\alpha^{\textbf{1}_{0}(n)}$, then $\lambda_0=\lambda_0^b /\left ( \tau \alpha \right ) , \lambda_1=2 \lambda_1^b /\left ( \tau+1 \right ), \lambda_2=3 \lambda_2^b /\left ( \tau+2 \right ),...$ is a birth-death-rate ratio sequence for the type 1 HP distribution with $\mathcal{F} = \left\{ 0 \right\}$, representing a double-modified Poisson ratio sequence. The Poisson-Lindley (PL) distribution \citep[p.~145]{Sankaran70} is also a stationary birth-death process (see Section S3 of the supplementary material). Table \ref{table:table_modify} provides examples of single-modified geometric and Poisson ratio sequences.

Suppose $g(n,\tau)=e^{-(2n+1)\tau}$, and $\lambda_{n}^b=\lambda/(n+1)^\nu$ (CMP). Then $\lambda_n=\lambda e^{-(2n+1)\tau }/(n+1)^\nu$. When $\nu=1$, $\lambda_n=\lambda e^{-(2n+1)\tau }/(n+1)$, representing the birth-death-rate ratio of \citet[p.~253, eq.(7)]{Puig24}, a single-modified Poisson ratio sequence. Taking the logarithm of both sides of $\lambda^b_n=(n+1)\lambda_n= \lambda (n+1)^{1-\nu} e^{-(2n+1)\tau}$ yields $\log \lambda_n^b= \log \lambda-\tau +(1- \nu) \log (n+1)-2\tau n$, which defines the birth-death-rate ratio regression model of \citet[p.~212]{Böh16}. The double-modified Poisson distribution (i.e., $f(n,\nu,\tau)=f(n,\nu)f(n,\tau)$, where $f(n,\nu)=(n!)^{1-\nu}$, and $f(n,\tau)=e^{-n^2\tau}$; see the fourth column in Table \ref{table:table_modify}.) is a useful candidate for the type 1 model, as its ratio sequence differs considerably from CMP when $\tau$ is significantly different from $0$. The stationary distributions of Puig et al. and Böhning can also be interpreted as weighted Poisson and CMP distributions, respectively.

\citet[p. 207, eq.(2)]{Böh16} thoroughly describes two special cases of \eqref{eq:f_type_I}, where $g(n)=1$ and $g(n)=1/(n+1)$ (see the third column in Table \ref{table:table_modify}), as power series distributions with $\lambda^b_n=r_x=\lambda$ (i.e, the geometric ratio sequence and a fundamental property of the power series distribution), and $f(n)=a_x$. When $g(n,\lambda,\tau)=\tau n/\lambda+1$ and $\lambda_n^b=\lambda/(n+1)$ (Poisson), \eqref{eq:f_type_I} defines the Katz family of count distributions \citep[p.~82, eq.(2.35)]{Johnson05}. This family modifies the Poisson birth-death-rate ratio sequence and, as a result, excludes the geometric distribution. Overall, the family of stationary birth-death distributions is broad and includes many well-known count distributions.

\afterpage{
	\begin{table}[H]
		\setlength{\tabcolsep}{0.2em}	
		\small 
		\centering
		\caption{Distributions, $p(n,\cdot)=f(n,\cdot)b(n, \cdot)/z(\cdot)$, for inflation-deflation count data generated by stationary birth-death processes using $\lambda_n$ in Table \ref{table:table_ratio}. Their base distributions are the geometric, Poisson, PL, NB, HP, and CMP, which are also stationary distributions. When $\alpha_2=\varphi_2=1$, these stationary distributions are zero-inflated ($1<\alpha_0,\varphi_0<\infty$) and zero-deflated ($0<\alpha_0,\varphi_0<1$).} 
		\label{table:table_prob}	
		\makebox[\textwidth]{\begin{tabular}{ >{\raggedright}b{1.5cm} >{\centering}p{2cm} >{\centering}p{3.5cm}	>{\raggedright\arraybackslash}b{8cm}}	
						\toprule 
						$p(n,\cdot)$ & $f(n,\cdot)$ &   $b(n, \cdot)$ & $z(\cdot)$\\				
						\midrule
						Geometric &  &  & \\[1ex]
						
						Type 1 & $\displaystyle \medmuskip=0mu \thickmuskip=0mu \alpha^{\textbf{1}_0(n)}_0\alpha^{\textbf{1}_2(n)}_2$ & $\displaystyle \medmuskip=0mu \thickmuskip=0mu (1-\lambda)\lambda^n$ & $\displaystyle \medmuskip=0mu \thickmuskip=0mu 1+ (\alpha_0-1)b(0,\lambda)+(\alpha_2-1)b(2,\lambda)$\\
						
						Type 2 & $\displaystyle \medmuskip=0mu \thickmuskip=0mu \varphi^{u_0(n)}_0\varphi^{u_2(n)}_2$ & $\displaystyle \medmuskip=0mu \thickmuskip=0mu (1-\lambda)\lambda^n$ & $\displaystyle \medmuskip=0mu \thickmuskip=0mu 1+(\varphi_0 \varphi_2-1)b(0,\lambda)+ ( \varphi_2-1)\sum_{k=1}^{2}b(k,\lambda)$ \\[2.5ex]
						
						Poisson &  &  & \\
						
						Type 1 & $\displaystyle \medmuskip=0mu \thickmuskip=0mu \alpha^{\textbf{1}_0(n)}_0\alpha^{\textbf{1}_2(n)}_2$ & $\displaystyle \medmuskip=0mu \thickmuskip=0mu \frac{\lambda^n e^{-\lambda}}{n!}$ & $\displaystyle \medmuskip=0mu \thickmuskip=0mu 1+ (\alpha_0-1)b(0,\lambda)+(\alpha_2-1)b(2,\lambda)$\\
						
						Type 2 & $\displaystyle \medmuskip=0mu \thickmuskip=0mu \varphi^{u_0(n)}_0\varphi^{u_2(n)}_2$ & $\displaystyle \medmuskip=0mu \thickmuskip=0mu \frac{\lambda^n e^{-\lambda}}{n!}$ & $\displaystyle \medmuskip=0mu \thickmuskip=0mu 1+(\varphi_0 \varphi_2-1)b(0,\lambda)+ ( \varphi_2-1)\sum_{k=1}^{2}b(k,\lambda)$ \\[2.5ex]
						
						PL &  &  & \\
						
						Type 1 & $\displaystyle \medmuskip=0mu \thickmuskip=0mu \alpha^{\textbf{1}_0(n)}_0\alpha^{\textbf{1}_2(n)}_2$ & $\displaystyle \medmuskip=0mu \thickmuskip=0mu \frac{\left( 1+\lambda+n\lambda \right)\lambda^n}{\left( 1-\lambda \right)^{-2}}$ & $\displaystyle \medmuskip=0mu \thickmuskip=0mu 1+ (\alpha_0-1)b(0,\lambda)+(\alpha_2-1)b(2,\lambda) $\\
						
						Type 2 & $\displaystyle \medmuskip=0mu \thickmuskip=0mu \varphi^{u_0(n)}_0\varphi^{u_2(n)}_2$ & $\displaystyle \medmuskip=0mu \thickmuskip=0mu \frac{\left( 1+\lambda+n\lambda \right)\lambda^n}{\left( 1-\lambda \right)^{-2}}$ & $\displaystyle \medmuskip=0mu \thickmuskip=0mu 1+(\varphi_0 \varphi_2-1)b(0,\lambda)+ ( \varphi_2-1)\sum_{k=1}^{2}b(k,\lambda)$ \\[2.5ex]
						
						NB &  &  & \\

						Type 1 & $\displaystyle \medmuskip=0mu \thickmuskip=0mu \alpha^{\textbf{1}_0(n)}_0\alpha^{\textbf{1}_2(n)}_2$ & $\displaystyle \medmuskip=0mu \thickmuskip=0mu \frac{ \left( r \right)_n \left( \lambda/r \right)^n}{n! \left( 1-\lambda/r \right)^{-r}}$ & $\displaystyle \medmuskip=0mu \thickmuskip=0mu 1+ (\alpha_0-1)b(0,\lambda,r)+(\alpha_2-1)b(2,\lambda,r) $\\
						
						Type 2 & $\displaystyle \medmuskip=0mu \thickmuskip=0mu \varphi^{u_0(n)}_0\varphi^{u_2(n)}_2$ & $\displaystyle \medmuskip=0mu \thickmuskip=0mu \frac{ \left( r \right)_n \left( \lambda/r \right)^n}{n! \left( 1-\lambda/r \right)^{-r}}$ & $\displaystyle \medmuskip=0mu \thickmuskip=0mu 1+(\varphi_0 \varphi_2-1)b(0,\lambda,r)+ ( \varphi_2-1)\sum_{k=1}^{2}b(k,\lambda,r)$ \\[2.5ex]					
						
						HP &  &  & \\
						
						Type 1 & $\displaystyle \medmuskip=0mu \thickmuskip=0mu \alpha^{\textbf{1}_0(n)}_0\alpha^{\textbf{1}_2(n)}_2$ & $\displaystyle \medmuskip=0mu \thickmuskip=0mu \frac{\lambda^n}{(\tau)_n \sum_{i=0}^{\infty }\frac{\lambda^i}{(\tau)_i}}$ & $\displaystyle \medmuskip=0mu \thickmuskip=0mu 1+ (\alpha_0-1)b(0,\lambda,\tau)+(\alpha_2-1)b(2,\lambda,\tau) $\\
						
						Type 2 & $\displaystyle \medmuskip=0mu \thickmuskip=0mu \varphi^{u_0(n)}_0\varphi^{u_2(n)}_2$ & $\displaystyle \medmuskip=0mu \thickmuskip=0mu \frac{\lambda^n}{(\tau)_n \sum_{i=0}^{\infty }\frac{\lambda^i}{(\tau)_i}}$ & $\displaystyle \medmuskip=0mu \thickmuskip=0mu 1+(\varphi_0 \varphi_2-1)b(0,\lambda,\tau)+ ( \varphi_2-1)\sum_{k=1}^{2}b(k,\lambda,\tau)$ \\[2.5ex]		
						
						CMP &  &  & \\
						
						Type 1 & $\displaystyle \medmuskip=0mu \thickmuskip=0mu \alpha^{\textbf{1}_0(n)}_0\alpha^{\textbf{1}_2(n)}_2$ & $\displaystyle \medmuskip=0mu \thickmuskip=0mu \frac{\lambda^n}{(n!)^\nu \sum_{i=0}^{\infty }\frac{\lambda^i}{\left( i! \right)^{\nu}}}$ & $\displaystyle \medmuskip=0mu \thickmuskip=0mu 1+ (\alpha_0-1)b(0,\lambda,\nu)+(\alpha_2-1)b(2,\lambda,\nu) $\\
						
						Type 2 & $\displaystyle \medmuskip=0mu \thickmuskip=0mu \varphi^{u_0(n)}_0\varphi^{u_2(n)}_2$ & $\displaystyle \medmuskip=0mu \thickmuskip=0mu \frac{\lambda^n}{(n!)^\nu \sum_{i=0}^{\infty }\frac{\lambda^i}{\left( i! \right)^{\nu}}}$ & $\displaystyle \medmuskip=0mu \thickmuskip=0mu 1+(\varphi_0 \varphi_2-1)b(0,\lambda,\nu)+ ( \varphi_2-1)\sum_{k=1}^{2}b(k,\lambda,\nu)$ \\[2.5ex]						
						\bottomrule		
						\multicolumn{4}{l}{\rule{0pt}{3ex}Note that $\left( r\right)_n=r\left( r+1\right)...\left ( r+n-1 \right )$ is the rising factorial, $\left( 1\right)_n=n!$, and $\left( r\right)_0=1$.}\\
		\end{tabular}}	
	\end{table}
}
		
The following simple result shows that the type 1 model does not have a dominating issue, unlike the multiple-inflation mixture model, so varying $\alpha_n$ with covariates may be not essential.

\begin{proposition}
	\label{non-domination}
	(Non-domination) Suppose the base distribution $b(n,\boldsymbol{\theta})$ of \eqref{eq:new_inflat} has $\lambda$ as a main parameter, and 
	\begin{equation*}
		\lim_{\lambda  \to \infty} b(n,\boldsymbol{\theta})=0 \:\:\: \text{for all}\:\: n.
	\end{equation*}
	Then
	\begin{equation*}
		\lim_{\lambda  \to \infty}p(n,\boldsymbol{\theta},\boldsymbol{\alpha})=0 \:\:\: \text{for all}\:\: n.
	\end{equation*}
\end{proposition}

\subsection{Equivalence between the mixture and type 1 stationary models}
\label{subsec:mix_stationary}
The equivalence of the hurdle \eqref{eq:hurdle} and zero-inflated \eqref{eq:zip} models in the absence of regressors is well established (\citealt[p.~121]{Winkelmann08}, \citealt[p.~256]{Cameron13}). However, it remains unrecognized that, without regressors, the mixture model \eqref{eq:zip} and its extensions \eqref{eq:mix_zk} are also equivalent to the stationary model \eqref{eq:new_inflat}. This relationship suggests new link functions for the mixture model \eqref{eq:mix_zk}. Furthermore, the stationary mechanism provides an intuitive framework for understanding inflation or deflation, as the ratios $\lambda_n$ are directly interpretable. These findings may guide model selection and facilitate clearer interpretation in applied contexts.
 
To estimate $\omega_0$ and $\omega_q$ in \eqref{eq:mix_zk} without constraints, \citet[p.~195]{Melkersson00}, \citet[p.~1574]{Lin13}, and \citet[p.~1820]{Arora21} reparameterize them as 
\begin{equation}
	\psi_0=\log \left ( \frac{\omega_0}{1-\omega_0-\omega_q} \right )\;\;   \text{and}\;\; \psi_q=\log \left ( \frac{\omega_q}{1-\omega_0-\omega_q} \right ).
	\label{eq:w_link}
\end{equation} 
Applying these link functions yields a parameterization of \eqref{eq:mix_zk}. The resulting model excludes the base distribution in \eqref{eq:mix_zk} because $\omega_0=0$ and $\omega_q=0$ are not possible arguments for the log functions \eqref{eq:w_link}. Therefore, the base parameter estimates cannot serve as starting values for maximum likelihood estimation. These link functions also map only the inflation (positive) area I onto $R^2$, as shown in Figure \ref{fig:w_gamma} (top). The following proposition introduces alternative link functions to resolve these issues.

\begin{proposition}
	\label{equi_shape}
	The model \eqref{eq:new_inflat} for parameters $\alpha_n>0$ and the model \eqref{eq:mix_zk} for parameters $1-\sum_{i=0}^{m}\omega_{n_i}>0$ and $\omega_n+(1-\sum_{i=0}^{m}\omega_{n_i})b(n,\boldsymbol{\theta})>0$ are equivalent in exact shape with
	\begin{equation}
		\omega_n=\left\{\begin{array}{cll}
			\displaystyle \frac{(\alpha_n-1)b(n,\boldsymbol{\theta})}{z(\boldsymbol{\theta},\boldsymbol{\alpha})} & , & n \in \mathcal{F} = \left\{ n_0,...,n_m\right\} \\
			\displaystyle 0 & , & \text{otherwise} \\
		\end{array}\right.\\
	\label{eq:al_to_w}
	\end{equation}
	or
	\begin{equation}
		\alpha_n=\left\{\begin{array}{cll}
			\displaystyle \frac{\omega_n + (1-\sum_{i=0}^{m}\omega_{n_i}) b(n,\boldsymbol{\theta})}{(1-\sum_{i=0}^{m}\omega_{n_i}) b(n,\boldsymbol{\theta})} & , & n \in \mathcal{F} = \left\{ n_0,...,n_m\right\} \\
			\displaystyle 1 & , & \text{otherwise}. \\
		\end{array}\right.\\
	\label{eq:w_to_al}
	\end{equation}
\end{proposition}

\begin{figure}[!t]
	\centering
	\begin{tikzpicture}
		\draw[black,thick] (-7.4, 0) -- (-5.1, 0);
		\draw[black,thick] (-3.6, 0) -- (-2.8, 0)node[below] {\small $\omega_0$};
		\draw[black,thick] (-5.1, -1.4) -- (-5.1, 0);
		\draw[black,thick] (-5.1, 1.5) -- (-5.1,2.2)node[above] {\small $\omega_q$};		
		\draw[thick] (-3.6, 0.1) -- (-3.6, -0.1) node[below] {\small $1$};
		\draw[thick] (0.1-5.1,1.5) -- (-0.1-5.1,1.5) node[left] {\small $1$};
		\node[below left] at (-5.1,0) {\small $0$};
		\draw[dashed,black,thick] (-3.6,0) -- (-5.1,1.5) -- (-5.1,0) -- cycle;
		\filldraw[black,fill=white] (-3.6,0) circle (3pt);
		\filldraw[black,fill=white] (-5.1,1.5) circle (3pt);
		\filldraw[black,fill=white] (-5.1,0) circle (3pt); 
		\draw[black] (0.4-5.1,0.4) node {\small I};

		\draw[black,thick] (-2.3, 0) -- (0, 0);
		\draw[black,thick] (1.5, 0) -- (2.3, 0)node[below] {\small $\alpha_0$};
		\draw[black,thick] (0,-1.4) -- (0, 0);
		\draw[black,thick] (0,1.5) -- (0, 2.2)node[above] {\small $\alpha_q$};
		\draw[thick] (1.5, 0.1) -- (1.5, -0.1) node[below] {\small $1$};
		\draw[thick] (0.1,1.5) -- (-0.1,1.5) node[left] {\small $1$};
		\node[below left] at (0,0) {\small $0$};
		\draw[dashed,black,thick] (0,0) -- (1.5,0) -- (1.5,1.5) -- (0,1.5) -- cycle;
		\filldraw[black,fill=white] (0,0) circle (3pt);
		\filldraw[black,fill=white] (1.5,0) circle (3pt);
		\filldraw[black,fill=white] (1.5,1.5) circle (3pt);
		\filldraw[black,fill=white] (0,1.5) circle (3pt);
		\draw[black] (0.4,0.4) node {\small I};

		\draw[scale=1, thick] (2.8, 0) -- (7.4, 0) node[below] {\small $\psi_0$};
		\draw[scale=1, thick] (5.1,-1.4) -- (5.1,2.2) node[above] {\small $\psi_q$};
		\draw[black] (-0.4+5.1,-0.4) node {\small I};
		\draw[black] ( 0.4+5.1,-0.4) node {\small I};
		\draw[black] ( 0.4+5.1, 0.4) node {\small I};
		\draw[black] (-0.4+5.1, 0.4) node {\small I};
		
		\draw[thick,-{Latex[length=3mm]}] (-4.8,3) arc (110:70:14);
		\node[fill=white] at (0,3.6) {\small $\begin{matrix}
				\psi_0 = \log \left ( \frac{\omega_0}{1-\omega_0-\omega_q} \right ) \\
				\psi_q = \log \left ( \frac{\omega_q}{1-\omega_0-\omega_q} \right )
			\end{matrix}$};
		
		\draw[thick,-{Latex[length=3mm]}] (-4.9,-1.5) arc (230:310:3.5);
		\node[fill=white] at (-3,-2.5) {\small $\begin{matrix}
				\alpha_0 = \frac{\omega_0}{1-\omega_q}  \\
				\alpha_q =  \frac{\omega_q}{1-\omega_0} 
			\end{matrix} $};
		\draw[thick,-{Latex[length=3mm]}] (0.4,-1.5) arc (230:310:3.5);
		\node[fill=white] at (2.5,-2.5) {\small $\begin{matrix}
				\psi_0=\log \left ( \frac{\alpha_0}{1-\alpha_0} \right ) \\
				\psi_q=\log \left ( \frac{\alpha_q}{1-\alpha_q} \right ) 
			\end{matrix} $};
	\end{tikzpicture}	
	\begin{tikzpicture}
		\draw[scale=1, thick] (-7.4, 0) -- (-2.8, 0) node[below] {\small $\omega_0$};
		\draw[scale=1, thick] (-5.1,-2.2) -- (-5.1,2.2) node[above] {\small $\omega_q$};
		\draw[thick] (-6.6, 0.1) -- (-6.6, -0.1) node[below] {\small $-1$};
		\draw[thick] (-3.6, 0.1) -- (-3.6, -0.1) node[below] {\small $1$};
		\draw[thick] (0.1-5.1,-1.5) -- (-0.1-5.1,-1.5) node[left] {\small $-1$};
		\draw[thick] (0.1-5.1,1.5) -- (-0.1-5.1,1.5) node[left] {\small $1$};
		\draw[dashed,black,thick] (-3.6,0) -- (-5.1,1.5) -- (-6.6,-1.5) -- cycle;
		\filldraw[black,fill=white] (-3.6,0) circle (3pt);
		\filldraw[black,fill=white] (-5.1,1.5) circle (3pt);
		\filldraw[black,fill=white] (-6.6,-1.5) circle (3pt); 
		\node[above right] at (-3.6,0) {\small $v_1$};
		\node[above right] at (-5.1,1.5) {\small $v_2$};
		\node[below right] at (-6.6,-1.5) {\small $v_3$}; 
		\draw[black] (0.3-5.1,0.3) node {\small I};
		\draw[black] (-0.3-5.1,0.3) node {\small II};
		\draw[black] (-0.3-5.1,-0.3) node {\small III};
		\draw[black] (0.3-5.1,-0.3) node {\small IV};
		
		\draw[black] (-4,1) node {\small $l_1$};
		\draw[black] (-6.1,0.5) node {\small $l_2$};
		\draw[black] (-4.5,-1) node {\small $l_3$};
		
		
		\draw[black,thick] (-2.3, 0) -- (0, 0);
		\draw[dashed,black,thick] (0, 0) -- (2.3, 0) node[below] {\small $\alpha_0$};
		\draw[black,thick] (0,-2.2) -- (0, 0);
		\draw[dashed,black,thick] (0,0) -- (0,2.2) node[above] {\small $\alpha_q$};
		\draw[black,thick] (0, 1.5) -- (2.3, 1.5);
		\draw[black,thick] (1.5,0) -- (1.5, 2.2);
		\draw[thick] (1.5, 0.1) -- (1.5, -0.1) node[below] {\small $1$};
		\draw[thick] (0.1,1.5) -- (-0.1,1.5) node[left] {\small $1$};
		\filldraw[black,fill=white] (0,0) circle (3pt);
		\node[below left] at (0,0) {\small $0$};
		\draw[black] ( 0.4+1.5, 0.4+1.5) node {\small I};
		\draw[black] (-0.4+1.5, 0.4+1.5) node {\small II};
		\draw[black] (-0.4+1.5,-0.4+1.5) node {\small III};
		\draw[black] ( 0.4+1.5,-0.4+1.5) node {\small IV};

		\draw[scale=1, thick] (2.8, 0) -- (7.4, 0) node[below] {\small $\psi_0$};
		\draw[scale=1, thick] (5.1,-2.2) -- (5.1,2.2) node[above] {\small $\psi_q$};
		\draw[black] ( 0.4+5.1, 0.4) node {\small I};
		\draw[black] (-0.4+5.1, 0.4) node {\small II};
		\draw[black] (-0.4+5.1,-0.4) node {\small III};
		\draw[black] ( 0.4+5.1,-0.4) node {\small IV};
		
		\draw[thick,-{Latex[length=3mm]}] (-4.8,3) arc (110:70:14);
		\node[fill=white] at (0,3.6) {\small $\begin{matrix}
				\psi_0 = \log \left ( \frac{\omega_0+(1-\omega_0-\omega_q)b(0,\boldsymbol{\theta})}{(1-\omega_0-\omega_q)b(0,\boldsymbol{\theta})} \right ) \\
				\psi_q = \log \left ( \frac{\omega_q+(1-\omega_0-\omega_q)b(q,\boldsymbol{\theta})}{(1-\omega_0-\omega_q)b(q,\boldsymbol{\theta})} \right )
			\end{matrix}$};
		
		\draw[thick,-{Latex[length=3mm]}] (-4.9,-2.5) arc (230:310:3.5);
		\node[fill=white] at (-3,-3.4) {\small $\begin{matrix}
				\alpha_0 = \frac{\omega_0+(1-\omega_0-\omega_q)b(0,\boldsymbol{\theta})}{(1-\omega_0-\omega_q)b(0,\boldsymbol{\theta})}  \\
				\alpha_q =  \frac{\omega_q+(1-\omega_0-\omega_q)b(q,\boldsymbol{\theta})}{(1-\omega_0-\omega_q)b(q,\boldsymbol{\theta})} 
			\end{matrix} $};
		\draw[thick,-{Latex[length=3mm]}] (0.4,-2.5) arc (230:310:3.5);
		\node[fill=white] at (2.5,-3.4) {\small $\begin{matrix}
				\psi_0 = \log (\alpha_0) \\
				\psi_q =  \log (\alpha_q) 
			\end{matrix} $};
	\end{tikzpicture}
	\caption{Mapping by the common (top) and new (bottom) link functions of \eqref{eq:mix_zk}. The lines $l_1$, $l_2$, and $l_3$ are $1-\omega_0-\omega_q=0$,  $\omega_0+(1-\omega_0-\omega_q))b(0,\boldsymbol{\theta})=0$, and $\omega_q+(1-\omega_0-\omega_q))b(q,\boldsymbol{\theta})=0$, respectively.}
	\label{fig:w_gamma}
\end{figure}

Guided by \Cref{equi_shape}, we introduce new link functions
\begin{equation}
	\psi_n=\log\left ( \frac{\omega_n + (1-\sum_{i=0}^{m}\omega_{n_i}) b(n,\boldsymbol{\theta})}{(1-\sum_{i=0}^{m}\omega_{n_i}) b(n,\boldsymbol{\theta})} \right), 
	\label{eq:new_link}
\end{equation}
for the inflation-deflation model \eqref{eq:mix_zk}. The proofs of propositions are available in the supplementary material. These functions enable unconstrained parameter estimation and facilitate modeling the relationship between $\boldsymbol{\omega}$ and covariates. The link functions bijectively map the $\boldsymbol{\omega}$-parameter space to the entire $(m+1)$-dimensional space $R^{m+1}$. To illustrate \eqref{eq:new_link}, we set $\mathcal{F}=\left\{ 0,q\right\}$, a special case of \eqref{eq:mix_zk} and \eqref{eq:new_inflat}. In this case, the functions depend on two additional parameters, allowing us to graph the mapping of $(\omega_0,\omega_q)$ onto $(\psi_0,\psi_q)$.

Figure \ref{fig:w_gamma} (bottom) shows how \eqref{eq:w_to_al} maps a triangle onto the first quadrant. The vertices $v_1=\lim_{\alpha_0 \to \infty} (\omega_0,\omega_q)=(1,0)$, $v_2=\lim_{\alpha_q \to \infty} (\omega_0,\omega_q)=(0,1)$, and $v_3=\lim_{\alpha_0 \to 0^+ } \lim_{\alpha_q \to 0^+ } (\omega_0,\omega_q)=(b(0,\lambda),b(q,\lambda))/(b(0,\lambda)+b(q,\lambda)-1)$ represent zero-degeneration, $q$-degeneration, and zero-and-$q$-truncation, respectively. On $l_1$, except at $v_1$ and $v_2$, the distributions $p(n,\boldsymbol{\theta},\omega_0,\omega_q)$ have non-zero probabilities only for $n=0$ and $q$. When points $(\omega_0,\omega_q)$ are on $l_2$ and $l_3$, $p(0,\boldsymbol{\theta},\omega_0,\omega_q)$ and $p(q,\boldsymbol{\theta},\omega_0,\omega_q)$ are always zero, respectively. Therefore, the distribution at $v_3$ is zero-and-$q$ truncated. Figure \ref{fig:w_gamma} (bottom) also highlights four regions of inflation and deflation. The areas I and III show only one type of inflation or deflation, while the areas II and IV allow both types within the same model. Specifically, the areas I and III indicate data inflation and deflation relative to a base distribution. Unlike \eqref{eq:w_link}, \eqref{eq:new_link} maps $(\omega_0,\omega_q)=(0,0)$ onto $(\psi_0,\psi_q)=(0,0)$, so base parameter estimates can be used as starting values for maximum likelihood estimation.

\subsection{Type 2 model specification}
\label{subsec:type2_model}
Apparent inflation and deflation in count data are observed only when compared to a specific base distribution and may differ with other base distributions. Modifying the birth-death-rate ratio sequence of a base distribution results in different inflation-deflation distributions. The type 1 model \eqref{eq:new_inflat}, with the ratio sequence \eqref{eq:ratios}, has the same shape as the mixture model \eqref{eq:mix_zk}, which may not reflect real-world phenomena. To address this, we introduce alternative birth-death-rate ratios
\begin{equation}
	\lambda_{n}=\left\{\begin{array}{cll}
		\displaystyle \frac{1}{\varphi_n^{\textbf{1}_{\mathcal{F}}({n})}}\lambda_n^b & , & n = 0,1,...,q \\
		\displaystyle \lambda_n^b & , & n>q, \\
	\end{array}\right.
	\label{eq:ratiosII}
\end{equation}
where $0<\varphi_n<\infty$. Applying the ratio test shows that the series \eqref{eq:pcon} with \eqref{eq:ratiosII} converges (i.e., $\lim_{n \to \infty } \lambda_n=\lim_{n \to \infty } \lambda^{b}_n<1$), confirming the existence of a solution for this stationary birth-death process.

By substituting \eqref{eq:ratiosII} into \eqref{eq:p0} and \eqref{eq:pn}, the stationary birth-death process generates a type 2 inflation-deflation stationary model
\begin{equation}
	p(n,\boldsymbol{\theta},\boldsymbol{\varphi})=\frac{f(n,\boldsymbol{\varphi}) b(n,\boldsymbol{\theta})}{z(\boldsymbol{\theta},\boldsymbol{\varphi})},
	\label{eq:new_probII}
\end{equation}
where 
\begin{equation}
	f(n,\boldsymbol{\varphi})=\prod_{i=0}^{m}\varphi_{n_i}^{u_{n_i}(n)},
	\label{eq:f2}
\end{equation}
\begin{equation}
	z(\boldsymbol{\theta},\boldsymbol{\varphi})=\displaystyle 1+ \sum_{i=0}^{m} \left( \prod_{j=i}^{m} \varphi_{n_j}-1 \right)\sum_{k=n_{i-1}+1 }^{n_i}b(k,\boldsymbol{\theta}),
	\label{eq:zII}
\end{equation}
$\boldsymbol{\varphi}=[\varphi_{n_0},\varphi_{n_1},...,\varphi_{n_m}]$, $n_{-1}=-1$, and $u_{n_i}(n)$ is the unit step function defined as $u_{n_i}(n)=1$ for $n\le n_i$ and $0$ otherwise. When $\mathcal{F} = \left\{ 0 \right\}$ and $\varphi_0 = e^{\psi}$, \eqref{eq:new_probII} simplifies to the distribution of Haslett et al.\eqref{eq:haslett}, which belongs to both type 1 and 2 families. The weight functions \eqref{eq:f1} and \eqref{eq:f2} demonstrate the inflation-deflation behavior of \eqref{eq:new_inflat} and \eqref{eq:new_probII} relative to their base distributions. For instance, with $\mathcal{F} = \left\{ q \neq0 \right\}$, $f(n,\varphi)=\varphi^{u_{q}(n)}$ displays an inflation-deflation interval (see Figures \ref{fig:wfig7} and \ref{fig:wfig8}), while $f(n,\alpha)=\alpha^{\textbf{1}_{q}(n)}$ displays an inflation-deflation point (see Figures \ref{fig:wfig2} and \ref{fig:wfig3}). Although these two distributions have the same base and number of parameters, their shapes differ. The supplementary material provides the derivation of \eqref{eq:new_probII} and the state-transition-rate-ratio diagram for this type 2 process.

\afterpage{
	\begin{figure}[H]
		\begin{subfigure}{0.5\textwidth}
			\centering
			\begin{tikzpicture}
				\newcommand\xd{0.2} 
				\newcommand\xs{1.3} 
				\newcommand\ys{1} 
				\path[draw=none,use as bounding box] (-1*\xs,-0.5*\ys) rectangle (5*\xs,3*\ys); 
				\draw[ultra thick] (-0.7*\xs, 0) -- (5*\xs, 0);
				
				
				
				\newcommand\xx{0} 
				\node[below] at (\xx*\xs,0) {\xx};
				\draw[thick] (\xx*\xs-\xd,0) -- (\xx*\xs-\xd,2*\ys) node[above] {$\scriptstyle \alpha_0$};
				\filldraw[black] (\xx*\xs-\xd,2*\ys) circle (2pt);
				\draw[thick,dashed] (\xx*\xs+\xd,0) -- (\xx*\xs+\xd,0.5*\ys) node[above] {$\frac{1}{\alpha_0}$};
				\filldraw[fill=white, draw=black] (\xx*\xs+\xd,0.5*\ys) circle (2pt);
				
				\renewcommand\xx{1} 
				\node[below] at (\xx*\xs,0) {\xx};
				\draw[thick] (\xx*\xs-\xd,0) -- (\xx*\xs-\xd,1*\ys) node[above] {$\scriptstyle 1$};
				\filldraw[black] (\xx*\xs-\xd,1*\ys) circle (2pt);
				\draw[thick,dashed] (\xx*\xs+\xd,0) -- (\xx*\xs+\xd,1*\ys) node[above] {$\scriptstyle 1$};
				\filldraw[fill=white, draw=black] (\xx*\xs+\xd,1*\ys) circle (2pt);
				
				\renewcommand\xx{2} 
				\node[below] at (\xx*\xs,0) {\xx};
				\draw[thick] (\xx*\xs-\xd,0) -- (\xx*\xs-\xd,1*\ys) node[above] {$\scriptstyle 1$};
				\filldraw[black] (\xx*\xs-\xd,1*\ys) circle (2pt);
				\draw[thick,dashed] (\xx*\xs+\xd,0) -- (\xx*\xs+\xd,1*\ys) node[above] {$\scriptstyle 1$};
				\filldraw[fill=white, draw=black] (\xx*\xs+\xd,1*\ys) circle (2pt);
				
				\renewcommand\xx{3} 
				\node[below] at (\xx*\xs,0) {\xx};
				\draw[thick] (\xx*\xs-\xd,0) -- (\xx*\xs-\xd,1*\ys) node[above] {$\scriptstyle 1$};
				\filldraw[black] (\xx*\xs-\xd,1*\ys) circle (2pt);
				\draw[thick,dashed] (\xx*\xs+\xd,0) -- (\xx*\xs+\xd,1*\ys) node[above] {$\scriptstyle 1$};
				\filldraw[fill=white, draw=black] (\xx*\xs+\xd,1*\ys) circle (2pt);
				
				\renewcommand\xx{4} 
				\node[below] at (\xx*\xs,0) {\xx};
				\draw[thick] (\xx*\xs-\xd,0) -- (\xx*\xs-\xd,1*\ys) node[above] {$\scriptstyle 1$};
				\filldraw[black] (\xx*\xs-\xd,1*\ys) circle (2pt);
				\draw[thick,dashed] (\xx*\xs+\xd,0) -- (\xx*\xs+\xd,1*\ys) node[above] {$\scriptstyle 1$};
				\filldraw[fill=white, draw=black] (\xx*\xs+\xd,1*\ys) circle (2pt);
				
				\node[] at (4.7*\xs,0.5*\ys) {$\boldsymbol{\dots}$};
			\end{tikzpicture}
			\caption{$\alpha_0=2$}
			\label{fig:wfig1}
			\begin{tikzpicture}
				\newcommand\xd{0.2} 
				\newcommand\xs{1.3} 
				\newcommand\ys{1} 
				\path[draw=none,use as bounding box] (-1*\xs,-0.5*\ys) rectangle (5*\xs,3*\ys); 
				\draw[ultra thick] (-0.7*\xs, 0) -- (5*\xs, 0);

				\newcommand\xx{0} 
				\node[below] at (\xx*\xs,0) {\xx};
				\draw[thick] (\xx*\xs-\xd,0) -- (\xx*\xs-\xd,1*\ys) node[above] {$\scriptstyle 1$};
				\filldraw[black] (\xx*\xs-\xd,1*\ys) circle (2pt);
				\draw[thick,dashed] (\xx*\xs+\xd,0) -- (\xx*\xs+\xd,2*\ys) node[above] {$\scriptstyle \alpha_1$};
				\filldraw[fill=white, draw=black] (\xx*\xs+\xd,2*\ys) circle (2pt);
				
				\renewcommand\xx{1} 
				\node[below] at (\xx*\xs,0) {\xx};
				\draw[thick] (\xx*\xs-\xd,0) -- (\xx*\xs-\xd,2*\ys) node[above] {$\scriptstyle \alpha_1$};
				\filldraw[black] (\xx*\xs-\xd,2*\ys) circle (2pt);
				\draw[thick,dashed] (\xx*\xs+\xd,0) -- (\xx*\xs+\xd,0.5*\ys) node[above] {$\frac{1}{\alpha_1}$};
				\filldraw[fill=white, draw=black] (\xx*\xs+\xd,0.5*\ys) circle (2pt);

				\renewcommand\xx{2} 
				\node[below] at (\xx*\xs,0) {\xx};
				\draw[thick] (\xx*\xs-\xd,0) -- (\xx*\xs-\xd,1*\ys) node[above] {$\scriptstyle 1$};
				\filldraw[black] (\xx*\xs-\xd,1*\ys) circle (2pt);
				\draw[thick,dashed] (\xx*\xs+\xd,0) -- (\xx*\xs+\xd,1*\ys) node[above] {$\scriptstyle 1$};
				\filldraw[fill=white, draw=black] (\xx*\xs+\xd,1*\ys) circle (2pt);
				
				\renewcommand\xx{3} 
				\node[below] at (\xx*\xs,0) {\xx};
				\draw[thick] (\xx*\xs-\xd,0) -- (\xx*\xs-\xd,1*\ys) node[above] {$\scriptstyle 1$};
				\filldraw[black] (\xx*\xs-\xd,1*\ys) circle (2pt);
				\draw[thick,dashed] (\xx*\xs+\xd,0) -- (\xx*\xs+\xd,1*\ys) node[above] {$\scriptstyle 1$};
				\filldraw[fill=white, draw=black] (\xx*\xs+\xd,1*\ys) circle (2pt);
				
				\renewcommand\xx{4} 
				\node[below] at (\xx*\xs,0) {\xx};
				\draw[thick] (\xx*\xs-\xd,0) -- (\xx*\xs-\xd,1*\ys) node[above] {$\scriptstyle 1$};
				\filldraw[black] (\xx*\xs-\xd,1*\ys) circle (2pt);
				\draw[thick,dashed] (\xx*\xs+\xd,0) -- (\xx*\xs+\xd,1*\ys) node[above] {$\scriptstyle 1$};
				\filldraw[fill=white, draw=black] (\xx*\xs+\xd,1*\ys) circle (2pt);
				
				\node[] at (4.7*\xs,0.5*\ys) {$\boldsymbol{\dots}$};
			\end{tikzpicture}
			\caption{$\alpha_1=2$}
			\label{fig:wfig2}
			\begin{tikzpicture}
				\newcommand\xd{0.2} 
				\newcommand\xs{1.3} 
				\newcommand\ys{1} 
				\path[draw=none,use as bounding box] (-1*\xs,-0.5*\ys) rectangle (5*\xs,3*\ys); 
				\draw[ultra thick] (-0.7*\xs, 0) -- (5*\xs, 0);

				\newcommand\xx{0} 
				\node[below] at (\xx*\xs,0) {\xx};
				\draw[thick] (\xx*\xs-\xd,0) -- (\xx*\xs-\xd,1*\ys) node[above] {$\scriptstyle 1$};
				\filldraw[black] (\xx*\xs-\xd,1*\ys) circle (2pt);
				\draw[thick,dashed] (\xx*\xs+\xd,0) -- (\xx*\xs+\xd,1*\ys) node[above] {$\scriptstyle 1$};
				\filldraw[fill=white, draw=black] (\xx*\xs+\xd,1*\ys) circle (2pt);
				
				\renewcommand\xx{1} 
				\node[below] at (\xx*\xs,0) {\xx};
				\draw[thick] (\xx*\xs-\xd,0) -- (\xx*\xs-\xd,1*\ys) node[above] {$\scriptstyle 1$};
				\filldraw[black] (\xx*\xs-\xd,1*\ys) circle (2pt);
				\draw[thick,dashed] (\xx*\xs+\xd,0) -- (\xx*\xs+\xd,2*\ys) node[above] {$\scriptstyle \alpha_2$};
				\filldraw[fill=white, draw=black] (\xx*\xs+\xd,2*\ys) circle (2pt);

				\renewcommand\xx{2} 
				\node[below] at (\xx*\xs,0) {\xx};
				\draw[thick] (\xx*\xs-\xd,0) -- (\xx*\xs-\xd,2*\ys) node[above] {$\scriptstyle \alpha_2$};
				\filldraw[black] (\xx*\xs-\xd,2*\ys) circle (2pt);
				\draw[thick,dashed] (\xx*\xs+\xd,0) -- (\xx*\xs+\xd,0.5*\ys) node[above] {$\frac{1}{\alpha_2}$};
				\filldraw[fill=white, draw=black] (\xx*\xs+\xd,0.5*\ys) circle (2pt);
				
				\renewcommand\xx{3} 
				\node[below] at (\xx*\xs,0) {\xx};
				\draw[thick] (\xx*\xs-\xd,0) -- (\xx*\xs-\xd,1*\ys) node[above] {$\scriptstyle 1$};
				\filldraw[black] (\xx*\xs-\xd,1*\ys) circle (2pt);
				\draw[thick,dashed] (\xx*\xs+\xd,0) -- (\xx*\xs+\xd,1*\ys) node[above] {$\scriptstyle 1$};
				\filldraw[fill=white, draw=black] (\xx*\xs+\xd,1*\ys) circle (2pt);
				
				\renewcommand\xx{4} 
				\node[below] at (\xx*\xs,0) {\xx};
				\draw[thick] (\xx*\xs-\xd,0) -- (\xx*\xs-\xd,1*\ys) node[above] {$\scriptstyle 1$};
				\filldraw[black] (\xx*\xs-\xd,1*\ys) circle (2pt);
				\draw[thick,dashed] (\xx*\xs+\xd,0) -- (\xx*\xs+\xd,1*\ys) node[above] {$\scriptstyle 1$};
				\filldraw[fill=white, draw=black] (\xx*\xs+\xd,1*\ys) circle (2pt);
				
				\node[] at (4.7*\xs,0.5*\ys) {$\boldsymbol{\dots}$};
			\end{tikzpicture}
			\caption{$\alpha_2=2$}
			\label{fig:wfig3}
			\begin{tikzpicture}
				\newcommand\xd{0.2} 
				\newcommand\xs{1.3} 
				\newcommand\ys{1} 
				\path[draw=none,use as bounding box] (-1*\xs,-0.5*\ys) rectangle (5*\xs,3*\ys); 
				\draw[ultra thick] (-0.7*\xs, 0) -- (5*\xs, 0);
				
				\newcommand\xx{0} 
				\node[below] at (\xx*\xs,0) {\xx};
				\draw[thick] (\xx*\xs-\xd,0) -- (\xx*\xs-\xd,2*\ys) node[above] {$\scriptstyle \alpha_0$};
				\filldraw[black] (\xx*\xs-\xd,2*\ys) circle (2pt);
				\draw[thick,dashed] (\xx*\xs+\xd,0) -- (\xx*\xs+\xd,0.5*\ys) node[above] {$\frac{1}{\alpha_0}$};
				\filldraw[fill=white, draw=black] (\xx*\xs+\xd,0.5*\ys) circle (2pt);
				
				\renewcommand\xx{1} 
				\node[below] at (\xx*\xs,0) {\xx};
				\draw[thick] (\xx*\xs-\xd,0) -- (\xx*\xs-\xd,1*\ys) node[above] {$\scriptstyle 1$};
				\filldraw[black] (\xx*\xs-\xd,1*\ys) circle (2pt);
				\draw[thick,dashed] (\xx*\xs+\xd,0) -- (\xx*\xs+\xd,1.5*\ys) node[above] {$\scriptstyle \alpha_2$}; 
				\filldraw[fill=white, draw=black] (\xx*\xs+\xd,1.5*\ys) circle (2pt);
				
				\renewcommand\xx{2} 
				\node[below] at (\xx*\xs,0) {\xx};
				\draw[thick] (\xx*\xs-\xd,0) -- (\xx*\xs-\xd,1.5*\ys) node[above] {$\scriptstyle \alpha_2$};
				\filldraw[black] (\xx*\xs-\xd,1.5*\ys) circle (2pt);
				\draw[thick,dashed] (\xx*\xs+\xd,0) -- (\xx*\xs+\xd,0.67*\ys) node[above] {$\frac{1}{\alpha_2}$};
				\filldraw[fill=white, draw=black] (\xx*\xs+\xd,0.67*\ys) circle (2pt);
				
				\renewcommand\xx{3} 
				\node[below] at (\xx*\xs,0) {\xx};
				\draw[thick] (\xx*\xs-\xd,0) -- (\xx*\xs-\xd,1*\ys) node[above] {$\scriptstyle 1$};
				\filldraw[black] (\xx*\xs-\xd,1*\ys) circle (2pt);
				\draw[thick,dashed] (\xx*\xs+\xd,0) -- (\xx*\xs+\xd,1*\ys) node[above] {$\scriptstyle 1$};
				\filldraw[fill=white, draw=black] (\xx*\xs+\xd,1*\ys) circle (2pt);
				
				\renewcommand\xx{4} 
				\node[below] at (\xx*\xs,0) {\xx};
				\draw[thick] (\xx*\xs-\xd,0) -- (\xx*\xs-\xd,1*\ys) node[above] {$\scriptstyle 1$};
				\filldraw[black] (\xx*\xs-\xd,1*\ys) circle (2pt);
				\draw[thick,dashed] (\xx*\xs+\xd,0) -- (\xx*\xs+\xd,1*\ys) node[above] {$\scriptstyle 1$};
				\filldraw[fill=white, draw=black] (\xx*\xs+\xd,1*\ys) circle (2pt);
				
				\node[] at (4.7*\xs,0.5*\ys) {$\boldsymbol{\dots}$};
			\end{tikzpicture}
			\caption{$\alpha_0=2, \alpha_2=1.5$}
			\label{fig:wfig4}
			\begin{tikzpicture}
				\newcommand\xd{0.2} 
				\newcommand\xs{1.3} 
				\newcommand\ys{1} 
				\path[draw=none,use as bounding box] (-1*\xs,-0.5*\ys) rectangle (5*\xs,3*\ys); 
				\draw[ultra thick] (-0.7*\xs, 0) -- (5*\xs, 0);
				
				\newcommand\xx{0} 
				\node[below] at (\xx*\xs,0) {\xx};
				\draw[thick] (\xx*\xs-\xd,0) -- (\xx*\xs-\xd,2*\ys) node[above] {$\scriptstyle \alpha_0$};
				\filldraw[black] (\xx*\xs-\xd,2*\ys) circle (2pt);
				\draw[thick,dashed] (\xx*\xs+\xd,0) -- (\xx*\xs+\xd,0.3*\ys) node[above] {$\frac{\alpha_1}{\alpha_0}$};
				\filldraw[fill=white, draw=black] (\xx*\xs+\xd,0.3*\ys) circle (2pt);
				
				\renewcommand\xx{1} 
				\node[below] at (\xx*\xs,0) {\xx};
				\draw[thick] (\xx*\xs-\xd,0) -- (\xx*\xs-\xd,0.6*\ys) node[above] {$\scriptstyle \alpha_1$};
				\filldraw[black] (\xx*\xs-\xd,0.6*\ys) circle (2pt);
				\draw[thick,dashed] (\xx*\xs+\xd,0) -- (\xx*\xs+\xd,2*\ys) node[above] {$\frac{\alpha_2}{\alpha_1}$};
				\filldraw[fill=white, draw=black] (\xx*\xs+\xd,2*\ys) circle (2pt);
				
				\renewcommand\xx{2} 
				\node[below] at (\xx*\xs,0) {\xx};
				\draw[thick] (\xx*\xs-\xd,0) -- (\xx*\xs-\xd,1.2*\ys) node[above] {$\scriptstyle \alpha_2$};
				\filldraw[black] (\xx*\xs-\xd,1.2*\ys) circle (2pt);
				\draw[thick,dashed] (\xx*\xs+\xd,0) -- (\xx*\xs+\xd,0.83*\ys) node[above] {$\frac{1}{\alpha_2}$};
				\filldraw[fill=white, draw=black] (\xx*\xs+\xd,0.83*\ys) circle (2pt);
				
				\renewcommand\xx{3} 
				\node[below] at (\xx*\xs,0) {\xx};
				\draw[thick] (\xx*\xs-\xd,0) -- (\xx*\xs-\xd,1*\ys) node[above] {$\scriptstyle 1$};
				\filldraw[black] (\xx*\xs-\xd,1*\ys) circle (2pt);
				\draw[thick,dashed] (\xx*\xs+\xd,0) -- (\xx*\xs+\xd,1*\ys) node[above] {$\scriptstyle 1$};
				\filldraw[fill=white, draw=black] (\xx*\xs+\xd,1*\ys) circle (2pt);
				
				\renewcommand\xx{4} 
				\node[below] at (\xx*\xs,0) {\xx};
				\draw[thick] (\xx*\xs-\xd,0) -- (\xx*\xs-\xd,1*\ys) node[above] {$\scriptstyle 1$};
				\filldraw[black] (\xx*\xs-\xd,1*\ys) circle (2pt);
				\draw[thick,dashed] (\xx*\xs+\xd,0) -- (\xx*\xs+\xd,1*\ys) node[above] {$\scriptstyle 1$};
				\filldraw[fill=white, draw=black] (\xx*\xs+\xd,1*\ys) circle (2pt);
				
				\node[] at (4.7*\xs,0.5*\ys) {$\boldsymbol{\dots}$};
			\end{tikzpicture}
			\caption{$\alpha_0=2, \alpha_1=0.6, \alpha_2=1.2$}
			\label{fig:wfig5}
		\end{subfigure}
		\begin{subfigure}{0.5\textwidth}
			\centering
			\begin{tikzpicture}
				\newcommand\xd{0.2} 
				\newcommand\xs{1.3} 
				\newcommand\ys{1} 
				\path[draw=none,use as bounding box] (-1*\xs,-0.5*\ys) rectangle (5*\xs,3*\ys); 
				\draw[ultra thick] (-0.7*\xs, 0) -- (5*\xs, 0);
				
				
				\newcommand\xx{0} 
				\node[below] at (\xx*\xs,0) {\xx};
				\draw[thick] (\xx*\xs-\xd,0) -- (\xx*\xs-\xd,2*\ys) node[above] {$\scriptstyle \varphi_0$};
				\filldraw[black] (\xx*\xs-\xd,2*\ys) circle (2pt);
				\draw[thick,dashed] (\xx*\xs+\xd,0) -- (\xx*\xs+\xd,0.5*\ys) node[above] {$\frac{1}{\varphi_0}$};
				\filldraw[fill=white, draw=black] (\xx*\xs+\xd,0.5*\ys) circle (2pt);
				
				\renewcommand\xx{1} 
				\node[below] at (\xx*\xs,0) {\xx};
				\draw[thick] (\xx*\xs-\xd,0) -- (\xx*\xs-\xd,1*\ys) node[above] {$\scriptstyle 1$};
				\filldraw[black] (\xx*\xs-\xd,1*\ys) circle (2pt);
				\draw[thick,dashed] (\xx*\xs+\xd,0) -- (\xx*\xs+\xd,1*\ys) node[above] {$\scriptstyle 1$};
				\filldraw[fill=white, draw=black] (\xx*\xs+\xd,1*\ys) circle (2pt);
				
				\renewcommand\xx{2} 
				\node[below] at (\xx*\xs,0) {\xx};
				\draw[thick] (\xx*\xs-\xd,0) -- (\xx*\xs-\xd,1*\ys) node[above] {$\scriptstyle 1$};
				\filldraw[black] (\xx*\xs-\xd,1*\ys) circle (2pt);
				\draw[thick,dashed] (\xx*\xs+\xd,0) -- (\xx*\xs+\xd,1*\ys) node[above] {$\scriptstyle 1$};
				\filldraw[fill=white, draw=black] (\xx*\xs+\xd,1*\ys) circle (2pt);
				
				\renewcommand\xx{3} 
				\node[below] at (\xx*\xs,0) {\xx};
				\draw[thick] (\xx*\xs-\xd,0) -- (\xx*\xs-\xd,1*\ys) node[above] {$\scriptstyle 1$};
				\filldraw[black] (\xx*\xs-\xd,1*\ys) circle (2pt);
				\draw[thick,dashed] (\xx*\xs+\xd,0) -- (\xx*\xs+\xd,1*\ys) node[above] {$\scriptstyle 1$};
				\filldraw[fill=white, draw=black] (\xx*\xs+\xd,1*\ys) circle (2pt);
				
				\renewcommand\xx{4} 
				\node[below] at (\xx*\xs,0) {\xx};
				\draw[thick] (\xx*\xs-\xd,0) -- (\xx*\xs-\xd,1*\ys) node[above] {$\scriptstyle 1$};
				\filldraw[black] (\xx*\xs-\xd,1*\ys) circle (2pt);
				\draw[thick,dashed] (\xx*\xs+\xd,0) -- (\xx*\xs+\xd,1*\ys) node[above] {$\scriptstyle 1$};
				\filldraw[fill=white, draw=black] (\xx*\xs+\xd,1*\ys) circle (2pt);
				
				\node[] at (4.7*\xs,0.5*\ys) {$\boldsymbol{\dots}$};
			\end{tikzpicture}
			\caption{$\varphi_0=2$}
			\label{fig:wfig6}
			\begin{tikzpicture}
				\newcommand\xd{0.2} 
				\newcommand\xs{1.3} 
				\newcommand\ys{1} 
				\path[draw=none,use as bounding box] (-1*\xs,-0.5*\ys) rectangle (5*\xs,3*\ys); 
				\draw[ultra thick] (-0.7*\xs, 0) -- (5*\xs, 0);
				
				\newcommand\xx{0} 
				\node[below] at (\xx*\xs,0) {\xx};
				\draw[thick] (\xx*\xs-\xd,0) -- (\xx*\xs-\xd,2*\ys) node[above] {$\scriptstyle \varphi_1$};
				\filldraw[black] (\xx*\xs-\xd,2*\ys) circle (2pt);
				\draw[thick,dashed] (\xx*\xs+\xd,0) -- (\xx*\xs+\xd,1*\ys) node[above] {$\scriptstyle 1$};
				\filldraw[fill=white, draw=black] (\xx*\xs+\xd,1*\ys) circle (2pt);
				
				\renewcommand\xx{1} 
				\node[below] at (\xx*\xs,0) {\xx};
				\draw[thick] (\xx*\xs-\xd,0) -- (\xx*\xs-\xd,2*\ys) node[above] {$\scriptstyle \varphi_1$};
				\filldraw[black] (\xx*\xs-\xd,2*\ys) circle (2pt);
				\draw[thick,dashed] (\xx*\xs+\xd,0) -- (\xx*\xs+\xd,0.5*\ys) node[above] {$\frac{1}{\varphi_1}$};
				\filldraw[fill=white, draw=black] (\xx*\xs+\xd,0.5*\ys) circle (2pt);
				
				\renewcommand\xx{2} 
				\node[below] at (\xx*\xs,0) {\xx};
				\draw[thick] (\xx*\xs-\xd,0) -- (\xx*\xs-\xd,1*\ys) node[above] {$\scriptstyle 1$};
				\filldraw[black] (\xx*\xs-\xd,1*\ys) circle (2pt);
				\draw[thick,dashed] (\xx*\xs+\xd,0) -- (\xx*\xs+\xd,1*\ys) node[above] {$\scriptstyle 1$};
				\filldraw[fill=white, draw=black] (\xx*\xs+\xd,1*\ys) circle (2pt);
				
				\renewcommand\xx{3} 
				\node[below] at (\xx*\xs,0) {\xx};
				\draw[thick] (\xx*\xs-\xd,0) -- (\xx*\xs-\xd,1*\ys) node[above] {$\scriptstyle 1$};
				\filldraw[black] (\xx*\xs-\xd,1*\ys) circle (2pt);
				\draw[thick,dashed] (\xx*\xs+\xd,0) -- (\xx*\xs+\xd,1*\ys) node[above] {$\scriptstyle 1$};
				\filldraw[fill=white, draw=black] (\xx*\xs+\xd,1*\ys) circle (2pt);
				
				\renewcommand\xx{4} 
				\node[below] at (\xx*\xs,0) {\xx};
				\draw[thick] (\xx*\xs-\xd,0) -- (\xx*\xs-\xd,1*\ys) node[above] {$\scriptstyle 1$};
				\filldraw[black] (\xx*\xs-\xd,1*\ys) circle (2pt);
				\draw[thick,dashed] (\xx*\xs+\xd,0) -- (\xx*\xs+\xd,1*\ys) node[above] {$\scriptstyle 1$};
				\filldraw[fill=white, draw=black] (\xx*\xs+\xd,1*\ys) circle (2pt);
				
				\node[] at (4.7*\xs,0.5*\ys) {$\boldsymbol{\dots}$};
			\end{tikzpicture}
			\caption{$\varphi_1=2$}
			\label{fig:wfig7}
			\begin{tikzpicture}
				\newcommand\xd{0.2} 
				\newcommand\xs{1.3} 
				\newcommand\ys{1} 
				\path[draw=none,use as bounding box] (-1*\xs,-0.5*\ys) rectangle (5*\xs,3*\ys); 
				\draw[ultra thick] (-0.7*\xs, 0) -- (5*\xs, 0);
				
				\newcommand\xx{0} 
				\node[below] at (\xx*\xs,0) {\xx};
				\draw[thick] (\xx*\xs-\xd,0) -- (\xx*\xs-\xd,2*\ys) node[above] {$\scriptstyle \varphi_2$};
				\filldraw[black] (\xx*\xs-\xd,2*\ys) circle (2pt);
				\draw[thick,dashed] (\xx*\xs+\xd,0) -- (\xx*\xs+\xd,1*\ys) node[above] {$\scriptstyle 1$};
				\filldraw[fill=white, draw=black] (\xx*\xs+\xd,1*\ys) circle (2pt);
				
				\renewcommand\xx{1} 
				\node[below] at (\xx*\xs,0) {\xx};
				\draw[thick] (\xx*\xs-\xd,0) -- (\xx*\xs-\xd,2*\ys) node[above] {$\scriptstyle \varphi_2$};
				\filldraw[black] (\xx*\xs-\xd,2*\ys) circle (2pt);
				\draw[thick,dashed] (\xx*\xs+\xd,0) -- (\xx*\xs+\xd,1*\ys) node[above] {$\scriptstyle 1$};
				\filldraw[fill=white, draw=black] (\xx*\xs+\xd,1*\ys) circle (2pt);
				
				\renewcommand\xx{2} 
				\node[below] at (\xx*\xs,0) {\xx};
				\draw[thick] (\xx*\xs-\xd,0) -- (\xx*\xs-\xd,2*\ys) node[above] {$\scriptstyle \varphi_2$};
				\filldraw[black] (\xx*\xs-\xd,2*\ys) circle (2pt);
				\draw[thick,dashed] (\xx*\xs+\xd,0) -- (\xx*\xs+\xd,0.5*\ys) node[above] {$\frac{1}{\varphi_2}$};
				\filldraw[fill=white, draw=black] (\xx*\xs+\xd,0.5*\ys) circle (2pt);
				
				\renewcommand\xx{3} 
				\node[below] at (\xx*\xs,0) {\xx};
				\draw[thick] (\xx*\xs-\xd,0) -- (\xx*\xs-\xd,1*\ys) node[above] {$\scriptstyle 1$};
				\filldraw[black] (\xx*\xs-\xd,1*\ys) circle (2pt);
				\draw[thick,dashed] (\xx*\xs+\xd,0) -- (\xx*\xs+\xd,1*\ys) node[above] {$\scriptstyle 1$};
				\filldraw[fill=white, draw=black] (\xx*\xs+\xd,1*\ys) circle (2pt);
				
				\renewcommand\xx{4} 
				\node[below] at (\xx*\xs,0) {\xx};
				\draw[thick] (\xx*\xs-\xd,0) -- (\xx*\xs-\xd,1*\ys) node[above] {$\scriptstyle 1$};
				\filldraw[black] (\xx*\xs-\xd,1*\ys) circle (2pt);
				\draw[thick,dashed] (\xx*\xs+\xd,0) -- (\xx*\xs+\xd,1*\ys) node[above] {$\scriptstyle 1$};
				\filldraw[fill=white, draw=black] (\xx*\xs+\xd,1*\ys) circle (2pt);
				
				\node[] at (4.7*\xs,0.5*\ys) {$\boldsymbol{\dots}$};
			\end{tikzpicture}
			\caption{$\varphi_2=2$}
			\label{fig:wfig8}
			\begin{tikzpicture}
				\newcommand\xd{0.2} 
				\newcommand\xs{1.3} 
				\newcommand\ys{1} 
				\path[draw=none,use as bounding box] (-1*\xs,-0.5*\ys) rectangle (5*\xs,3*\ys); 
				\draw[ultra thick] (-0.7*\xs, 0) -- (5*\xs, 0);
				
				\newcommand\xx{0} 
				\node[below] at (\xx*\xs,0) {\xx};
				\draw[thick] (\xx*\xs-\xd,0) -- (\xx*\xs-\xd,2*\ys) node[above] {$\scriptstyle \varphi_0 \varphi_2$};
				\filldraw[black] (\xx*\xs-\xd,2*\ys) circle (2pt);
				\draw[thick,dashed] (\xx*\xs+\xd,0) -- (\xx*\xs+\xd,0.75*\ys) node[above] {$\frac{1}{\varphi_0}$};
				\filldraw[fill=white, draw=black] (\xx*\xs+\xd,0.75*\ys) circle (2pt);
				
				\renewcommand\xx{1} 
				\node[below] at (\xx*\xs,0) {\xx};
				\draw[thick] (\xx*\xs-\xd,0) -- (\xx*\xs-\xd,1.5*\ys) node[above] {$\scriptstyle \varphi_2$};
				\filldraw[black] (\xx*\xs-\xd,1.5*\ys) circle (2pt);
				\draw[thick,dashed] (\xx*\xs+\xd,0) -- (\xx*\xs+\xd,1*\ys) node[above] {$\scriptstyle 1$};
				\filldraw[fill=white, draw=black] (\xx*\xs+\xd,1*\ys) circle (2pt);
				
				\renewcommand\xx{2} 
				\node[below] at (\xx*\xs,0) {\xx};
				\draw[thick] (\xx*\xs-\xd,0) -- (\xx*\xs-\xd,1.5*\ys) node[above] {$\scriptstyle \varphi_2$};
				\filldraw[black] (\xx*\xs-\xd,1.5*\ys) circle (2pt);
				\draw[thick,dashed] (\xx*\xs+\xd,0) -- (\xx*\xs+\xd,0.67*\ys) node[above] {$\frac{1}{\varphi_2}$};
				\filldraw[fill=white, draw=black] (\xx*\xs+\xd,0.67*\ys) circle (2pt);
				
				\renewcommand\xx{3} 
				\node[below] at (\xx*\xs,0) {\xx};
				\draw[thick] (\xx*\xs-\xd,0) -- (\xx*\xs-\xd,1*\ys) node[above] {$\scriptstyle 1$};
				\filldraw[black] (\xx*\xs-\xd,1*\ys) circle (2pt);
				\draw[thick,dashed] (\xx*\xs+\xd,0) -- (\xx*\xs+\xd,1*\ys) node[above] {$\scriptstyle 1$};
				\filldraw[fill=white, draw=black] (\xx*\xs+\xd,1*\ys) circle (2pt);
				
				\renewcommand\xx{4} 
				\node[below] at (\xx*\xs,0) {\xx};
				\draw[thick] (\xx*\xs-\xd,0) -- (\xx*\xs-\xd,1*\ys) node[above] {$\scriptstyle 1$};
				\filldraw[black] (\xx*\xs-\xd,1*\ys) circle (2pt);
				\draw[thick,dashed] (\xx*\xs+\xd,0) -- (\xx*\xs+\xd,1*\ys) node[above] {$\scriptstyle 1$};
				\filldraw[fill=white, draw=black] (\xx*\xs+\xd,1*\ys) circle (2pt);
				
				\node[] at (4.7*\xs,0.5*\ys) {$\boldsymbol{\dots}$};
			\end{tikzpicture}
			\caption{$\varphi_0=1.3, \varphi_2=1.5$}
			\label{fig:wfig9}
			\begin{tikzpicture}
				\newcommand\xd{0.2} 
				\newcommand\xs{1.3} 
				\newcommand\ys{1} 
				\path[draw=none,use as bounding box] (-1*\xs,-0.5*\ys) rectangle (5*\xs,3*\ys); 
				\draw[ultra thick] (-0.7*\xs, 0) -- (5*\xs, 0);
				
				\newcommand\xx{0} 
				\node[below] at (\xx*\xs,0) {\xx};
				\draw[thick] (\xx*\xs-\xd,0) -- (\xx*\xs-\xd,2*\ys) node[above] {$\scriptstyle \varphi_0 \varphi_1 \varphi_2$};
				\filldraw[black] (\xx*\xs-\xd,2*\ys) circle (2pt);
				\draw[thick,dashed] (\xx*\xs+\xd,0) -- (\xx*\xs+\xd,0.3*\ys) node[above] {$\frac{1}{\varphi_0}$};
				\filldraw[fill=white, draw=black] (\xx*\xs+\xd,0.3*\ys) circle (2pt);
				
				\renewcommand\xx{1} 
				\node[below] at (\xx*\xs,0) {\xx};
				\draw[thick] (\xx*\xs-\xd,0) -- (\xx*\xs-\xd,0.6*\ys) node[above] {$\scriptstyle \varphi_1 \varphi_2$};
				\filldraw[black] (\xx*\xs-\xd,0.6*\ys) circle (2pt);
				\draw[thick,dashed] (\xx*\xs+\xd,0) -- (\xx*\xs+\xd,2*\ys) node[above] {$\frac{1}{\varphi_1}$};
				\filldraw[fill=white, draw=black] (\xx*\xs+\xd,2*\ys) circle (2pt);
				
				\renewcommand\xx{2} 
				\node[below] at (\xx*\xs,0) {\xx};
				\draw[thick] (\xx*\xs-\xd,0) -- (\xx*\xs-\xd,1.2*\ys) node[above] {$\scriptstyle \varphi_2$};;
				\filldraw[black] (\xx*\xs-\xd,1.2*\ys) circle (2pt);
				\draw[thick,dashed] (\xx*\xs+\xd,0) -- (\xx*\xs+\xd,0.83*\ys) node[above] {$\frac{1}{\varphi_2}$};
				\filldraw[fill=white, draw=black] (\xx*\xs+\xd,0.83*\ys) circle (2pt);
				
				\renewcommand\xx{3} 
				\node[below] at (\xx*\xs,0) {\xx};
				\draw[thick] (\xx*\xs-\xd,0) -- (\xx*\xs-\xd,1*\ys) node[above] {$\scriptstyle 1$};
				\filldraw[black] (\xx*\xs-\xd,1*\ys) circle (2pt);
				\draw[thick,dashed] (\xx*\xs+\xd,0) -- (\xx*\xs+\xd,1*\ys) node[above] {$\scriptstyle 1$};
				\filldraw[fill=white, draw=black] (\xx*\xs+\xd,1*\ys) circle (2pt);
				
				\renewcommand\xx{4} 
				\node[below] at (\xx*\xs,0) {\xx};
				\draw[thick] (\xx*\xs-\xd,0) -- (\xx*\xs-\xd,1*\ys) node[above] {$\scriptstyle 1$};
				\filldraw[black] (\xx*\xs-\xd,1*\ys) circle (2pt);
				\draw[thick,dashed] (\xx*\xs+\xd,0) -- (\xx*\xs+\xd,1*\ys) node[above] {$\scriptstyle 1$};
				\filldraw[fill=white, draw=black] (\xx*\xs+\xd,1*\ys) circle (2pt);
				
				\node[] at (4.7*\xs,0.5*\ys) {$\boldsymbol{\dots}$};
			\end{tikzpicture}
			\caption{$\varphi_0=3.3, \varphi_1=0.5, \varphi_2=1.2$}
			\label{fig:wfig10}
		\end{subfigure}
		\caption{Several possible weight-function patterns for the type 1 (left part) and type 2 (right part) stationary models with $\mathcal{F}=\left\{ 0 \right\},\left\{ 1\right\},\left\{ 2\right\},\left\{ 0,2 \right\}$, and $\left\{0,1,2 \right\}$. The solid lines (\rule[0.5ex]{0.6cm}{0.5pt}) indicate $f(n,\cdot)$, while the dashed lines  ($\scriptstyle \text{- - -}$) indicate $g(n,\cdot)=f(n+1,\cdot)/f(n,\cdot)$.}
		\label{fig:weight_functions}
	\end{figure}
}

The ratio of type 2 probabilities for consecutive values typically differs from the ratio of base probabilities. For consecutive values at $n$ and $n+1$, the type 2 probability ratio is
\begin{equation}
	\lambda_{n}=g(n,\boldsymbol{\varphi})\lambda_{n}^b,
	\label{eq:prob_ratio}
\end{equation}
where
\begin{equation}
	g(n,\boldsymbol{\varphi})=\prod_{i=0}^{m} \varphi_{n_i}^{-\textbf{1}_{n_i}(n)},
	\label{eq:g_type_II}
\end{equation}
$\lambda_n=p(n+1,\boldsymbol{\theta},\boldsymbol{\varphi})/p(n,\boldsymbol{\theta},\boldsymbol{\varphi})$,
$g(n,\boldsymbol{\varphi})=f(n+1,\boldsymbol{\varphi})/f(n,\boldsymbol{\varphi})$,  $\lambda_{n}^b=b(n+1,\boldsymbol{\theta})/b(n,\boldsymbol{\theta})$, and $\textbf{1}_{n_i}(n)=u_{n_i}(n)-u_{n_i}(n+1)$. In contrast to $f(n,\boldsymbol{\varphi})$ \eqref{eq:f2}, $g(n,\boldsymbol{\varphi})$ does not represent the inflation-deflation behavior described in \eqref{eq:new_probII}. When $\varphi_{n_0}=...=\varphi_{n_m}=1$, this relationship reduces to $\lambda_{n}^b$.

Taking the logarithm of both sides of \eqref{eq:prob_ratio} gives
\begin{equation}
	\log \lambda_n= \log \lambda_n^b - \sum_{i=0}^{m} \textbf{1}_{n_i}(n) \log \varphi_{n_i}.
	\label{eq:log_prob_ratio}
\end{equation}
If $g(n,\varphi_1)=\varphi_1^{-\textbf{1}_1(n)}$ (see Figure \ref{fig:wfig7}), and $\lambda_n^b= \lambda/(n+1)^\nu$ (CMP), then $\lambda_0=\lambda, \lambda_1=\lambda/(\varphi_1 2^\nu), \lambda_2=\lambda/3^\nu, \lambda_3=\lambda/4^\nu,...$, so \eqref{eq:log_prob_ratio} simplifies to $\log \lambda_n= \log \lambda - \nu \log (n+1) - \textbf{1}_1(n) \log \varphi_1$. \citet[p.~3284]{Lerd25} recently introduced this ratio regression model to allow for one-inflation. Since $f(n,\varphi_1)=\varphi_{1}^{u_1(n)}$ (see Figure \ref{fig:wfig7}), the distribution is classified as zero-one-inflated rather than one-inflated. The one-inflation ratio regression model is $\log \lambda_n=\log \lambda - \nu \log (n+1) +(\textbf{1}_1(n+1)-\textbf{1}_1(n)) \log \alpha_1$, which belongs to the type 1 family (see Figure \ref{fig:wfig2}).

Generalizing the distribution of Lerdsuwansri et al. to cases where $\mathcal{F}=\left\{ q \neq 0\right\}$ is straightforward, as demonstrated below:
\begin{equation}
	p(n,\boldsymbol{\theta},\varphi) = \frac{\varphi^{u_q(n)} b(n,\boldsymbol{\theta})}{1+(\varphi-1)\sum_{k=0}^{q} b(k,\boldsymbol{\theta})}.
	\label{eq:prob_lerd}
\end{equation}
Assume $b(n,\boldsymbol{\theta})$ follows a Poisson distribution. When $\varphi=1$, the distribution remains Poisson. However, when $\varphi=1+\left ( \lambda-q \right )/\sum_{k=0}^{q-1}(q-k)\frac{\lambda^k e^{-\lambda}}{k!}$ for $\lambda \neq q$, the distribution is no longer Poisson, despite having both mean and variance equal to $q$ (see \Cref{lam_al_relation} and Figure \ref{fig:his_equi}).

Figure \ref{fig:weight_functions} illustrates how the weight functions $f(n,\cdot)$ and $g(n,\cdot)$ change as $\mathcal{F}$ varies to reflect different deviations from a base distribution. Figure \ref{fig:wfig8} demonstrates that the type 2 model with $\mathcal{F}=\left\{ 2 \right\}$ can account for zero-to-two inflation-deflation, while Figure \ref{fig:wfig3} shows that the type 1 model with $\mathcal{F}=\left\{ 2 \right\}$ only accounts for two inflation-deflation. Both stationary distributions above share the same base and number of parameters. Figures \ref{fig:wfig1} and \ref{fig:wfig6} present the weight functions $f(n,\cdot)$ of  \citet[p.~226]{Haslett22} and $g(n,\cdot)$ of \citet[p.~256]{Puig24}, respectively. Figures \ref{fig:wfig5} and \ref{fig:wfig10} indicate that type 1 and type 2 stationary models with $\mathcal{F}=\left\{ 0,1,...,q\right\}$ are identical in shape, as their weight functions match. Therefore, the type 2 model serves as an alternative parameterization of the mixture model \eqref{eq:mix_zk}.

For independent and identically distributed cases, the exponential-family distribution of Haslett et al \eqref{eq:haslett} has a key property: the sample mean is the maximum likelihood estimator (MLE) of the population mean, as seen in the Poisson and NB distributions \citep[p.~227]{Haslett22}. Because the type 1 and type 2 models include this distribution as a special case, they likely share this property.

\subsection{Exponential family}
\label{subsec:expo_family}
Before examining this possibility, we present the definition of an exponential family of count distributions \citep[p.~54]{Bickel15}. 

\begin{definition}
	\label{expo_family}
	A collection of count distributions
	\begin{equation*}
		\left\{ q(n,\boldsymbol{\eta}), n \in \left\{ 0,1,2,...\right\}, \boldsymbol{\eta} \in \mathcal{E} =\left\{ \boldsymbol{\eta} \in R^l : A(\boldsymbol{\eta})<\infty  \right\} \right\}
		\label{eq:collection}
	\end{equation*}
    of the form 
	\begin{equation*}
	q(n,\boldsymbol{\eta})=h(n) \text{exp}\left\{ \textbf{T}(n)\boldsymbol{\eta}^T-A(\boldsymbol{\eta})\right\}
	\label{eq:expo_form}
	\end{equation*}
	is a canonical (natural) $l$-parameter exponential family, where $h(n)$ is a real function, $\textbf{T}(n)$ is a $1 \times l$ vector of sufficient statistics, $\boldsymbol{\eta}$ is a $1 \times l$ canonical parameter vector, $A(\boldsymbol{\eta})=\log\sum_{n=0}^{\infty} h(n)\text{exp}[\textbf{T}(n) \boldsymbol{\eta}^T]$ is a cumulant (normalizing) function, and $\mathcal{E}$ is a canonical parameter space.
\end{definition}

By substituting $\lambda_n$ for the geometric, Poisson, NB, HP, and CMP distributions (see Table \ref{table:table_ratio}) into \eqref{eq:p0} and \eqref{eq:pn}, we obtain their probability mass functions (see Table \ref{table:table_prob}): 
\begin{equation} 
	p(n,\lambda)=\text{exp}\left\{ n\log\lambda+\log\left(1- \lambda \right) \right\},
	\label{eq:geometric}
\end{equation}
\begin{equation}
p(n,\lambda)=\frac{1}{n!}\text{exp}\left\{ n\log\lambda-\lambda \right\},
	\label{eq:poisson}
\end{equation}
\begin{equation} 
	p(n,\lambda,r)=\frac{\left( r \right)_n}{n!}\text{exp}\left\{ n\log\left( \frac{\lambda}{r} \right)+r\log\left(1- \frac{\lambda}{r} \right) \right\},
	\label{eq:nb}
\end{equation}
\begin{equation} 
	p(n,\lambda,\tau)=\frac{1}{\left( \tau \right)_n}\text{exp}\left\{ n\log\lambda-\log z\left(\lambda,\tau \right) \right\},
	\label{eq:hp}
\end{equation}
and
\begin{equation} 
	p(n,\lambda,\nu)=\text{exp}\left\{ n\log \lambda-\nu\log n!-\log z\left(\lambda,\nu \right) \right\}, 
	\label{eq:cmp}
\end{equation}
where $z(\lambda,\tau)=\sum_{i=0}^{\infty }\lambda^i/(\tau)_i$, and $z(\lambda,\nu)=\sum_{i=0}^{\infty }\lambda^i/(i!)^\nu$. These stationary distributions are members of the exponential family, with corresponding functions $h,\, \textbf{T},\, \boldsymbol{\eta}$, and $A$ listed in Table \ref{table:table_expo_family}. The geometric, Poisson, NB, and HP are one-parameter exponential family distributions, while the CMP is a two-parameter exponential family distribution. Table \ref{table:table_expo_family} indicates that for the geometric, Poisson, NB, HP, and CMP, $T_1(n)$ equals $n$. The Poisson, HP, and CMP have usual parameter spaces $0 < \lambda < \infty $ and canonical parameter spaces $-\infty <\eta_1<\infty$. The NB has a usual parameter space $0 < \lambda/r < 1 $ and a canonical parameter space $-\infty <\eta_1<0$. The geometric is a special case of the NB where $r=1$. When $\lambda=0$, these distributions are degenerate and concentrated at zero, so they are excluded from the exponential family. The PL distribution does not have the form of an exponential family as defined in \Cref{expo_family} because its $\log (1+\lambda+n\lambda)+n\log \lambda$ cannot be factored into separate functions of $\lambda$ and $n$.

For the geometric, Poisson, NB, HP, and CMP distributions, each $\lambda_n$ shares the common factor $\lambda=\gamma/\mu$ (see Table \ref{table:table_ratio}). Consequently, these stationary distributions have $T_1(n)=n$, which aids parameter estimation. The MLE of the population mean is the sample mean. However, because no components of $\textbf{T}(n)$ equal $n^2$ for these distributions, the MLE of the population variance is not equal to the sample variance. To address this, the birth-death-rate ratio sequences can be adjusted. For example, the modified Poisson ratio sequence of \citet[p.~253, eq.(7)]{Puig24} is $\lambda_n=e^{-(2n+1)\tau}\lambda/(n+1)$. With $f(n,\tau)=e^{-n^2\tau}$, and $b(n,\lambda)=\lambda^n e^{-\lambda}/n!$ (Poisson), this stationary distribution has $\textbf{T}(n)=[n,n^2]$.

By making modifications to $\lambda_n$ of the geometric, Poisson, NB, HP, and CMP distributions, the following proposition shows that the type 1 and type 2 models also belong to the exponential family. 

\begin{proposition}
	\label{ZkIS_expofamily}
	Suppose $b(n,\boldsymbol{\theta})$ form a canonical $l$-parameter exponential family with
	\begin{equation*}
		\begin{matrix}
			h(n), \\
			\textbf{T}_b (n)=[T_1(n),...,T_l(n)], \\
			\boldsymbol{\eta}_b=[\eta_1,...,\eta_l] \in \mathcal{E}_b,\, \text{and}  \\
			A_b (\boldsymbol{\eta}_b)=\log z_b (\boldsymbol{\eta}_b).
		\end{matrix}
	\end{equation*}
	Then \eqref{eq:new_inflat} form a canonical $(l+m+1)$-parameter exponential family with
	\begin{equation*}
		\begin{matrix}
			h(n), \\
			\textbf{T} (n)=[\textbf{T}_b(n),T_{l+1}(n)=\textbf{1}_{n_0}(n),...,T_{l+m+1}(n)=\textbf{1}_{n_m}(n)], \\
			\boldsymbol{\eta}=[\boldsymbol{\eta}_b,\eta_{l+1}=\log \alpha_{n_0},...,\eta_{l+m+1}=\log \alpha_{n_m} ] \in \mathcal{E}_b \times  R \times ...\times R,\, \text{and} \\
			A (\boldsymbol{\eta})=A_b (\boldsymbol{\eta}_b) + \log z (\boldsymbol{\eta}),
		\end{matrix}
	\end{equation*}
	and \eqref{eq:new_probII} form a canonical $(l+m+1)$-parameter exponential family with
	\begin{equation*}
		\begin{matrix}
			h(n), \\
			\textbf{T} (n)=[\textbf{T}_b(n),T_{l+1}(n)=u_{n_0}(n),...,T_{l+m+1}(n)=u_{n_m}(n)], \\
			\boldsymbol{\eta}=[\boldsymbol{\eta}_b,\eta_{l+1}=\log \varphi_{n_0},...,\eta_{l+m+1}=\log \varphi_{n_m} ] \in \mathcal{E}_b \times  R \times ...\times R,\, \text{and} \\
			A (\boldsymbol{\eta})=A_b (\boldsymbol{\eta}_b) + \log z (\boldsymbol{\eta}).
		\end{matrix}
	\end{equation*}
\end{proposition}

\afterpage{
	\begin{table}[H]
		\setlength{\tabcolsep}{0.2em}	
		\centering
		\caption{Characteristics of some type 1 and type 2 stationary distributions in the exponential family. The normalizing constants, denoted by $z(\cdot)$, are defind by \eqref{eq:z} and \eqref{eq:zII}.}    
		\label{table:table_expo_family}	
		\makebox[\textwidth]{\begin{tabular}{ >{\raggedright}p{3.2cm} >{\raggedright}p{6.1cm}	>{\raggedright\arraybackslash}p{6.1cm}}	
						\toprule 
						Base & Type 1 & Type 2\\				
						\midrule
						\underline{Geometric} & & \\						
						$\medmuskip=0mu \thickmuskip=0mu h(n)=1$ & $\medmuskip=0mu \thickmuskip=0mu h(n)=1$ & $\medmuskip=0mu \thickmuskip=0mu h(n)=1$ \\			
						$\medmuskip=0mu \thickmuskip=0mu \textbf{T}_b(n)=n$ & $\medmuskip=0mu \thickmuskip=0mu \textbf{T}(n)=[n,\textbf{1}_{n_0}(n),...,\textbf{1}_{n_m}(n)]$ & $\medmuskip=0mu \thickmuskip=0mu \textbf{T}(n)=[n,u_{n_0}(n),...,u_{n_m}(n)]$ \\
						$\medmuskip=0mu \thickmuskip=0mu \boldsymbol{\eta}_b=\log \lambda$ & $\medmuskip=0mu \thickmuskip=0mu \boldsymbol{\eta}=[\boldsymbol{\eta}_b,\log \alpha_{n_0},...,\log \alpha_{n_m} ]$ & $\medmuskip=0mu \thickmuskip=0mu \boldsymbol{\eta}=[\boldsymbol{\eta}_b,\log \varphi_{n_0},...,\log \varphi_{n_m} ]$\\
						$\medmuskip=0mu \thickmuskip=0mu -\log \left ( 1-e^{\eta_1} \right )$ & $\medmuskip=0mu \thickmuskip=0muA(\boldsymbol{\eta})=A_b(\boldsymbol{\eta_b}) + \log z (\boldsymbol{\eta})$ &  $\medmuskip=0mu \thickmuskip=0mu A(\boldsymbol{\eta})=A_b(\boldsymbol{\eta_b}) + \log z (\boldsymbol{\eta})$\\
						$\medmuskip=0mu \thickmuskip=0mu \mathcal{E}_b= (-\infty,0)$ & $\medmuskip=0mu \thickmuskip=0mu \mathcal{E}=\mathcal{E}_b \times  R \times ...\times R$ & $\medmuskip=0mu \thickmuskip=0mu \mathcal{E}=\mathcal{E}_b \times  R \times ...\times R$
						\\
						
						\underline{Poisson} & & \\						
						$\medmuskip=0mu \thickmuskip=0mu h(n)=1/n!$ & $\medmuskip=0mu \thickmuskip=0mu h(n)=1/n!$ & $\medmuskip=0mu \thickmuskip=0mu h(n)=1/n!$\\			
						$\medmuskip=0mu \thickmuskip=0mu \textbf{T}_b(n)=n$ & $\medmuskip=0mu \thickmuskip=0mu \textbf{T}(n)= [n,\textbf{1}_{n_0}(n),...,\textbf{1}_{n_m}(n)]$  & $\medmuskip=0mu \thickmuskip=0mu \textbf{T}(n)=[n,u_{n_0}(n),...,u_{n_m}(n)]$ \\
						$\medmuskip=0mu \thickmuskip=0mu \boldsymbol{\eta}_b=\log \lambda$ & $\medmuskip=0mu \thickmuskip=0mu \boldsymbol{\eta}=[\boldsymbol{\eta}_b,\log \alpha_{n_0},...,\log \alpha_{n_m} ]$ & $\medmuskip=0mu \thickmuskip=0mu \boldsymbol{\eta}=[\boldsymbol{\eta}_b,\log \varphi_{n_0},...,\log \varphi_{n_m} ]$\\
						$\medmuskip=0mu \thickmuskip=0mu A_b(\boldsymbol{\eta_b})=e^{\eta_1}$ & $\medmuskip=0mu \thickmuskip=0mu A(\boldsymbol{\eta})=A_b(\boldsymbol{\eta_b}) + \log z (\boldsymbol{\eta})$ & $\medmuskip=0mu \thickmuskip=0mu A(\boldsymbol{\eta})=A_b(\boldsymbol{\eta_b}) + \log z (\boldsymbol{\eta})$\\
						$\medmuskip=0mu \thickmuskip=0mu \mathcal{E}_b=R$ & $\medmuskip=0mu \thickmuskip=0mu \mathcal{E}=\mathcal{E}_b \times  R \times ...\times R$ & $\medmuskip=0mu \thickmuskip=0mu \mathcal{E}=\mathcal{E}_b \times  R \times ...\times R$
						\\
						
						\underline{NB} & & \\						
						$\medmuskip=0mu \thickmuskip=0mu \left( r\right)_n/n!$ & $\medmuskip=0mu \thickmuskip=0mu \left( r\right)_n/n!$ & $\medmuskip=0mu \thickmuskip=0mu \left( r\right)_n/n!$ \\			
						$n$ & $\medmuskip=0mu \thickmuskip=0mu [n,\textbf{1}_{n_0}(n),...,\textbf{1}_{n_m}(n)]$ & $\medmuskip=0mu \thickmuskip=0mu [n,u_{n_0}(n),...,u_{n_m}(n)]$ \\
						$\medmuskip=0mu \thickmuskip=0mu \log (\lambda/r)$ & $\medmuskip=0mu \thickmuskip=0mu [\log (\lambda/r),\log \alpha_{n_0},...,\log \alpha_{n_m} ]$ & $\medmuskip=0mu \thickmuskip=0mu [\log (\lambda/r),\log \varphi_{n_0},...,\log \varphi_{n_m} ]$\\
						$\medmuskip=0mu \thickmuskip=0mu -r \log \left ( 1-e^{\eta_1} \right )$ & $\medmuskip=0mu \thickmuskip=0mu -r \log \left ( 1-e^{\eta_1} \right ) + \log z (\boldsymbol{\eta})$ &  $\medmuskip=0mu \thickmuskip=0mu -r \log \left ( 1-e^{\eta_1} \right ) + \log z (\boldsymbol{\eta})$\\
						$\medmuskip=0mu \thickmuskip=0mu (-\infty,0)$ & $\medmuskip=0mu \thickmuskip=0mu (-\infty,0) \times  R \times ...\times R$ & $\medmuskip=0mu \thickmuskip=0mu (-\infty,0) \times  R \times ...\times R$
						\\
						
						\underline{HP} & & \\						
						$\medmuskip=0mu \thickmuskip=0mu 1/\left( \tau\right)_n$ & $\medmuskip=0mu \thickmuskip=0mu 1/\left( \tau\right)_n$ & $\medmuskip=0mu \thickmuskip=0mu 1/\left( \tau\right)_n$ \\			
						$n$ & $\medmuskip=0mu \thickmuskip=0mu [n,\textbf{1}_{n_0}(n),...,\textbf{1}_{n_m}(n)]$ & $\medmuskip=0mu \thickmuskip=0mu [n,u_{n_0}(n),...,u_{n_m}(n)]$ \\
						$\medmuskip=0mu \thickmuskip=0mu \log \lambda$ & $\medmuskip=0mu \thickmuskip=0mu [\log \lambda,\log \alpha_{n_0},...,\log \alpha_{n_m} ]$ & $\medmuskip=0mu \thickmuskip=0mu [\log \lambda,\log \varphi_{n_0},...,\log \varphi_{n_m} ]$\\
						$\medmuskip=0mu \thickmuskip=0mu \log \sum e^{\eta_1 i}/(\tau)_i$ & $\medmuskip=0mu \thickmuskip=0mu \log \sum e^{\eta_1 i}/(\tau)_i + \log z (\boldsymbol{\eta})$ &  $\medmuskip=0mu \thickmuskip=0mu \log \sum e^{\eta_1 i}/(\tau)_i + \log z (\boldsymbol{\eta})$\\
						$R$ & $\medmuskip=0mu \thickmuskip=0mu R \times  R \times ...\times R$ & $\medmuskip=0mu \thickmuskip=0mu R \times  R \times ...\times R$
						\\
						
						\underline{CMP} & & \\						
						$1$ & $1$ & $1$\\			
						$\medmuskip=0mu \thickmuskip=0mu [n, \log n!]$ & $\medmuskip=0mu \thickmuskip=0mu [n, \log n!,\textbf{1}_{n_0}(n),...,\textbf{1}_{n_m}(n)]$  & $\medmuskip=0mu \thickmuskip=0mu [n, \log n!,u_{n_0}(n),...,u_{n_m}(n)]$ \\
						$\medmuskip=0mu \thickmuskip=0mu [\log \lambda,-\nu]$ & $\medmuskip=0mu \thickmuskip=0mu [\log \lambda,-\nu,\log \alpha_{n_0},...,\log \alpha_{n_m} ]$ & $\medmuskip=0mu \thickmuskip=0mu [\log \lambda,-\nu,\log \varphi_{n_0},...,\log \varphi_{n_m} ]$\\ 
						$\medmuskip=0mu \thickmuskip=0mu \log \sum e^{\eta_1 i}(i!)^{\eta_2}$ & $\medmuskip=0mu \thickmuskip=0mu \log \sum e^{\eta_1 i}(i!)^{\eta_2} + \log z (\boldsymbol{\eta})$ & $\medmuskip=0mu \thickmuskip=0mu \log \sum e^{\eta_1 i}(i!)^{\eta_2} + \log z (\boldsymbol{\eta})$\\
						$\medmuskip=0mu \thickmuskip=0mu R \times (-\infty,0)$ & $\medmuskip=0mu \thickmuskip=0mu R \times (-\infty,0) \times  R \times ...\times R$ & $\medmuskip=0mu \thickmuskip=0mu R \times (-\infty,0) \times  R \times ...\times R$
						\\
						\bottomrule
						\multicolumn{3}{l}{\rule{0pt}{3ex}Note that $\sum =\sum_{i=0}^{\infty}$, $\left( r\right)_n=r\left( r+1\right)...\left ( r+n-1 \right )$, and $\left( r\right)_0=1$.}\\
				\end{tabular}}	
			\end{table}
		}

Table \ref{table:table_expo_family} demonstrates how several stationary distributions can be expressed as exponential-family distributions. The inflation-deflation distributions in the last two columns of Table \ref{table:table_expo_family} inherit the property of their base distributions: the MLE of the population mean equals the sample mean. This holds because $T_1(n)$ equals $n$. The exponential family structures in \eqref{eq:new_inflat} and \eqref{eq:new_probII} are quite simple. For example, the functions $h(n)$ in \eqref{eq:new_inflat} and in the base distribution are the same. The vector-valued function $\textbf{T}(n)$ in \eqref{eq:new_inflat} combines $\textbf{T}_b(n)$ and $\textbf{1}_{n_i}(n)$ for $i=0,...,m$. $\boldsymbol{\eta}$ can be divided into $\boldsymbol{\eta}_b$, which contains the base distribution parameters, and $[\log \alpha_{n_0},...,\log \alpha_{n_m} ]$, which contains the logarithms of the inflation-deflation parameters. Adding $A_b(\boldsymbol{\eta_b})$ and the logarithm of the normalizing constant function \eqref{eq:z} yields $A(\boldsymbol{\eta})$.

The exponential family has an interesting connection to the stationary distribution that satisfies \eqref{eq:pcon}. That is,
\begin{equation}
	A(\boldsymbol{\eta}) - \log h\left ( 0 \right ) - \textbf{T}\left ( 0 \right ) \boldsymbol{\eta}^T=\log \left( 1+\lambda_0 + \lambda_0\lambda_1 + \lambda_0\lambda_1\lambda_2 + ...\right),
	\label{eq:A_relation}
\end{equation}
where 
\begin{equation*}
	\lambda_0\lambda_1...\lambda_n = \frac{h\left ( n+1 \right )\exp \left\{ \textbf{T}\left ( n+1 \right ) \boldsymbol{\eta}^T \right\}}{h\left ( 0 \right )\exp \left\{ \textbf{T}\left ( 0 \right ) \boldsymbol{\eta}^T \right\}}.
\end{equation*}

\noindent Obviously if $h\left ( 0 \right )=1$, and $\textbf{T}\left ( 0 \right )=\left [ 0,...,0 \right ]$, then $A(\boldsymbol{\eta})=\log ( 1+\lambda_0+\lambda_0\lambda_1 + \lambda_0\lambda_1\lambda_2 + ...)$, where $\lambda_0\lambda_1...\lambda_n =h\left ( n+1 \right )\exp \left\{ \textbf{T}\left ( n+1 \right ) \boldsymbol{\eta}^T \right\}$. For example, the geometric, Poisson, NB, HP, and CMP have this equality. The type 1 Poisson distribution with $\mathcal{F}=\left\{ 0,1,2 \right\}$ has $A(\boldsymbol{\eta})=\log ( \alpha_0+\alpha_1\lambda + \alpha_2 \lambda^2/2! + \lambda^3/3! + \lambda^4/4! + ...)$, where $\boldsymbol{\eta}=[ \log \lambda, \log \alpha_0, \log \alpha_1, \log \alpha_2 ]$, because $h\left ( 0 \right )=1$, and $\textbf{T}\left ( 0 \right )= [ 0,1,0,0  ]$. The type 2 distribution with the same base and $\mathcal{F}$ has $A(\boldsymbol{\eta})=\log ( \varphi_0 \varphi_1 \varphi_2+\varphi_1 \varphi_2\lambda + \varphi_2 \lambda^2/2! +\lambda^3/3! + \lambda^4/4!+ ...)$, where $\boldsymbol{\eta}= [ \log \lambda, \log \varphi_0, \log \varphi_1, \log \varphi_2 ]$, because $h\left ( 0 \right )=1$, and $\textbf{T}\left ( 0 \right )= [ 0,1,1,1 ]$. The derivation of \eqref{eq:A_relation} is given in the supplementary material. 

The CMP, a stationary distribution, can model both underdispersed ($1<\nu<\infty$) and overdispersed ($0<\nu<1$) data (\citealt[p.~1294]{Kokonendji08}, \citealt[p.~252]{Puig24}). However, it is not always the most suitable choice, as its shape can differ substantially from the observed data. For instance, \citet[p.~471]{Winkelmann95} reports data on the number of children born, with a variance and mean of approximately 2.328 and 2.384, respectively. Although these pure-birth data exhibit mild underdispersion, the CMP does not provide the best fit according to the log-likelihood (see \citealt[Table~4]{Skul22}, \citealt[Table~1]{Hun25}, \citealt[Table~3]{Baker26}). Additionally, as a birth-death process, the CMP cannot describe pure-birth count data. This case demonstrates that dispersion alone does not guarantee an adequate fit for equi-, over-, or underdispersed count data. To provide further clarification, the next section examines the variance-mean relationship in the type 2 model \eqref{eq:new_probII} using the following proposition, which demonstrates that many equidispersed distributions are not Poisson, and many overdispersed and underdispersed distributions are not CMP. The mean and variance of the type 1 model \eqref{eq:new_inflat} are provided in Section S8 of the supplementary material.

\begin{proposition}
	\label{ZkIS_meanvariance}
	Suppose $N$ is a type 2 stationary distribution \eqref{eq:new_probII}, and $b(n,\boldsymbol{\theta})$ is a member of the exponential family with $T_1(n)=n$. Then
	\begin{equation*} 
		E[N]=E_b[N]+\frac{1}{z\left ( \boldsymbol{\theta},\boldsymbol{\varphi} \right )}\sum_{i=0}^{m} \left ( \prod_{j=i}^{m}\varphi_{n_j}-1 \right )  \sum_{k=n_{i-1}+1}^{n_i}b\left ( k,\boldsymbol{\theta} \right ) \left ( k-E_b[N] \right ),
	\end{equation*}
	and
	\begin{align*}
		&V[N] =V_b [N]-\left ( E[N]- E_b[N]\right )^2 \\
		& +\frac{1}{z\left ( \boldsymbol{\theta},\boldsymbol{\varphi} \right )}\sum_{i=0}^{m} \left ( \prod_{j=i}^{m}\varphi_{n_j}-1 \right )  \sum_{k=n_{i-1}+1}^{n_i}b\left ( k,\boldsymbol{\theta} \right ) \left ( \left ( k-E_b[N] \right )^2-V_b[N] \right ),
	\end{align*}
	where $n_{-1}=-1$. $E_b[N]$ and $V_b[N]$ denote the mean and variance of $b(n,\boldsymbol{\theta})$, respectively.
\end{proposition}

\section{Variance-mean relationship}
\label{sec:variance_mean}
\begin{figure}[!t]
	\centering
	\includegraphics[width=0.9\textwidth]{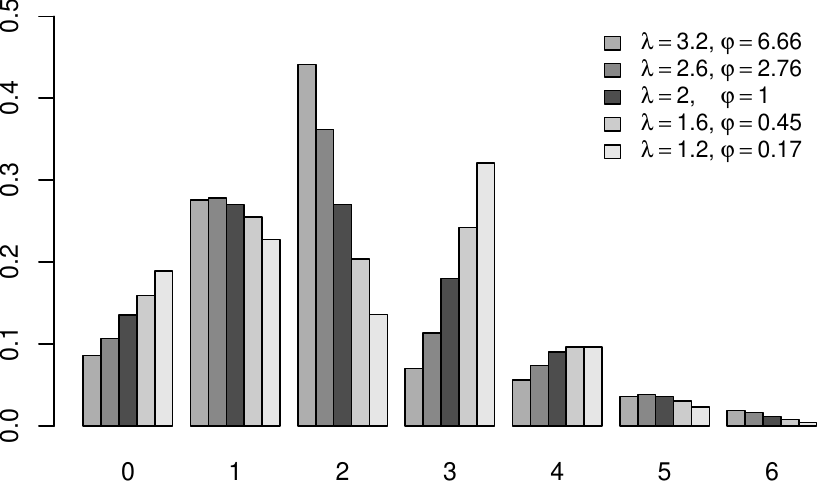}
	
	\caption{Count distributions with equal mean = 2, variance = 2, and variance-to-mean ratio = 1.  These five distributions are all equidispersed, but only the black histogram is a Poisson distribution.}
	\label{fig:his_equi}
\end{figure}

The dispersion index, defined as the variance-to-mean ratio in a count distribution, classifies distributions as equidispersed (ratio = 1), overdispersed (ratio > 1), or underdispersed (ratio < 1). The Poisson distribution has a dispersion index of one. Wikipedia (\url{https://en.wikipedia.org/wiki/Index_of_dispersion}, accessed 2/23/2026) reports ``The relevance of the index of dispersion is that it has a value of 1 when the probability distribution of the number of occurrences in an interval is a Poisson distribution. Thus the measure can be used to assess whether observed data can be modeled using a Poisson process.'' \citet[p.~357]{Karlis00} state ``\ldots the index of dispersion test is very common in testing the Poisson assumption\ldots This quantity is directly related to the variance-to-mean ratio\ldots .'' \citet[p.~1287]{Kokonendji08} write, ``One of the popular measures to detect such departures from the Poisson distribution is the so-called \textit{Fisher index} which is the ratio of variance to the mean ($\lessgtr 1$) of the count distribution.''

\subsection{Equidispersion}
\label{subsec:equidispersion}

\begin{table}[!t]
	\setlength{\tabcolsep}{0.2em}	
	\centering
	\caption{Skewness ($\gamma$), kurtosis ($\kappa_1$) within 1 standard deviation of the mean, kurtosis ($\kappa$) and the ratio between $\kappa_1$ and $\kappa$ for the five count distributions in  Figure \ref{fig:his_equi} with various birth-death-rate ratio sequences ($\lambda_0,\lambda_1,\lambda_2,...$)}    \
	\label{table:table_seq}
	\resizebox{\textwidth}{!}{\begin{tabular}{ >{\centering}p{0.7cm} >{\centering}p{1.4cm} >{\centering}p{0.8cm} >{\centering}p{0.8cm} >{\centering}p{0.8cm} >{\centering}p{0.8cm} >{\centering}p{0.8cm} >{\centering}p{0.8cm} >{\centering}p{1.5cm} >{\centering}p{1.5cm} >{\centering}p{0.8cm}  >{\centering\arraybackslash}p{1.6cm}}
							
			\toprule 
			$\lambda$ & $\varphi$ & $\lambda_0$ & $\lambda_1$ & $\lambda_2$ & $\lambda_3$ & $\lambda_4$ & $\lambda_5$ & $\gamma$ & $\kappa_1$ & $\kappa$ & $\kappa_1/\kappa$ \\
			\midrule
			3.2 & 6.661 & 3.20 & 1.60 & \underline{0.16} & 0.80 & 0.64 & 0.53 & 1.5556 & 0.0866 & 6.62  & 0.0131\\
			2.6 & 2.756 & 2.60 & 1.30 & \underline{0.31} & 0.65 & 0.52 & 0.43 & 1.1314 & 0.0981 & 4.88  & 0.0201\\
			2.0 & 1.000 & 2.00 & 1.00 & \underline{0.67} & 0.50 & 0.40 & 0.33 & 0.7071 & 0.1128 & 3.50  & 0.0322\\
			1.6 & 0.450 & 1.60 & 0.80 & \underline{1.18} & 0.40 & 0.32 & 0.27 & 0.4243 & 0.1244 & 2.78  & 0.0447\\
			1.2 & 0.170 & 1.20 & 0.60 & \underline{2.35} & 0.30 & 0.24 & 0.20 & 0.1414 & 0.1372 & 2.22  & 0.0618\\
			\bottomrule  
			\multicolumn{12}{l}{\rule{0pt}{3ex}Note that $\lambda_2=\lambda/(3\varphi)$, and  $\lambda_n=\lambda/(n+1)$ for $n \neq2$.}\\
			\end{tabular}}
	\end{table}

Figure \ref{fig:his_equi} presents five probability mass functions for various $\lambda\in\left\{ 1.2,1.6,2,2.6,3.2 \right\}$ and $\varphi\in\left\{ 0.17,0.45,1,2.76,6.66 \right\}$ values, each with a mean and variance of 2. These examples illustrate how shape, skewness, and kurtosis vary among equidispersed distributions. We construct these distributions using the following proposition.

	\begin{proposition}
		\label{lam_al_relation}
		If $b(n,\boldsymbol{\theta})$ is a Poisson distribution, $\mathcal{F}= \left\{q \neq0 \right\}$, and
		\begin{equation*}
			\varphi=1+\frac{\lambda-q}{\sum_{k=0}^{q-1}(q-k)\frac{\lambda^k e^{-\lambda}}{k!}},
		\end{equation*}
		then the type 2 stationary model \eqref{eq:new_probII} is equidispersed with mean and variance equal to $q$.
	\end{proposition}
	
	The count distributions with $(\lambda = 3.2, \varphi = 6.661)$ and $(\lambda = 2.6, \varphi = 2.756)$ are unimodal, each with a single peak at 2. These distributions show greater right skewness than the Poisson distribution $(\lambda = 2, \varphi = 1)$ as shown in the first three lines of Table \ref{table:table_seq}. In contrast, the distributions with $(\lambda = 1.6, \varphi = 0.45)$ and $(\lambda = 1.2, \varphi = 0.17)$ are bimodal, with peaks at 1 and 3, and display less right skewness than the Poisson (see the last two lines in Table \ref{table:table_seq}). Skewness for a count distribution is calculated as $\gamma = E[((N-\mu)/ \sigma)^3]$, and as $1/\sqrt{\lambda}$ for a Poisson distribution. The dispersion index alone does not indicate the shape of equidispersed distributions, whether unimodal, bimodal, or multimodal. ``Dispersion index'' and ``departure from Poissonity'' are unrelated. Therefore, using the dispersion index to measure ``departure from Poissonity'' is incorrect.

\afterpage{
	\begin{table}[H]
		\setlength{\tabcolsep}{0.2em}	
		\small 
		\centering
		\caption{The sequence $\left\{ a_n = (n+1)\lambda_n,\: n=0,1,2,... \right\}$ behavior of the geometric, Poisson, PL, NB, HP, CMP, type 1, and type 2 distributions, all of which are stationary distributions. Both type 1 and type 2 Poisson distributions have $\mathcal{F}=\left\{ 0 \right\}$. Note that $0<\lambda<\infty$ except for the geometric and PL distributions where $0<\lambda<1.$} 
		\label{table:theta_sequence}	
				\makebox[\textwidth]{\begin{tabular}{lcccc}
						\toprule 
						\multirow{2}{*}{Distribution} & \multicolumn{1}{c}{sequence} & \multicolumn{1}{c}{constant} & \multicolumn{1}{c}{increasing} & \multicolumn{1}{c}{decreasing}\\
						& \multicolumn{1}{c}{$a_n=(n+1)\lambda_n$} & \multicolumn{1}{c}{(equidispersion)} & \multicolumn{1}{c}{(overdispersion)} & \multicolumn{1}{c}{(underdispersion)}\\	
						\midrule
						Geometric & $(n+1)\lambda$ & - & \checkmark & -\\[2ex]
						
						Poisson & $\lambda$ & \checkmark & - & -\\[2ex]
						
						PL & $\displaystyle  \frac{(n+1)(1+(n+2)\lambda)\lambda}{1+(n+1)\lambda}$ & - & \checkmark & -\\[2ex]
						
						\multirow{2}{*}{NB} &  \multirow{2}{*}{$\displaystyle \left ( \frac{n}{r}+1 \right ) \lambda$}& \multirow{2}{*}{-}&\multicolumn{1}{c}{\checkmark}& \multirow{2}{*}{-}\\
						& & &\multicolumn{1}{c}{$(0<\lambda /r<1)$}&\\[2ex]
						
						\multirow{2}{*}{HP} &  \multirow{2}{*}{$\displaystyle \frac{(n+1)\lambda}{n+\tau}$}&
						\multicolumn{1}{c}{\checkmark}&\multicolumn{1}{c}{\checkmark}& \multicolumn{1}{c}{\checkmark}\\
						& &\multicolumn{1}{c}{$(\tau=1)$} &\multicolumn{1}{c}{$(1<\tau<\infty)$}&\multicolumn{1}{c}{$(0<\tau<1)$}\\[2ex]

						\multirow{2}{*}{CMP} &  \multirow{2}{*}{$\displaystyle  \frac{(n+1)\lambda}{(n+1)^{\nu}}$}&
						\multicolumn{1}{c}{\checkmark}&\multicolumn{1}{c}{\checkmark}& \multicolumn{1}{c}{\checkmark}\\
						& &\multicolumn{1}{c}{$(\nu=1)$} &\multicolumn{1}{c}{$(0<\nu<1)$}&\multicolumn{1}{c}{$(1<\nu<\infty)$}\\[2ex]
						
						\multirow{2}{*}{Type 1} &  \multirow{2}{*}{$\displaystyle  \frac{\lambda}{\alpha^{\textbf{1}_{0}(n)}}$}&
						\multicolumn{1}{c}{\checkmark}&\multicolumn{1}{c}{\checkmark}& \multicolumn{1}{c}{\checkmark}\\
						& &\multicolumn{1}{c}{$(\alpha=1)$} &\multicolumn{1}{c}{$(1<\alpha<\infty)$}&\multicolumn{1}{c}{$(0<\alpha<1)$}\\[2ex]
						
						\multirow{2}{*}{Type 2} &  \multirow{2}{*}{$\displaystyle  \frac{\lambda}{\varphi^{\textbf{1}_{0}(n)}}$}&
						\multicolumn{1}{c}{\checkmark}&\multicolumn{1}{c}{\checkmark}& \multicolumn{1}{c}{\checkmark}\\
						& &\multicolumn{1}{c}{$(\varphi=1)$} &\multicolumn{1}{c}{$(1<\varphi<\infty)$}&\multicolumn{1}{c}{$(0<\varphi<1)$}\\
						
						\bottomrule	
				\end{tabular}}	
			\end{table}
		}

	Figure \ref{fig:his_equi} shows five distributions with a dispersion index of one; only one is Poisson. In count-data terms, all are equidispersed. We omit the tails, which extend to infinity, as the central portion (about 99.4\%) sufficiently represents their shapes. Table \ref{table:table_seq} provides the kurtosis within one standard deviation of the mean $\kappa_1$ and the kurtosis $\kappa$ for each distribution. The $\kappa_1$ to $\kappa$ ratios are all below 7\% (see Table \ref{table:table_seq}, last column), indicating that kurtosis primarily reflects the tails \citep{Westfall14}. For count distributions, kurtosis is calculated as $\kappa = E[((N-\mu)/ \sigma)^4]$. For a Poisson distribution, it is even simpler than skewness: $1/\lambda + 3$.

\subsection{Combination of equi-, over-, and underdispersion}
\label{subsec:combi_dispersion}
The NB, HP, and CMP are flexible count distributions that address overdispersion or underdispersion and include the Poisson distribution as a special case. However, their equidispersed distributions are always Poisson, which does not accurately represent real equidispersed (non-Poisson) count data. As a result, using statistical methods based on NB, HP, or CMP assumptions (i.e., birth-death-rate ratio sequences) can significantly affect the validity of statistical inference. The geometric, Poisson, and PL are one-parameter ($\lambda$) distributions with different birth-death-rate ratio sequences (see Table \ref{table:table_ratio}), so the geometric and PL do not include the Poisson as a special case and are always overdispersed. In contrast, the type 1 and type 2 Poisson models include the Poisson distribution as a special case and also offer many equidispersed distributions that are not Poisson, as discussed in the previous section. These alternative models may not have a simple relationship between mean and variance.

Analysis of sequence $\left\{ a_n = (n+1)\lambda_n,\: n=0,1,2,... \right\}$ behavior has been employed to demonstrate the relationship between the mean and variance in stationary birth-death processes. \citet{Wise62} and \citet[Corollary~1]{Puig24} show that an increasing sequence $a_n$ leads to an overdispersed distribution, whereas a decreasing sequence $a_n$ leads to an underdispersed distribution. Table \ref{table:theta_sequence} summarizes the sequence behavior for various stationary distributions. The parameters $\tau, \nu, \alpha$, and $\varphi$, except for $\lambda$, indicate overdispersion or underdispersion and are referred to as dispersion parameters. The geometric, Poisson, PL, and NB distributions cannot accommodate datasets that combine equi-, over-, and underdispersion. The HP and CMP models, when applied with varying ratio ($\lambda$) and dispersion ($\tau$ and $\nu$) parameters in a double-regression model, can address such tridispersed data. However, this approach often doubles the number of parameters without significantly increasing the log-likelihood (\citealt{Skul22}, \citealt{Hun25}).

\afterpage{
	\begin{figure}[H]
		\begin{subfigure}{0.5\textwidth}
			\centering
			\includegraphics[height=4.6cm]{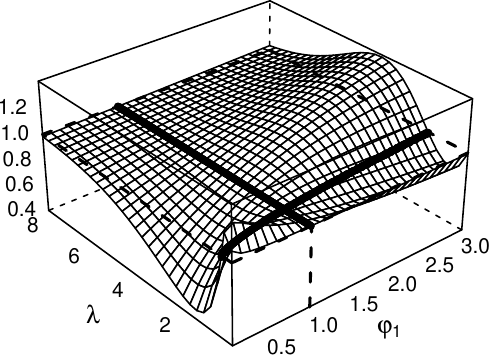}
			\caption{$q=1$}
			\label{fig:sfig1}
			\includegraphics[height=4.6cm]{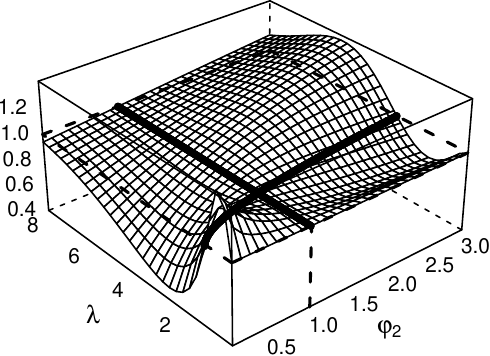}
			\caption{$q=2$}
			\label{fig:sfig2}
			\includegraphics[height=4.6cm]{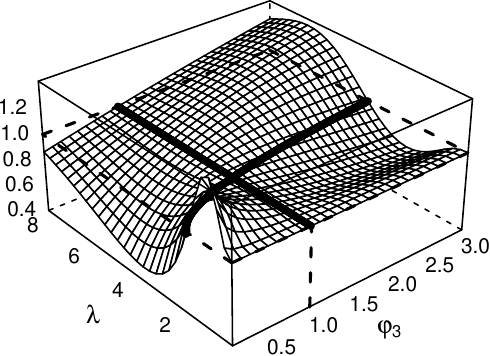}
			\caption{$q=3$}
			\label{fig:sfig3}
			\includegraphics[height=4.6cm]{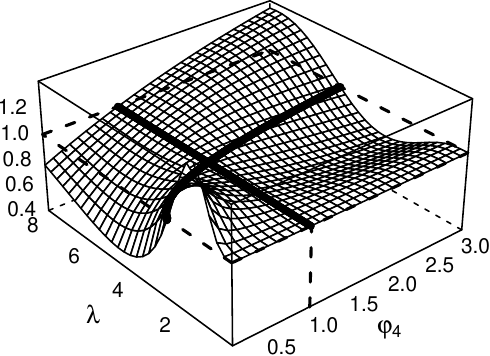}
			\caption{$q=4$}
			\label{fig:sfig4}
		\end{subfigure}
		\begin{subfigure}{0.5\textwidth}
			\centering
			\includegraphics[height=4.6cm]{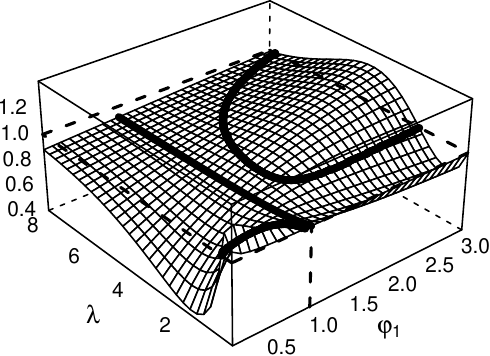}
			\caption{$q=1$, $\nu=1.1$}
			\label{fig:sfig5}
			\includegraphics[height=4.6cm]{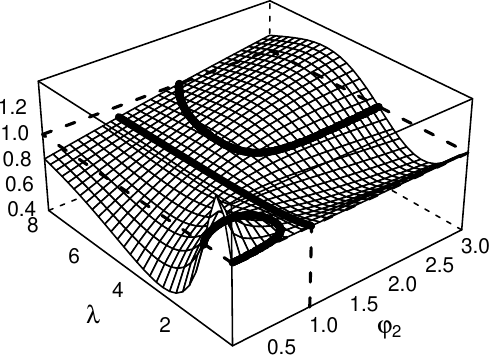}
			\caption{$q=2$, $\nu=1.1$}
			\label{fig:sfig6}
			\includegraphics[height=4.6cm]{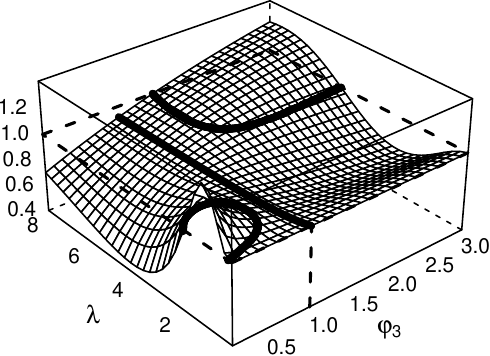}
			\caption{$q=3$, $\nu=1.1$ }
			\label{fig:sfig7}
			\includegraphics[height=4.6cm]{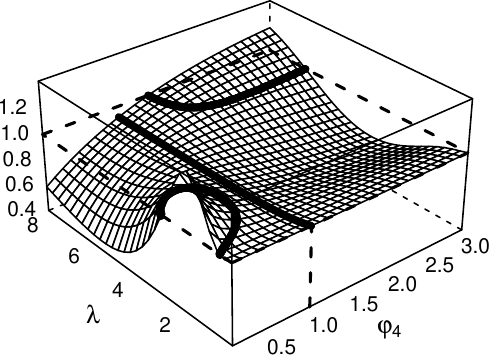}
			\caption{$q=4$, $\nu=1.1$}
			\label{fig:sfig8}
		\end{subfigure}
		\caption{Dispersion surfaces of the type 2 Poisson (left) and CMP (right) models with $\mathcal{F}=\left\{ q \right\}$. The straight (Poisson) and curved (non-Poisson) lines display equidispersion, that is, the level lines are equal to one. An exception is the CMP lines with $\varphi_q=1$, which show underdispersion.}
		\label{fig:vmr}
	\end{figure}
}

Typically, the sequence $a_n$ of the type 1 or type 2 stationary models is not monotonic, so Corollary 1 of \citet{Puig24} cannot be applied. An exception occurs with the type 1 and type 2 Poisson models when $\mathcal{F}= \left\{ 0\right\}$ (see the last two lines in Table \ref{table:theta_sequence}). We now examine the relationship that explains the inequality between the mean and variance, using the simplest inflation-deflation stationary models. Figure \ref{fig:vmr} presents the variance-mean ratios for the type 2 Poisson and CMP models with $\mathcal{F}= \left\{ q \neq 0\right\}$. $a_n=\lambda/\varphi_q^{\textbf{1}_{q}(n)}$ and $a_n=(n+1)^{1-\nu}\lambda/\varphi_q^{\textbf{1}_{q}(n)}$  are the non-monotonic sequences for the Poisson and CMP bases, respectively. The curved (non-Poisson) lines in Figures \ref{fig:sfig1}-\ref{fig:sfig4} show the functions $\varphi_q=1+(\lambda-q)/\sum_{k=0}^{q-1}\left ( q-k \right )\left ( \lambda ^{k}e^{-\lambda }/k! \right )$ in \Cref{lam_al_relation}. We observe that $\lambda$ indicates equidispersion, overdispersion, or underdispersion when $\varphi_q$ is fixed. For instance, if $\varphi_3=0.2$, the type 2 Poisson model can display equidispersion ($\lambda=2.055$), overdispersion ($0<\lambda<2.055$), or underdispersion ($2.055<\lambda<8$) (see Figure \ref{fig:sfig3}), and the type 2 CMP model can display equidispersion ($\lambda=0.237, 2.159$), overdispersion ($0.237<\lambda<2.159$), or underdispersion ($0<\lambda<0.237$, and $2.159<\lambda<8$) (see Figure \ref{fig:sfig7}). Unlike the classical HP and CMP models, $\lambda$ serves as a dispersion parameter. This feature is essential for regression analysis of count data, including count time-series models. Although the HP and CMP are stationary birth-death processes, they lack this property. We conclude that even the simplest models in \eqref{eq:new_probII}, $\mathcal{F}= \left\{ q \neq 0\right\}$, offer considerable flexibility for count data. Note that when $\varphi_q=1$, the lines on the left (right) correspond to equidispersed Poisson distributions (underdispersed CMP distributions).


\begin{thebibliography}{99}
\renewcommand{\itemsep}{0pt}

\bibitem[\protect\citeauthoryear{Arora et al.}{2021}]{Arora21} Arora, M., Rao Chaganty, N., and Sellers, K. F. (2021). A flexible regression model for zero-and k-inflated count data. {\it Journal of Statistical Computation and Simulation} 91(9), 1815--1845.

\bibitem[\protect\citeauthoryear{Baker}{2026}]{Baker26} Baker, R. (2026). A model for underdispersed count data. {\it Statistical Papers}, 67(2):30, 1--14.

\bibitem[\protect\citeauthoryear{Bardwell and Crow}{1964}]{Bardwell64} Bardwell, G. E. and Crow, E. L. (1964). A two-parameter family of hyper-Poisson distributions. {\it Journal of the American Statistical Association} 59(305), 133--141.

\bibitem[\protect\citeauthoryear{Bickel and Doksum}{2015}]{Bickel15} Bickel, P. J. and Doksum, K. A. (2015). {\it Mathematical Statistics: Basic Ideas and Selected Topics Volume I}, 2nd Ed. Chapman and Hall/CRC.

\bibitem[\protect\citeauthoryear{Böhning and Junnumtuam}{2025}]{Böh25} Böhning, D. and Junnumtuam, S. (2025). Some general points for inflation models. {\it Statistics and Probability Letters} 219, 1--6.

\bibitem[\protect\citeauthoryear{Böhning}{2016}]{Böh16} Böhning, D. (2016). Ratio plot and ratio regression with applications to social and medical sciences. {\it Statistical Science}, 205--218.

\bibitem[\protect\citeauthoryear{Boswell and Patil}{1970}]{Boswell70} Boswell, M. T. and Patil, G. P. (1970). Chance mechanisms generating the negative binomial distributions. {\it Random counts in models and structures} 1, 3--22.

\bibitem[\protect\citeauthoryear{Cameron}{2013}]{Cameron13} Cameron, A.C. and Trivedi, P.K. (2013). {\it Regression Analysis of Count Data}. Cambridge University Press.

\bibitem[\protect\citeauthoryear{Conway and Maxwell}{1962}]{Conway62} Conway, R. W. and Maxwell, W. L. (1962). A queuing model with state dependent service rates. {\it Journal of Industrial Engineering} 12, 132--136.

\bibitem[\protect\citeauthoryear{Cox and Miller}{1965}]{Cox65} Cox, D. R. and Miller, H.D. (1965). {\it The Theory of Stochastic Processes}. Chapman and Hall, London, UK.

\bibitem[\protect\citeauthoryear{Cragg}{1971}]{Cragg71} Cragg, J. G. (1971). Some statistical models for limited dependent variables with application to the demand for durable goods. {\it Econometrica: Journal of the Econometric Society} 39(5), 829--844.

\bibitem[\protect\citeauthoryear{Crawford et al.}{2018}]{Crawford18} Crawford, F. W., Ho, L. S. T., and Suchard, M. A. (2018). Computational methods for birth-death processes. {\it WIREs Computational Statistics} 10(2), 1--22.

\bibitem[\protect\citeauthoryear{Del Castillo and Pérez-Casany}{1998}]{Castillo98} Del Castillo, J. and Pérez-Casany, M. (1998). Weighted Poisson distributions for overdispersion and underdispersion situations. {\it Annals of the Institute of Statistical Mathematics} 50(3), 567-585.

\bibitem[\protect\citeauthoryear{Faddy}{1997}]{Faddy97} Faddy, M. J. (1997). Extended Poisson process modelling and analysis of count data. {\it Biometrical Journal} 39(4), 431--440.

\bibitem[\protect\citeauthoryear{Feng}{2021}]{Feng21} Feng, C. X. (2021). A comparison of zero-inflated and hurdle models for modeling zero-inflated count data. {\it Journal of statistical distributions and applications} 8(1), 1--19.

\bibitem[\protect\citeauthoryear{Grimmett and Stirzaker}{2001}]{Grim01} Grimmett, G. and Stirzaker, D. (2001). {\it Probability and Random Processes}, 3rd Ed. Oxford University Press, USA.

\bibitem[\protect\citeauthoryear{Haslett et al.}{2022}]{Haslett22} Haslett, J., Parnell, A. C., Hinde, J., and de Andrade Moral, R. (2022). Modelling excess zeros in count data: A new perspective on modelling approaches. {\it International statistical review} 90(2), 216--236.

\bibitem[\protect\citeauthoryear{Hunkrajok and Skulpakdee}{2025}]{Hun25} Hunkrajok, M. and Skulpakdee, W. (2025). A simple algorithm for computing the probabilities of count models based on pure birth processes. {\it Computational Statistics} 40(1), 249--272.

\bibitem[\protect\citeauthoryear{Janardan}{2005}]{Janardan05} Janardan, K. G. (2005). A discrete distribution associated with a pure birth process. {\it Statistical Papers} 46(4), 587-597.

\bibitem[\protect\citeauthoryear{Johnson and Kotz}{1969}]{Johnson69} Johnson, N. L. and Kotz, S. (1969). {\it Distributions in Statistics: Discrete Distributions}, Wiley, New York.

\bibitem[\protect\citeauthoryear{Johnson et al.}{2005}]{Johnson05} Johnson, N. L., Kemp, A. W., and Kotz, S. (2005). {\it Univariate discrete distributions}, John Wiley \& Sons.

\bibitem[\protect\citeauthoryear{Karlis and Xekalaki}{2000}]{Karlis00} Karlis, D. and Xekalaki, E. (2000). A simulation comparison of several procedures for testing the Poisson assumption. {\it Journal of the Royal Statistical Society Series D: The Statistician} 49(3), 355--382.

\bibitem[\protect\citeauthoryear{Kokonendji et al.}{2008}]{Kokonendji08} Kokonendji, C. C., Mizere, D., and Balakrishnan, N. (2008). Connections of the Poisson weight function to overdispersion and underdispersion. {\it Journal of Statistical Planning and Inference} 138(5), 1287--1296.


\bibitem[\protect\citeauthoryear{Lambert}{1992}]{Lambert92} Lambert, D. (1992). Zero-inflated Poisson regression, with an application to defects in manufacturing. {\it Technometrics} 34(1), 1--14.

\bibitem[\protect\citeauthoryear{Lerdsuwansri et al.}{2025}]{Lerd25} Lerdsuwansri, R., Pijitrattana, P., Sangnawakij, P., Lanumteang, K., Maruotti, A., Friedl, H., and Böhning, D. (2025). Identifying one-inflation in regression models for ratio estimators in single-source capture-recapture problems. {\it Journal of Statistical Computation and Simulation} 95(15), 3279--3299.

\bibitem[\protect\citeauthoryear{Lin and Tsai}{2013}]{Lin13} Lin, T. H. and Tsai, M. H. (2013). Modeling health survey data with excessive zero and K responses. {\it Statistics in Medicine} 32(9), 1572--1583.

\bibitem[\protect\citeauthoryear{Melkersson and Rooth}{2000}]{Melkersson00} Melkersson, M. and Rooth, D. O. (2000). Modeling female fertility using inflated count data models. {\it Journal of Population Economics} 13(2), 189--203.

\bibitem[\protect\citeauthoryear{Mullahy}{1986}]{Mullahy86} Mullahy, J. (1986). Specification and testing of some modified count data models. {\it Journal of econometrics} 33(3), 341--365.

\bibitem[\protect\citeauthoryear{Puig et al.}{2024}]{Puig24} Puig, P., Valero, J., and Fernández‐Fontelo, A. (2024). Some mechanisms leading to underdispersion: Old and new proposals. {\it Scandinavian Journal of Statistics} 51(1), 245--267.

\bibitem[\protect\citeauthoryear{R Core Team}{2019}]{RCore} R Core Team (2019). R: A language and environment for statistical computing. R Foundation for Statistical Computing, Vienna, Austria. 

\bibitem[\protect\citeauthoryear{Ridout et al.}{1998}]{Ridout98} Ridout, M., Demétrio, C. G., and Hinde, J. (1998). Models for count data with many zeros. In {\it Proceedings of the XIXth international biometric conference}, vol. 19, no. 1, pp. 179-192. Cape Town, South Africa: International Biometric Society Invited Papers.

\bibitem[\protect\citeauthoryear{Sankaran}{1970}]{Sankaran70} Sankaran, M. (1970). 275. note: The discrete poisson-lindley distribution. {\it Biometrics} 26(1), 145--149.

\bibitem[\protect\citeauthoryear{Singh}{1963}]{Singh63} Singh, S. (1963). A note on inflated Poisson distribution. {\it Journal of the Indian statistical association} 1, 140--144.

\bibitem[\protect\citeauthoryear{Skulpakdee and Hunkrajok}{2022}]{Skul22} Skulpakdee, W. and Hunkrajok, M. (2022). Unusual-event processes for count data. {\it SORT-Statistics and Operations Research Transactions} 46(1), 39--66.


\bibitem[\protect\citeauthoryear{Su et al.}{2013}]{Su13} Su, X. , Fan, J. , Levine, R. A. , Tan, X. and Tripathi, A. (2013). Multiple-inflation Poisson model with L1 regularization. {\it Statistica Sinica} 23, 1071--1090.


\bibitem[\protect\citeauthoryear{Westfall}{2014}]{Westfall14} Westfall, P.H.  (2014). Kurtosis as peakedness, 1905–2014. RIP. {\it The American Statistician} 68(3), 191--195.

\bibitem[\protect\citeauthoryear{Winkelmann}{1995}]{Winkelmann95} Winkelmann, R. (1995). Duration dependence and dispersion in count-data models. {\it Journal of Business \& Economic Statistics} 13(4), 467--474.  

\bibitem[\protect\citeauthoryear{Winkelmann}{2008}]{Winkelmann08} Winkelmann, R. (2008). {\it Econometric Analysis of Count Data}. Heidelberg: Springer Berlin Heidelberg, Berlin.

\bibitem[\protect\citeauthoryear{Wise}{1962}]{Wise62} Wise, J. (1962). The relationship between the mean and variance of a stationary birth-death process, and its economic application. {\it Biometrika} 49(1/2), 253--255.

\bibitem[\protect\citeauthoryear{Zhang et al.}{2016}]{Zhang16} Zhang, C., Tian, G. L., and Ng, K. W. (2016). Properties of the zero-and-one inflated Poisson distribution and likelihood-based inference methods. {\it Statistics and its interface} 9(1), 11--32.

\end{thebibliography}

\begin{thebibliography}{99}
	\renewcommand{\itemsep}{0pt}
	
	\bibitem[\protect\citeauthoryear{Bickel and Doksum}{2015}]{sBickel15} Bickel, P. J. and Doksum, K. A. (2015). {\it Mathematical Statistics: Basic Ideas and Selected Topics Volume I}, 2nd Ed. Chapman and Hall/CRC.
	
	\bibitem[\protect\citeauthoryear{Sankaran}{1970}]{sSankaran70} Sankaran, M. (1970). 275. note: The discrete poisson-lindley distribution. {\it Biometrics} 26(1), 145--149.
	
\end{thebibliography}
\appendix

\newpage
\setcounter{page}{1}
\setcounter{equation}{0}
\renewcommand{\theequation}{S\arabic{equation}}
\setcounter{section}{0}
\renewcommand{\thesection}{S\arabic{section}}
\setcounter{table}{0}
\renewcommand{\thetable}{S\arabic{table}}
\setcounter{figure}{0}
\renewcommand{\thefigure}{S\arabic{figure}}
\setcounter{proposition}{0}
\renewcommand{\theproposition}{S\arabic{proposition}}

\addtocontents{toc}{\protect\fi}

\begin{landscape}
\pagestyle{empty}%

\begin{center}
	\title{\huge Supplemental material for ``Stationary birth-death processes generating inflation-deflation distributions: Avoiding the issue of dominance''}
\end{center}

\tableofcontents

\newpage
\section{Derivation of the type 1 model}
\vspace{-0.5cm}
\begin{figure}[!h]
	\begin{subfigure}{\linewidth}
		\centering
		\resizebox{\linewidth}{!}{
			\begin{tikzpicture}
				\newcommand\xr{0.5}
				\newcommand\xs{2}
				
				\newcommand\xx{0} 
				\draw[black, very thick] (\xx,0) node {$\scriptstyle 0$} circle (\xr);
				\draw[ultra thick, ->] (\xx+\xr,\xr) arc (120:60:\xs);
				\draw[ultra thick, ->] (\xx+\xs+\xr,-\xr) arc (300:240:\xs);
				\node[above] at (\xx+\xr+\xs/2,0.8) {$\scriptstyle \gamma_0$};
				\node[below] at (\xx+\xr+\xs/2,-0.8) {$\scriptstyle \mu_1$};
				\node[below] at (\xx+\xr+\xs/2-0.4,-1.3) {$\scriptstyle \lambda_0=\frac{\gamma_0}{\mu_1}$};
				\node[below] at (\xx+\xr+\xs/2+0.3,-1.9) {$\scriptstyle =\frac{\alpha_1^{\textbf{1}_{\mathcal{F}}(1)}}{\alpha_0^{\textbf{1}_{\mathcal{F}}(0)}}\lambda_0^b$};
				
				\renewcommand\xx{3}
				\draw[black, very thick] (\xx,0) node {$\scriptstyle 1$} circle (\xr);
				\draw[ultra thick, ->] (\xx+\xr,\xr) arc (120:60:\xs);
				\draw[ultra thick, ->] (\xx+\xs+\xr,-\xr) arc (300:240:\xs);
				\node[above] at (\xx+\xr+\xs/2,0.8) {$\scriptstyle \gamma_1$};
				\node[below] at (\xx+\xr+\xs/2,-0.8) {$\scriptstyle \mu_2$};
				\node[below] at (\xx+\xr+\xs/2-0.4,-1.3) {$\scriptstyle \lambda_1=\frac{\gamma_1}{\mu_2}$};
				\node[below] at (\xx+\xr+\xs/2+0.3,-1.9) {$\scriptstyle =\frac{\alpha_2^{\textbf{1}_{\mathcal{F}}(2)}}{\alpha_1^{\textbf{1}_{\mathcal{F}}(1)}}\lambda_1^b$};

				\renewcommand\xx{2*3}
				\draw[black, very thick] (\xx,0) node {$\scriptstyle 2$} circle (\xr);
				\draw[black, very thick] (\xx+1,0) node {\textbf{...}};
				
				\renewcommand\xx{3*3-1}
				\draw[black, very thick] (\xx,0) node {$\scriptstyle q-1$} circle (\xr);
				\draw[ultra thick, ->] (\xx+\xr,\xr) arc (120:60:\xs);
				\draw[ultra thick, ->] (\xx+\xs+\xr,-\xr) arc (300:240:\xs);
				\node[above] at (\xx+\xr+\xs/2,0.8) {$\scriptstyle \gamma_{q-1}$};
				\node[below] at (\xx+\xr+\xs/2,-0.8) {$\scriptstyle \mu_q$};
				\node[below] at (\xx+\xr+\xs/2-0.4,-1.3) {$\scriptstyle \lambda_{q-1}=\frac{\gamma_{q-1}}{\mu_q}$};
				\node[below] at (\xx+\xr+\xs/2+0.65,-1.9) {$\scriptstyle =\frac{\alpha_q^{\textbf{1}_{\mathcal{F}}(q)}}{\alpha_{q-1}^{\textbf{1}_{\mathcal{F}}({q-1})}}\lambda_{q-1}^b$};

				\renewcommand\xx{4*3-1}
				\draw[black, very thick] (\xx,0) node {$\scriptstyle q$} circle (\xr);
				\draw[ultra thick, ->] (\xx+\xr,\xr) arc (120:60:\xs);
				\draw[ultra thick, ->] (\xx+\xs+\xr,-\xr) arc (300:240:\xs);
				\node[above] at (\xx+\xr+\xs/2,0.8) {$\scriptstyle \gamma_{q}$};
				\node[below] at (\xx+\xr+\xs/2,-0.8) {$\scriptstyle \mu_{q+1}$};
				\node[below] at (\xx+\xr+\xs/2-0.1,-1.3) {$\scriptstyle \lambda_{q}=\frac{\gamma_{q}}{\mu_{q+1}}$};
				\node[below] at (\xx+\xr+\xs/2+0.4,-2.2) {$\scriptstyle =\frac{1}{\alpha_{q}^{\textbf{1}_{\mathcal{F}}({q})}}\lambda_{q}^b$};

				\renewcommand\xx{5*3-1}
				\draw[black, very thick] (\xx,0) node {$\scriptstyle q+1$} circle (\xr);
				\draw[ultra thick, ->] (\xx+\xr,\xr) arc (120:60:\xs);
				\draw[ultra thick, ->] (\xx+\xs+\xr,-\xr) arc (300:240:\xs);
				\node[above] at (\xx+\xr+\xs/2,0.8) {$\scriptstyle \gamma_{q+1}$};
				\node[below] at (\xx+\xr+\xs/2,-0.8) {$\scriptstyle \mu_{q+2}$};
				\node[below] at (\xx+\xr+\xs/2,-1.3) {$\scriptstyle \lambda_{q+1}=\frac{\gamma_{q+1}}{\mu_{q+2}}$};
				\node[below] at (\xx+\xr+\xs/2+0.33,-2.2) {$\scriptstyle =\lambda_{q+1}^b$};

				\renewcommand\xx{6*3-1}
				\draw[black, very thick] (\xx,0) node {$\scriptstyle q+2$} circle (\xr);
				\draw[black, very thick] (\xx+1,0) node {\textbf{...}};
			\end{tikzpicture}
		}
	\end{subfigure}
	\caption{State-transition-rate ratio diagram for type 1 stationary process}
\end{figure}
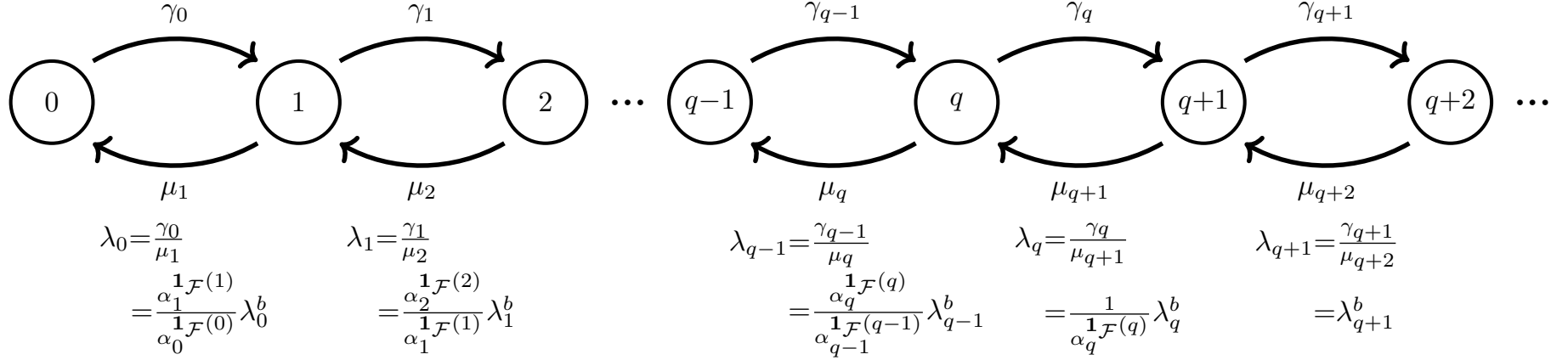
By substituting $\lambda_{n}=\left\{\begin{array}{cll}
	\displaystyle \frac{\alpha_{n+1}^{\textbf{1}_{\mathcal{F}}(n+1)}}{\alpha_n^{\textbf{1}_{\mathcal{F}}(n)}}\lambda_n^b & , & n = 0,1,...,q-1 \\
	\displaystyle \frac{1}{\alpha_q^{\textbf{1}_{\mathcal{F}}(q)}}\lambda_q^b & , & n = q \\
	\displaystyle \lambda_n^b & , & n > q, \\
\end{array}\right.$, where $\mathcal{F}=\left\{ n_0,n_1,...,n_m \right\}\subseteq \left\{ 0,1,...,q \right\} \subset \left\{ 0,1,... \right\}$
, and $\textbf{1}_{\mathcal{F}}(n)$ is the indicator function defined as $\textbf{1}_{\mathcal{F}}(n)=1$ if $n \in \mathcal{F}$ and $0$ otherwise, into \\

\vspace{0.1cm}
\noindent $p(0,\boldsymbol{\theta},\boldsymbol{\alpha})^{-1}=  1+\sum_{i=0}^{\infty}\lambda_0\lambda_1...\lambda_{i}=1+\sum_{i=0}^{\infty}\prod_{j=0}^{i}\lambda_j$,
we have

\begin{align*}
	p(0,\boldsymbol{\theta},\boldsymbol{\alpha})^{-1}
	&= 1+\sum_{i=0}^{q-1}\frac{\alpha_{i+1}^{\textbf{1}_{\mathcal{F}}(i+1)}}{\alpha_0^{\textbf{1}_{\mathcal{F}}(0)}}\prod_{j=0}^{i}\lambda_j^b+\sum_{i=q}^{\infty}\frac{1}{\alpha_0^{\textbf{1}_{\mathcal{F}}(0)}}\prod_{j=0}^{i}   \lambda_j^b\\
	&=1+\sum_{i=0}^{q-1}\frac{\alpha_{i+1}^{\textbf{1}_{\mathcal{F}}(i+1)}}{\alpha_0^{\textbf{1}_{\mathcal{F}}(0)}}\prod_{j=0}^{i} \lambda_j^b-\left( \frac{1}{\alpha_0^{\textbf{1}_{\mathcal{F}}(0)}}+\sum_{i=0}^{q-1}\frac{1}{\alpha_0^{\textbf{1}_{\mathcal{F}}(0)}}\prod_{j=0}^{i}   \lambda_j^b \right) +\left( \frac{1}{\alpha_0^{\textbf{1}_{\mathcal{F}}(0)}}+\sum_{i=0}^{q-1}\frac{1}{\alpha_0^{\textbf{1}_{\mathcal{F}}(0)}}\prod_{j=0}^{i}   \lambda_j^b \right)+\sum_{i=q}^{\infty }\frac{1}{\alpha_0^{\textbf{1}_{\mathcal{F}}(0)}}\prod_{j=0}^{i}   \lambda_j^b\\
	&=\left( 1-\frac{1}{\alpha_0^{\textbf{1}_{\mathcal{F}}(0)}} \right)+\sum_{i=0}^{q-1}\left( \frac{\alpha_{i+1}^{\textbf{1}_{\mathcal{F}}(i+1)}}{\alpha_0^{\textbf{1}_{\mathcal{F}}(0)}}-\frac{1}{\alpha_0^{\textbf{1}_{\mathcal{F}}(0)}} \right)\prod_{j=0}^{i} \lambda_j^b+\frac{1}{\alpha_0^{\textbf{1}_{\mathcal{F}}(0)}}\left( 1+\sum_{i=0}^{\infty }\prod_{j=0}^{i} \lambda_j^b \right)\\
	&=\frac{1}{\alpha_0^{\textbf{1}_{\mathcal{F}}(0)}}\left( \left( \alpha_0^{\textbf{1}_{\mathcal{F}}(0)} -1\right)+\sum_{i=0}^{q-1}\left( \alpha_{i+1}^{\textbf{1}_{\mathcal{F}}(i+1)}-1 \right)\prod_{j=0}^{i} \lambda_j^b+z_b \left ( \boldsymbol{\theta} \right )  \right)\\
	&=\frac{z_b \left ( \boldsymbol{\theta} \right )}{\alpha_0^{\textbf{1}_{\mathcal{F}}(0)}}\left( \left( \alpha_0^{\textbf{1}_{\mathcal{F}}(0)} -1\right)\frac{1}{z_b \left ( \boldsymbol{\theta} \right ) }+\sum_{i=0}^{q-1}\left( \alpha_{i+1}^{\textbf{1}_{\mathcal{F}}(i+1)}-1 \right)\frac{\prod_{j=0}^{i} \lambda_j^b}{z_b \left ( \boldsymbol{\theta} \right ) }+1  \right)\\
	&=\frac{z_b \left ( \boldsymbol{\theta} \right )}{\alpha_0^{\textbf{1}_{\mathcal{F}}(0)}}\left( 1+\left( \alpha_0^{\textbf{1}_{\mathcal{F}}(0)} -1\right)b\left ( 0, \boldsymbol{\theta} \right )+\sum_{i=0}^{q-1}\left( \alpha_{i+1}^{\textbf{1}_{\mathcal{F}}(i+1)}-1 \right)b\left ( i+1, \boldsymbol{\theta} \right ) \right)\\
	&=\frac{z_b \left ( \boldsymbol{\theta} \right )}{\alpha_0^{\textbf{1}_{\mathcal{F}}(0)}}\left( 1+\left( \alpha_0^{\textbf{1}_{\mathcal{F}}(0)} -1\right)b\left ( 0, \boldsymbol{\theta} \right )+\sum_{i=1}^{q}\left( \alpha_{i}^{\textbf{1}_{\mathcal{F}}(i)}-1 \right)b\left (i, \boldsymbol{\theta} \right ) \right)\\
	&=\frac{z_b \left ( \boldsymbol{\theta} \right )}{\alpha_0^{\textbf{1}_{\mathcal{F}}(0)}}\left( 1+\sum_{i=0}^{q}\left( \alpha_{i}^{\textbf{1}_{\mathcal{F}}(i)}-1 \right)b\left (i, \boldsymbol{\theta} \right ) \right)\\ &=\frac{1}{\alpha_0^{\textbf{1}_{\mathcal{F}}(0)} b\left ( 0, \boldsymbol{\theta} \right )}\left( 1+\sum_{i=0}^{m}\left( \alpha_{n_i}-1 \right)b\left (n_i, \boldsymbol{\theta} \right ) \right).
\end{align*}
\begin{align*}
	p(0,\boldsymbol{\theta},\boldsymbol{\alpha})&=\frac{\alpha_0^{\textbf{1}_{\mathcal{F}}(0)} b\left ( 0, \boldsymbol{\theta} \right )}{z\left (\boldsymbol{\theta}, \boldsymbol{\alpha} \right )}.
\end{align*}

\noindent Substituting $\lambda_n$ and $p(0,\boldsymbol{\theta},\boldsymbol{\alpha})$ into $p(n,\boldsymbol{\theta},\boldsymbol{\alpha})=\lambda_0\lambda_1...\lambda_{n-1}p(0,\boldsymbol{\theta},\boldsymbol{\alpha})$
gives
\begin{align*}
	p(n,\boldsymbol{\theta},\boldsymbol{\alpha})
	&=\left\{\begin{array}{cll}
		\displaystyle \frac{\alpha_0^{\textbf{1}_{\mathcal{F}}(0)} b\left ( 0, \boldsymbol{\theta} \right )}{z\left (\boldsymbol{\theta}, \boldsymbol{\alpha} \right )} & , & n = 0 \\
		\displaystyle \frac{\alpha_0^{\textbf{1}_{\mathcal{F}}(0)} \alpha_1^{\textbf{1}_{\mathcal{F}}(1)} \cdots \alpha_{n-1}^{\textbf{1}_{\mathcal{F}}(n-1)} \alpha_{n}^{\textbf{1}_{\mathcal{F}}(n)} }{\alpha_0^{\textbf{1}_{\mathcal{F}}(0)} \alpha_1^{\textbf{1}_{\mathcal{F}}(1)} \cdots \alpha_{n-1}^{\textbf{1}_{\mathcal{F}}(n-1)}} \frac{\lambda_0^b \lambda_1^b ... \lambda_{n-1}^b b\left ( 0, \boldsymbol{\theta} \right )}{z\left (\boldsymbol{\theta}, \boldsymbol{\alpha} \right )}  & , & n = 1,...,q \\
		\displaystyle \frac{\alpha_0^{\textbf{1}_{\mathcal{F}}(0)} \alpha_1^{\textbf{1}_{\mathcal{F}}(1)} \cdots \alpha_{q}^{\textbf{1}_{\mathcal{F}}(q)} }{\alpha_0^{\textbf{1}_{\mathcal{F}}(0)} \alpha_1^{\textbf{1}_{\mathcal{F}}(1)} \cdots \alpha_{q}^{\textbf{1}_{\mathcal{F}}(q)}} \frac{\lambda_0^b \lambda_1^b ... \lambda_{q}^b b\left ( 0, \boldsymbol{\theta} \right )}{z\left (\boldsymbol{\theta}, \boldsymbol{\alpha} \right )} & , & n = q+1  \\
		\displaystyle \frac{\alpha_0^{\textbf{1}_{\mathcal{F}}(0)} \alpha_1^{\textbf{1}_{\mathcal{F}}(1)} \cdots \alpha_{q}^{\textbf{1}_{\mathcal{F}}(q)} }{\alpha_0^{\textbf{1}_{\mathcal{F}}(0)} \alpha_1^{\textbf{1}_{\mathcal{F}}(1)} \cdots \alpha_{q}^{\textbf{1}_{\mathcal{F}}(q)}} \frac{\lambda_0^b \lambda_1^b ... \lambda_{n-1}^b b\left ( 0, \boldsymbol{\theta} \right )}{z\left (\boldsymbol{\theta}, \boldsymbol{\alpha} \right )}  & , &n>q+1  \\
	\end{array}\right.\\
	&=\left\{\begin{array}{cll}
		\displaystyle \frac{\alpha_0^{\textbf{1}_{\mathcal{F}}(0)} b\left ( 0, \boldsymbol{\theta} \right )}{z\left (\boldsymbol{\theta}, \boldsymbol{\alpha} \right )} & , & n = 0 \\
		\displaystyle \frac{\alpha_{n}^{\textbf{1}_{\mathcal{F}}(n)} b(n,\boldsymbol{\theta})}{z(\boldsymbol{\theta},\boldsymbol{\alpha})} & , & n = 1,...,q \\
		\displaystyle \frac{ b(n,\boldsymbol{\theta})}{z(\boldsymbol{\theta},\boldsymbol{\alpha})} & , & n = q+1  \\
		\displaystyle \frac{ b(n,\boldsymbol{\theta})}{z(\boldsymbol{\theta},\boldsymbol{\alpha})} & , &n>q+1  \\
	\end{array}\right.\\
	&= \displaystyle \frac{\alpha_{n}^{\textbf{1}_{\mathcal{F}}(n)} b(n,\boldsymbol{\theta})}{z(\boldsymbol{\theta},\boldsymbol{\alpha})} = \frac{ \prod_{i=0}^{m}\alpha_{n_i}^{\textbf{1}_{n_i}(n)} b(n,\boldsymbol{\theta})}{z(\boldsymbol{\theta},\boldsymbol{\alpha})}.
\end{align*}
\noindent This completes the derivation.

\section{Derivation of the type 2 model}
\vspace{-0.5cm}
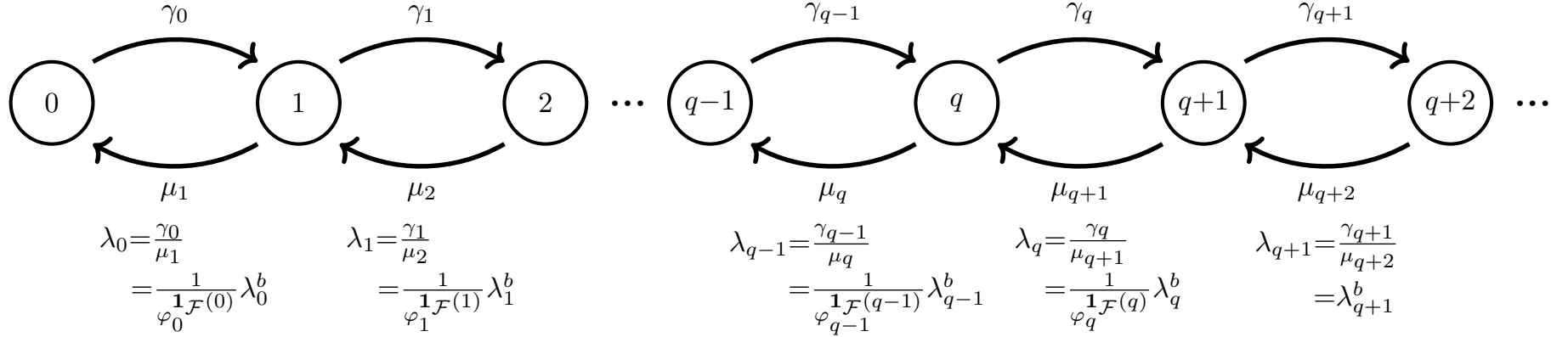
\begin{figure}[!h]
	\begin{subfigure}{\linewidth}
		\centering
		\resizebox{\linewidth}{!}{
			\begin{tikzpicture}
				\newcommand\xr{0.5}
				\newcommand\xs{2}

				\newcommand\xx{0} 
				\draw[black, very thick] (\xx,0) node {$\scriptstyle 0$} circle (\xr);
				\draw[ultra thick, ->] (\xx+\xr,\xr) arc (120:60:\xs);
				\draw[ultra thick, ->] (\xx+\xs+\xr,-\xr) arc (300:240:\xs);
				\node[above] at (\xx+\xr+\xs/2,0.8) {$\scriptstyle \gamma_0$};
				\node[below] at (\xx+\xr+\xs/2,-0.8) {$\scriptstyle \mu_1$};
				\node[below] at (\xx+\xr+\xs/2-0.4,-1.3) {$\scriptstyle \lambda_0=\frac{\gamma_0}{\mu_1}$};
				\node[below] at (\xx+\xr+\xs/2+0.3,-1.9) {$\scriptstyle =\frac{1}{\varphi_0^{\textbf{1}_{\mathcal{F}}({0})}}\lambda_0^b$};
				
				\renewcommand\xx{3}
				\draw[black, very thick] (\xx,0) node {$\scriptstyle 1$} circle (\xr);
				\draw[ultra thick, ->] (\xx+\xr,\xr) arc (120:60:\xs);
				\draw[ultra thick, ->] (\xx+\xs+\xr,-\xr) arc (300:240:\xs);
				\node[above] at (\xx+\xr+\xs/2,0.8) {$\scriptstyle \gamma_1$};
				\node[below] at (\xx+\xr+\xs/2,-0.8) {$\scriptstyle \mu_2$};
				\node[below] at (\xx+\xr+\xs/2-0.4,-1.3) {$\scriptstyle \lambda_1=\frac{\gamma_1}{\mu_2}$};
				\node[below] at (\xx+\xr+\xs/2+0.3,-1.9) {$\scriptstyle =\frac{1}{\varphi_1^{\textbf{1}_{\mathcal{F}}({1})}}\lambda_1^b$};
				
				\renewcommand\xx{2*3}
				\draw[black, very thick] (\xx,0) node {$\scriptstyle 2$} circle (\xr);
				\draw[black, very thick] (\xx+1,0) node {\textbf{...}};
				
				\renewcommand\xx{3*3-1}
				\draw[black, very thick] (\xx,0) node {$\scriptstyle q-1$} circle (\xr);
				\draw[ultra thick, ->] (\xx+\xr,\xr) arc (120:60:\xs);
				\draw[ultra thick, ->] (\xx+\xs+\xr,-\xr) arc (300:240:\xs);
				\node[above] at (\xx+\xr+\xs/2,0.8) {$\scriptstyle \gamma_{q-1}$};
				\node[below] at (\xx+\xr+\xs/2,-0.8) {$\scriptstyle \mu_{q}$};
				\node[below] at (\xx+\xr+\xs/2-0.4,-1.3) {$\scriptstyle \lambda_{q-1}=\frac{\gamma_{q-1}}{\mu_{q}}$};
				\node[below] at (\xx+\xr+\xs/2+0.65,-1.9) {$\scriptstyle =\frac{1}{\varphi_{q-1}^{\textbf{1}_{\mathcal{F}}({q-1})}}\lambda_{q-1}^b$};
				
				\renewcommand\xx{4*3-1}
				\draw[black, very thick] (\xx,0) node {$\scriptstyle q$} circle (\xr);
				\draw[ultra thick, ->] (\xx+\xr,\xr) arc (120:60:\xs);
				\draw[ultra thick, ->] (\xx+\xs+\xr,-\xr) arc (300:240:\xs);
				\node[above] at (\xx+\xr+\xs/2,0.8) {$\scriptstyle \gamma_{q}$};
				\node[below] at (\xx+\xr+\xs/2,-0.8) {$\scriptstyle \mu_{q+1}$};
				\node[below] at (\xx+\xr+\xs/2-0.1,-1.3) {$\scriptstyle \lambda_{q}=\frac{\gamma_{q}}{\mu_{q+1}}$};
				\node[below] at (\xx+\xr+\xs/2+0.4,-1.9) {$\scriptstyle =\frac{1}{\varphi_{q}^{\textbf{1}_{\mathcal{F}}({q})}}\lambda_{q}^b$};
				
				\renewcommand\xx{5*3-1}
				\draw[black, very thick] (\xx,0) node {$\scriptstyle q+1$} circle (\xr);
				\draw[ultra thick, ->] (\xx+\xr,\xr) arc (120:60:\xs);
				\draw[ultra thick, ->] (\xx+\xs+\xr,-\xr) arc (300:240:\xs);
				\node[above] at (\xx+\xr+\xs/2,0.8) {$\scriptstyle \gamma_{q+1}$};
				\node[below] at (\xx+\xr+\xs/2,-0.8) {$\scriptstyle \mu_{q+2}$};
				\node[below] at (\xx+\xr+\xs/2,-1.3) {$\scriptstyle \lambda_{q+1}=\frac{\gamma_{q+1}}{\mu_{q+2}}$};
				\node[below] at (\xx+\xr+\xs/2+0.33,-2.0) {$\scriptstyle =\lambda_{q+1}^b$};
				
				\renewcommand\xx{6*3-1}
				\draw[black, very thick] (\xx,0) node {$\scriptstyle q+2$} circle (\xr);
				\draw[black, very thick] (\xx+1,0) node {\textbf{...}};
			\end{tikzpicture}
		}
	\end{subfigure}
	\caption{State-transition-rate ratio diagram for type 2 stationary process}. 
\end{figure}
By substituting $\lambda_{n}=\left\{\begin{array}{cll}
	\displaystyle \frac{1}{\varphi_n^{\textbf{1}_{\mathcal{F}}({n})}}\lambda_n^b & , & n = 0,1,...,q \\
	\displaystyle \lambda_n^b & , & n>q, \\
\end{array}\right.$, where $\mathcal{F}=\left\{ n_0,n_1,...,n_m \right\}\subseteq \left\{ 0,1,...,q \right\} \subset \left\{ 0,1,... \right\}$
, and $\textbf{1}_{\mathcal{F}}(n)$ is the indicator function defined as $\textbf{1}_{\mathcal{F}}(n)=1$ if $n \in \mathcal{F}$ and $0$ otherwise, into\\

\vspace{0.1cm}
\noindent $p(0,\boldsymbol{\theta},\boldsymbol{\varphi})^{-1} =  1+\sum_{i=0}^{\infty}\lambda_0\lambda_1...\lambda_{i}=1+\sum_{i=0}^{\infty}\prod_{j=0}^{i}\lambda_j$,
we have

\begin{align*}
	p(0,\boldsymbol{\theta},\boldsymbol{\varphi})^{-1}
	&= 1+\sum_{i=0}^{q-1}\frac{1}{\prod_{j=0}^{i}\varphi _{j}^{\textbf{1}_{\mathcal{F}}({j})}} \prod_{j=0}^{i}\lambda_j^b+\sum_{i=q}^{\infty}\frac{1}{\prod_{j=0}^{q}\varphi _{j}^{\textbf{1}_{\mathcal{F}}({j})}} \prod_{j=0}^{i}   \lambda_j^b\\
\end{align*}

\begin{align*}
	&= 1+\sum_{i=0}^{q-1}\frac{1}{\prod_{j=0}^{i}\varphi _{j}^{\textbf{1}_{\mathcal{F}}({j})}} \prod_{j=0}^{i}\lambda_j^b -\left( \frac{1}{\prod_{j=0}^{q}\varphi _{j}^{\textbf{1}_{\mathcal{F}}({j})}}+\sum_{i=0}^{q-1}\frac{1}{\prod_{j=0}^{q}\varphi _{j}^{\textbf{1}_{\mathcal{F}}({j})}}\prod_{j=0}^{i}   \lambda_j^b \right) + \left( \frac{1}{\prod_{j=0}^{q}\varphi _{j}^{\textbf{1}_{\mathcal{F}}({j})}}+\sum_{i=0}^{q-1}\frac{1}{\prod_{j=0}^{q}\varphi _{j}^{\textbf{1}_{\mathcal{F}}({j})}}\prod_{j=0}^{i}   \lambda_j^b \right)\\ 
	&\quad +\sum_{i=q}^{\infty}\frac{1}{\prod_{j=0}^{q}\varphi _{j}^{\textbf{1}_{\mathcal{F}}({j})}} \prod_{j=0}^{i}   \lambda_j^b  \\ 
	&= \left ( 1- \frac{1}{\prod_{j=0}^{q}\varphi _{j}^{\textbf{1}_{\mathcal{F}}({j})}} \right ) 
	+ \sum_{i=0}^{q-1} \left ( \frac{1}{\prod_{j=0}^{i}\varphi _{j}^{\textbf{1}_{\mathcal{F}}({j})}}-\frac{1}{\prod_{j=0}^{q}\varphi _{j}^{\textbf{1}_{\mathcal{F}}({j})}} \right )  \prod_{j=0}^{i}\lambda_j^b  + \frac{1}{\prod_{j=0}^{q}\varphi _{j}^{\textbf{1}_{\mathcal{F}}({j})}} \left ( 1+ \sum_{i=0}^{\infty}\prod_{j=0}^{i}\lambda_j^b \right ) \\
	&= \frac{1}{\prod_{j=0}^{q}\varphi _{j}^{\textbf{1}_{\mathcal{F}}({j})}} \left ( \left (\prod_{j=0}^{q}\varphi _{j}^{\textbf{1}_{\mathcal{F}}({j})}- 1 \right ) 
	+ \sum_{i=0}^{q-1} \left ( \prod_{j=i+1}^{q}\varphi _{j}^{\textbf{1}_{\mathcal{F}}({j})}- 1 \right )  \prod_{j=0}^{i}\lambda_j^b  + z_b\left ( \boldsymbol{\theta } \right )  \right ) \\
	&= \frac{z_b\left ( \boldsymbol{\theta } \right ) }{\prod_{j=0}^{q}\varphi _{j}^{\textbf{1}_{\mathcal{F}}({j})}} \left ( \left (\prod_{j=0}^{q}\varphi _{j}^{\textbf{1}_{\mathcal{F}}({j})}- 1 \right ) \frac{1}{z_b\left ( \boldsymbol{\theta } \right )}
	+ \sum_{i=0}^{q-1} \left ( \prod_{j=i+1}^{q}\varphi _{j}^{\textbf{1}_{\mathcal{F}}({j})}- 1 \right )  \frac{\prod_{j=0}^{i}\lambda_j^b}{z_b\left ( \boldsymbol{\theta } \right )}  + 1 \right ) \\
	&= \frac{z_b\left ( \boldsymbol{\theta } \right ) }{\prod_{j=0}^{q}\varphi _{j}^{\textbf{1}_{\mathcal{F}}({j})}} \left ( 1+ \left (\prod_{j=0}^{q}\varphi _{j}^{\textbf{1}_{\mathcal{F}}({j})}- 1 \right ) b\left ( 0,\boldsymbol{\theta } \right )
	+ \sum_{i=0}^{q-1} \left ( \prod_{j=i+1}^{q}\varphi _{j}^{\textbf{1}_{\mathcal{F}}({j})}- 1 \right )  b\left ( i+1,\boldsymbol{\theta } \right )  \right )  \\
	&= \frac{z_b\left ( \boldsymbol{\theta } \right ) }{\prod_{j=0}^{q}\varphi _{j}^{\textbf{1}_{\mathcal{F}}({j})}} \left ( 1+ \left (\prod_{j=0}^{q}\varphi _{j}^{\textbf{1}_{\mathcal{F}}({j})}- 1 \right ) b\left ( 0,\boldsymbol{\theta } \right )
	+ \sum_{i=1}^{q} \left ( \prod_{j=i}^{q}\varphi _{j}^{\textbf{1}_{\mathcal{F}}({j})}- 1 \right )  b\left ( i,\boldsymbol{\theta } \right )  \right )  \\
	&= \frac{z_b\left ( \boldsymbol{\theta } \right ) }{\prod_{j=0}^{q}\varphi _{j}^{\textbf{1}_{\mathcal{F}}({j})}} \left ( 1
	+ \sum_{i=0}^{q} \left ( \prod_{j=i}^{q}\varphi _{j}^{\textbf{1}_{\mathcal{F}}({j})}- 1 \right )  b\left ( i,\boldsymbol{\theta } \right )  \right ) \\
\end{align*}

\begin{align*}
	&= \frac{z_b\left ( \boldsymbol{\theta } \right ) }{\prod_{j=0}^{q}\varphi _{j}^{\textbf{1}_{\mathcal{F}}({j})}} \left ( 1
	+ \sum_{k=0}^{n_0} \left ( \prod_{j=0}^{m}\varphi _{n_j}- 1 \right )  b\left ( k,\boldsymbol{\theta } \right )
	+ \sum_{k=n_0+1}^{n_1} \left ( \prod_{j=1}^{m}\varphi _{n_j}- 1 \right )  b\left ( k,\boldsymbol{\theta } \right )
	+...+ \sum_{k=n_{m-1}+1}^{n_m} \left ( \prod_{j=m}^{m}\varphi _{n_j}- 1 \right )  b\left ( k,\boldsymbol{\theta } \right )
	\right )  \\
	&=\frac{z_b\left ( \boldsymbol{\theta } \right ) }{\prod_{j=0}^{q}\varphi _{j}^{\textbf{1}_{\mathcal{F}}({j})}} \left ( 1
	+  \left ( \prod_{j=0}^{m}\varphi _{n_j}- 1 \right ) \sum_{k=n_{-1}+1}^{n_0} b\left ( k,\boldsymbol{\theta } \right )
	+ \left ( \prod_{j=1}^{m}\varphi _{n_j}- 1 \right ) \sum_{k=n_0+1}^{n_1} b\left ( k,\boldsymbol{\theta } \right )
	+... + \left ( \prod_{j=m}^{m}\varphi _{n_j}- 1 \right ) \sum_{k=n_{m-1}+1}^{n_m} b\left ( k,\boldsymbol{\theta } \right )
	\right ) \\
	&= \frac{1}{\prod_{j=0}^{q}\varphi _{j}^{\textbf{1}_{\mathcal{F}}({j})} b\left ( 0,\boldsymbol{\theta } \right )} \left ( 1 
	+ \sum_{i=0}^{m} \left ( \prod_{j=i}^{m}\varphi_{n_j}-1 \right ) \sum_{k=n_{i-1}+1}^{n_i} b\left ( k,\boldsymbol{\theta } \right )  \right ). \\
\end{align*}
\begin{equation*}
	\displaystyle p(0,\boldsymbol{\theta},\boldsymbol{\varphi}) =\frac{\prod_{j=0}^{q}\varphi _{j}^{\textbf{1}_{\mathcal{F}}({j})} b\left ( 0,\boldsymbol{\theta} \right )}{z\left ( \boldsymbol{\theta},\boldsymbol{\varphi } \right ) }. 
\end{equation*}

\vspace{0.5cm}
\noindent Substituting $\lambda_n$ and $p(0,\boldsymbol{\theta},\boldsymbol{\varphi})$ into $p(n,\boldsymbol{\theta},\boldsymbol{\varphi})=\lambda_0\lambda_1...\lambda_{n-1}p(0,\boldsymbol{\theta},\boldsymbol{\varphi})$
gives

\begin{align*}
	p(n,\boldsymbol{\theta},\boldsymbol{\varphi})
	&=\left\{\begin{array}{cll}
		\displaystyle \frac{\prod_{j=0}^{q}\varphi _{j}^{\textbf{1}_{\mathcal{F}}({j})} b\left ( 0,\boldsymbol{\theta} \right )}{z\left ( \boldsymbol{\theta},\boldsymbol{\varphi } \right ) } & , & n = 0 \\
		\displaystyle \frac{\prod_{j=0}^{q}\varphi _{j}^{\textbf{1}_{\mathcal{F}}({j})} }{\varphi _{0}^{\textbf{1}_{\mathcal{F}}({0})}\varphi _{1}^{\textbf{1}_{\mathcal{F}}({1})}...\varphi _{n-1}^{\textbf{1}_{\mathcal{F}}({n-1})}}\frac{\lambda_0^b\lambda_1^b...\lambda_{n-1}^b b\left ( 0,\boldsymbol{\theta} \right )}{z\left ( \boldsymbol{\theta},\boldsymbol{\varphi } \right ) } \displaystyle  & , & n = 1,2,...,q\\
		\displaystyle \frac{\prod_{j=0}^{q}\varphi _{j}^{\textbf{1}_{\mathcal{F}}({j})} }{\varphi _{0}^{\textbf{1}_{\mathcal{F}}({0})}\varphi _{1}^{\textbf{1}_{\mathcal{F}}({1})}...\varphi _{q}^{\textbf{1}_{\mathcal{F}}({q})}}\frac{\lambda_0^b\lambda_1^b...\lambda_{q}^b b\left ( 0,\boldsymbol{\theta} \right )}{z\left ( \boldsymbol{\theta},\boldsymbol{\varphi } \right ) }\displaystyle  & , & n = q+1  \\
		\displaystyle \frac{\prod_{j=0}^{q}\varphi _{j}^{\textbf{1}_{\mathcal{F}}({j})} }{\varphi _{0}^{\textbf{1}_{\mathcal{F}}({0})}\varphi _{1}^{\textbf{1}_{\mathcal{F}}({1})}...\varphi _{q}^{\textbf{1}_{\mathcal{F}}({q})}}\frac{\lambda_0^b\lambda_1^b...\lambda_{n}^b b\left ( 0,\boldsymbol{\theta} \right )}{z\left ( \boldsymbol{\theta},\boldsymbol{\varphi } \right ) }\displaystyle  & , & n>q+1 \\
	\end{array}\right.\\
	&=\left\{\begin{array}{cll}
		\displaystyle \frac{\prod_{j=0}^{q}\varphi _{j}^{\textbf{1}_{\mathcal{F}}({j})} b\left ( 0,\boldsymbol{\theta} \right )}{z\left ( \boldsymbol{\theta},\boldsymbol{\varphi } \right ) } & , & n = 0 \\
		\displaystyle \frac{\prod_{j=n}^{q}\varphi _{j}^{\textbf{1}_{\mathcal{F}}({j})} b\left ( 0,\boldsymbol{\theta} \right )}{z\left ( \boldsymbol{\theta},\boldsymbol{\varphi } \right ) }  & , & n=1,2,...,q\\
		\displaystyle \frac{ b\left ( n,\boldsymbol{\theta} \right )}{z\left ( \boldsymbol{\theta},\boldsymbol{\varphi } \right ) } & , & n > q \\
	\end{array}\right.\\
	&= \frac{\prod_{j=n}^{q}\varphi _{j}^{\textbf{1}_{\mathcal{F}}({j})} b\left (n,\boldsymbol{\theta} \right )}{z\left ( \boldsymbol{\theta},\boldsymbol{\varphi } \right ) }\\
	&= \frac{\prod_{i=0}^{m}\varphi_{n_i}^{u_{n_i}(n)} b(n,\boldsymbol{\theta})}{z(\boldsymbol{\theta},\boldsymbol{\varphi})}.
\end{align*}

\noindent This completes the derivation.

\section{Derivations of the PL ratio sequence and its probability mass function}

\vspace{-0.5cm}
\begin{figure}[!h]
	\begin{subfigure}{\linewidth}
		\centering
		\resizebox{\linewidth}{!}{
			\begin{tikzpicture}
				\newcommand\xr{0.5}
				\newcommand\xs{2}

				\newcommand\xx{0} 
				\draw[black, very thick] (\xx,0) node {$\scriptstyle 0$} circle (\xr);
				\draw[ultra thick, ->] (\xx+\xr,\xr) arc (120:60:\xs);
				\draw[ultra thick, ->] (\xx+\xs+\xr,-\xr) arc (300:240:\xs);
				\node[above] at (\xx+\xr+\xs/2,0.8) {$\scriptstyle \gamma_0$};
				\node[below] at (\xx+\xr+\xs/2,-0.8) {$\scriptstyle \mu_1$};
				\node[below] at (\xx+\xr+\xs/2-0.25,-1.4) {$\scriptstyle \lambda_0=\frac{\gamma_0}{\mu_1}$};
				\node[below] at (\xx+\xr+\xs/2+0.3,-2) {$\scriptstyle =\frac{\left ( 1+2\lambda \right )}{\left ( 1+\lambda \right )}\lambda$};
				
				\renewcommand\xx{3}
				\draw[black, very thick] (\xx,0) node {$\scriptstyle 1$} circle (\xr);
				\draw[ultra thick, ->] (\xx+\xr,\xr) arc (120:60:\xs);
				\draw[ultra thick, ->] (\xx+\xs+\xr,-\xr) arc (300:240:\xs);
				\node[above] at (\xx+\xr+\xs/2,0.8) {$\scriptstyle \gamma_1$};
				\node[below] at (\xx+\xr+\xs/2,-0.8) {$\scriptstyle \mu_2$};
				\node[below] at (\xx+\xr+\xs/2-0.25,-1.4) {$\scriptstyle \lambda_1=\frac{\gamma_1}{\mu_2}$};
				\node[below] at (\xx+\xr+\xs/2+0.3,-2) {$\scriptstyle =\frac{\left ( 1+3\lambda \right )}{\left ( 1+2\lambda \right )}\lambda$};
				
				\renewcommand\xx{2*3}
				\draw[black, very thick] (\xx,0) node {$\scriptstyle 2$} circle (\xr);
				\draw[black, very thick] (\xx+1,0) node {\textbf{...}};
				
				\renewcommand\xx{3*3-1}
				\draw[black, very thick] (\xx,0) node {$\scriptstyle q-1$} circle (\xr);
				\draw[ultra thick, ->] (\xx+\xr,\xr) arc (120:60:\xs);
				\draw[ultra thick, ->] (\xx+\xs+\xr,-\xr) arc (300:240:\xs);
				\node[above] at (\xx+\xr+\xs/2,0.8) {$\scriptstyle \gamma_{q-1}$};
				\node[below] at (\xx+\xr+\xs/2,-0.8) {$\scriptstyle \mu_{q}$};
				\node[below] at (\xx+\xr+\xs/2-0.4,-1.4) {$\scriptstyle \lambda_{q-1}=\frac{\gamma_{q-1}}{\mu_{q}}$};
				\node[below] at (\xx+\xr+\xs/2+0.45,-2) {$\scriptstyle =\frac{\left ( 1+\left ( q+1 \right )\lambda \right ) }{\left ( 1+q\lambda \right )}\lambda$};
				
				\renewcommand\xx{4*3-1}
				\draw[black, very thick] (\xx,0) node {$\scriptstyle q$} circle (\xr);
				\draw[ultra thick, ->] (\xx+\xr,\xr) arc (120:60:\xs);
				\draw[ultra thick, ->] (\xx+\xs+\xr,-\xr) arc (300:240:\xs);
				\node[above] at (\xx+\xr+\xs/2,0.8) {$\scriptstyle \gamma_{q}$};
				\node[below] at (\xx+\xr+\xs/2,-0.8) {$\scriptstyle \mu_{q+1}$};
				\node[below] at (\xx+\xr+\xs/2-0.3,-1.4) {$\scriptstyle \lambda_{q}=\frac{\gamma_{q}}{\mu_{q+1}}$};
				\node[below] at (\xx+\xr+\xs/2+0.35,-2) {$\scriptstyle =\frac{\left ( 1+\left ( q+2 \right )\lambda \right )}{\left ( 1+\left ( q+1 \right )\lambda \right )}\lambda$};
				
				\renewcommand\xx{5*3-1}
				\draw[black, very thick] (\xx,0) node {$\scriptstyle q+1$} circle (\xr);
				\draw[ultra thick, ->] (\xx+\xr,\xr) arc (120:60:\xs);
				\draw[ultra thick, ->] (\xx+\xs+\xr,-\xr) arc (300:240:\xs);
				\node[above] at (\xx+\xr+\xs/2,0.8) {$\scriptstyle \gamma_{q+1}$};
				\node[below] at (\xx+\xr+\xs/2,-0.8) {$\scriptstyle \mu_{q+2}$};
				\node[below] at (\xx+\xr+\xs/2-0.3,-1.4) {$\scriptstyle \lambda_{q+1}=\frac{\gamma_{q+1}}{\mu_{q+2}}$};
				\node[below] at (\xx+\xr+\xs/2+0.55,-2.0) {$\scriptstyle =\frac{\left ( 1+\left ( q+3 \right )\lambda \right )}{\left ( 1+\left ( q+2 \right )\lambda \right )}\lambda$};
				
				\renewcommand\xx{6*3-1}
				\draw[black, very thick] (\xx,0) node {$\scriptstyle q+2$} circle (\xr);
				\draw[black, very thick] (\xx+1,0) node {\textbf{...}};
			\end{tikzpicture}
		}
	\end{subfigure}
	\caption{State-transition-rate ratio diagram for the PL distribution}. 
\end{figure}
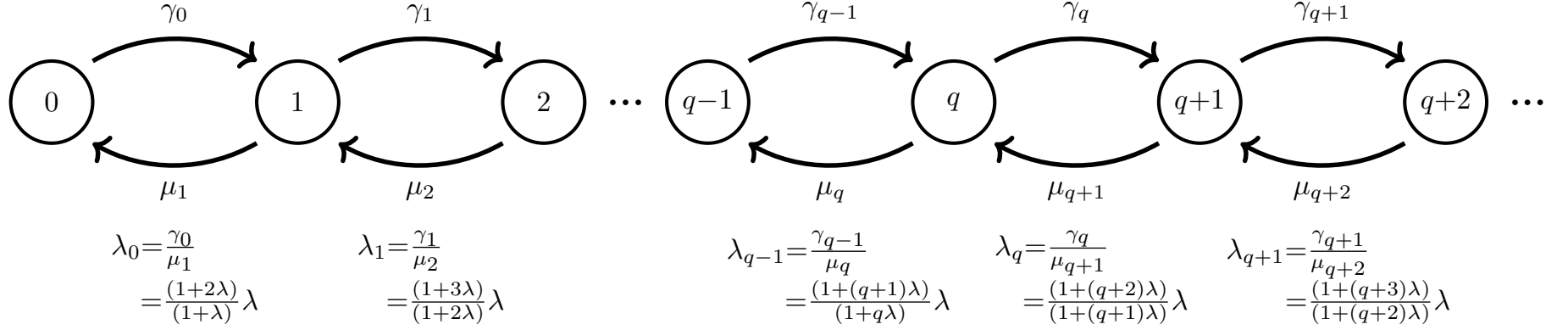
\citet[p.~145]{sSankaran70} proposes the PL distribution 
\begin{equation*}
	p\left ( n,\phi  \right )=\frac{\phi ^{2}\left ( \phi+2+n \right )}{\left ( \phi+1 \right )^{n+3}},
\end{equation*}
where $\phi >0$. We know that 
\begin{equation*}
	\lambda_n = \frac{p\left ( n+1,\phi  \right )}{p\left ( n,\phi  \right )} = \frac{\phi ^{2}\left ( \phi+2+n+1 \right )}{\left ( \phi+1 \right )^{n+4}} \frac{\left ( \phi+1 \right )^{n+3}}{\phi ^{2}\left ( \phi+2+n \right )} = \frac{\left ( \phi+3+n \right )}{\left ( \phi+1 \right )\left ( \phi+2+n \right )}. 
\end{equation*}
Let $\phi=\left ( \mu-\gamma\right )/\gamma$. Then
\begin{equation*}
	\lambda_n = \frac{\left ( \frac{\mu-\gamma}{\gamma}+3+n \right )}{\left ( \frac{\mu-\gamma}{\gamma}+1 \right )\left ( \frac{\mu-\gamma}{\gamma}+2+n \right )}= \frac{\left ( \mu-\gamma+3\gamma+n\gamma \right )\gamma }{ \left ( \mu-\gamma+\gamma \right )\left ( \mu-\gamma+2\gamma+n\gamma \right )} = \frac{\left ( \mu+2\gamma+n\gamma \right )\gamma }{ \left ( \mu+\gamma+n\gamma \right )\mu}.
\end{equation*}
Let $\lambda=\gamma/\mu$. Then
\begin{equation*}
	\lambda_n = \frac{\left ( 1+2\frac{\gamma}{\mu}+n\frac{\gamma}{\mu}\right )}{ \left ( 1+\frac{\gamma}{\mu}+n\frac{\gamma}{\mu} \right )} \frac{\gamma}{\mu} = \frac{\left ( 1+2\lambda+n\lambda\right )}{ \left ( 1+\lambda+n\lambda \right )} \lambda.
\end{equation*}
The PL distribution exists if $\lambda<1$ because $\lim_{n \to \infty} \lambda_n=\lambda$. By substituting $\lambda_n$ into $p\left ( 0,\lambda\right )^{-1}=1+\sum_{i=0}^{\infty}\lambda_0\lambda_1...\lambda_{i}$, we have
\begin{align*}
	p\left ( 0,\lambda  \right )^{-1}
	&= 1 + \frac{\left ( 1+2\lambda \right )}{ \left ( 1+\lambda \right )} \lambda + \frac{\left ( 1+3\lambda \right )}{ \left ( 1+\lambda \right )} \lambda^2 + \frac{\left ( 1+4\lambda \right )}{ \left ( 1+\lambda \right )} \lambda^3 + \frac{\left ( 1+5\lambda \right )}{ \left ( 1+\lambda \right )} \lambda^4 + ...\\
	&= \frac{\left ( 1+\lambda \right ) + \left ( 1+2\lambda \right )\lambda + \left ( 1+3\lambda \right )\lambda^2 + \left ( 1+4\lambda \right )\lambda^3 + \left ( 1+5\lambda \right )\lambda^4 + ... }{\left ( 1+\lambda \right )}\\
	&= \frac{1+\lambda +\lambda +2\lambda^2 + \lambda^2 + 3\lambda^3 + \lambda^3 + 4\lambda^4 + \lambda^4 + 5 \lambda^5 + ... }{\left ( 1+\lambda \right )}\\
	&= \frac{1+2\lambda +3\lambda^2 + 4\lambda^3 + 5\lambda^4 + ... }{\left ( 1+\lambda \right )}\\
	&= \frac{1+\left ( \lambda +\lambda  \right ) + \left ( \lambda^2 +\lambda\lambda+\lambda^2  \right ) + \left ( \lambda^3 +\lambda^2\lambda+\lambda \lambda^2+\lambda^3 \right ) + ... }{\left ( 1+\lambda \right )}\\
\end{align*}
\begin{align*}
	&= \frac{h_0\left ( \lambda, \lambda  \right ) + h_1\left ( \lambda, \lambda  \right ) + h_2\left ( \lambda, \lambda  \right ) + h_3\left ( \lambda, \lambda  \right ) + ... }{\left ( 1+\lambda \right )}\\
	&= \frac{\left ( \frac{1}{1-\lambda} \right )\left ( \frac{1}{1-\lambda} \right ) }{\left ( 1+\lambda \right )}\\
	&= \frac{1}{\left ( 1-\lambda \right )^2 \left ( 1+\lambda \right )}.
\end{align*}
$h_d\left ( \lambda, \lambda  \right ),\, d=0,1,2,...$, is the complete homogeneous symmetric polynomial with two variables. 
\begin{equation*}
	p\left ( 0,\lambda\right )=\left ( 1-\lambda \right )^2 \left ( 1+\lambda \right ) .
\end{equation*}
Substituting $\lambda_n$ and $p\left ( 0,\lambda\right )$ into $p\left ( n,\lambda\right )=\lambda_0\lambda_1...\lambda_{n-1}p\left ( 0,\lambda\right )$ gives
\begin{equation*}
	p\left ( n,\lambda\right )= \left ( 1-\lambda \right )^2 \left ( \frac{1+2\lambda +\left ( n-1 \right )\lambda }{1+\lambda } \right )\lambda^n \left ( 1+\lambda \right ) = \left ( 1-\lambda \right )^2 \left ( 1+\lambda +n\lambda \right )\lambda^n. 
\end{equation*}
Another easier derivation of this new form of the PL is obtained by setting $\phi=\left ( 1-\lambda  \right )/\lambda $.
\begin{equation*}
	p\left ( n,\phi \right )=\frac{\phi ^{2}\left ( \phi+2+n \right )}{\left ( \phi+1 \right )^{n+3}}.
\end{equation*}
\begin{equation*}
	p\left ( n,\lambda \right )=\left ( \frac{1-\lambda }{\lambda } \right )^{2}\left ( \frac{1-\lambda }{\lambda } +2+n \right )\lambda ^{n+3}=\left ( 1-\lambda \right )^2 \left ( 1+\lambda +n\lambda \right )\lambda^n.
\end{equation*}
We can recognize the PL distribution as a stationary birth-death process. This completes the derivations. 

\newpage
\section{Derivation of Equation (35)}
\vspace{-0.5cm}
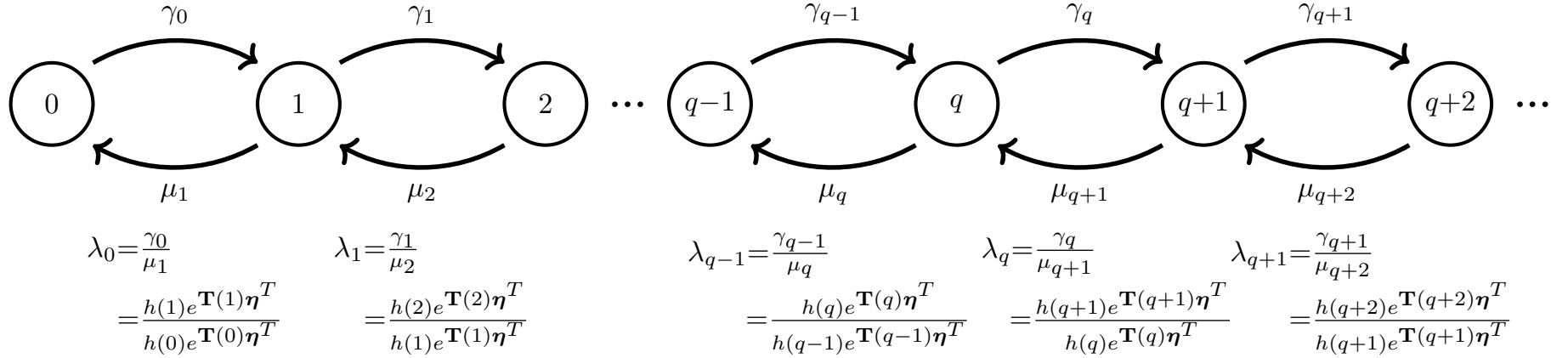
\begin{figure}[!h]
	\begin{subfigure}{\linewidth}
		\centering
		\resizebox{\linewidth}{!}{
			\begin{tikzpicture}
				\newcommand\xr{0.5}
				\newcommand\xs{2}

				\newcommand\xx{0} 
				\draw[black, very thick] (\xx,0) node {$\scriptstyle 0$} circle (\xr);
				\draw[ultra thick, ->] (\xx+\xr,\xr) arc (120:60:\xs);
				\draw[ultra thick, ->] (\xx+\xs+\xr,-\xr) arc (300:240:\xs);
				\node[above] at (\xx+\xr+\xs/2,0.8) {$\scriptstyle \gamma_0$};
				\node[below] at (\xx+\xr+\xs/2,-0.8) {$\scriptstyle \mu_1$};
				\node[below] at (\xx+\xr+\xs/2-0.55,-1.4) {$\scriptstyle \lambda_0=\frac{\gamma_0}{\mu_1}$};
				\node[below] at (\xx+\xr+\xs/2+0.3,-2.0) {$\scriptstyle =\frac{h\left ( 1 \right )e^{\textbf{T}\left ( 1 \right )\boldsymbol{\eta}^T}}{h\left ( 0 \right )e^{\textbf{T}\left ( 0 \right )\boldsymbol{\eta}^T}}$};
				
				\renewcommand\xx{3}
				\draw[black, very thick] (\xx,0) node {$\scriptstyle 1$} circle (\xr);
				\draw[ultra thick, ->] (\xx+\xr,\xr) arc (120:60:\xs);
				\draw[ultra thick, ->] (\xx+\xs+\xr,-\xr) arc (300:240:\xs);
				\node[above] at (\xx+\xr+\xs/2,0.8) {$\scriptstyle \gamma_1$};
				\node[below] at (\xx+\xr+\xs/2,-0.8) {$\scriptstyle \mu_2$};
				\node[below] at (\xx+\xr+\xs/2-0.55,-1.4) {$\scriptstyle \lambda_1=\frac{\gamma_1}{\mu_2}$};
				\node[below] at (\xx+\xr+\xs/2+0.3,-2.0) {$\scriptstyle =\frac{h\left ( 2 \right )e^{\textbf{T}\left ( 2 \right )\boldsymbol{\eta}^T}}{h\left ( 1 \right )e^{\textbf{T}\left ( 1 \right )\boldsymbol{\eta}^T}}$};
				
				\renewcommand\xx{2*3}
				\draw[black, very thick] (\xx,0) node {$\scriptstyle 2$} circle (\xr);
				\draw[black, very thick] (\xx+1,0) node {\textbf{...}};
				
				\renewcommand\xx{3*3-1}
				\draw[black, very thick] (\xx,0) node {$\scriptstyle q-1$} circle (\xr);
				\draw[ultra thick, ->] (\xx+\xr,\xr) arc (120:60:\xs);
				\draw[ultra thick, ->] (\xx+\xs+\xr,-\xr) arc (300:240:\xs);
				\node[above] at (\xx+\xr+\xs/2,0.8) {$\scriptstyle \gamma_{q-1}$};
				\node[below] at (\xx+\xr+\xs/2,-0.8) {$\scriptstyle \mu_{q}$};
				\node[below] at (\xx+\xr+\xs/2-0.9,-1.4) {$\scriptstyle \lambda_{q-1}=\frac{\gamma_{q-1}}{\mu_{q}}$};
				\node[below] at (\xx+\xr+\xs/2+0.3,-2.0) {$\scriptstyle =\frac{h\left ( q \right )e^{\textbf{T}\left ( q \right )\boldsymbol{\eta}^T}}{h\left ( q-1 \right )e^{\textbf{T}\left ( q-1 \right )\boldsymbol{\eta}^T}}$};
				
				\renewcommand\xx{4*3-1}
				\draw[black, very thick] (\xx,0) node {$\scriptstyle q$} circle (\xr);
				\draw[ultra thick, ->] (\xx+\xr,\xr) arc (120:60:\xs);
				\draw[ultra thick, ->] (\xx+\xs+\xr,-\xr) arc (300:240:\xs);
				\node[above] at (\xx+\xr+\xs/2,0.8) {$\scriptstyle \gamma_{q}$};
				\node[below] at (\xx+\xr+\xs/2,-0.8) {$\scriptstyle \mu_{q+1}$};
				\node[below] at (\xx+\xr+\xs/2-0.5,-1.4) {$\scriptstyle \lambda_{q}=\frac{\gamma_{q}}{\mu_{q+1}}$};
				\node[below] at (\xx+\xr+\xs/2+0.5,-2.0) {$\scriptstyle =\frac{h\left ( q+1 \right )e^{\textbf{T}\left ( q+1 \right )\boldsymbol{\eta}^T}}{h\left ( q \right )e^{\textbf{T}\left ( q \right )\boldsymbol{\eta}^T}}$};
				
				\renewcommand\xx{5*3-1}
				\draw[black, very thick] (\xx,0) node {$\scriptstyle q+1$} circle (\xr);
				\draw[ultra thick, ->] (\xx+\xr,\xr) arc (120:60:\xs);
				\draw[ultra thick, ->] (\xx+\xs+\xr,-\xr) arc (300:240:\xs);
				\node[above] at (\xx+\xr+\xs/2,0.8) {$\scriptstyle \gamma_{q+1}$};
				\node[below] at (\xx+\xr+\xs/2,-0.8) {$\scriptstyle \mu_{q+2}$};
				\node[below] at (\xx+\xr+\xs/2-0.3,-1.4) {$\scriptstyle \lambda_{q+1}=\frac{\gamma_{q+1}}{\mu_{q+2}}$};
				\node[below] at (\xx+\xr+\xs/2+0.9,-2.0) {$\scriptstyle =\frac{h\left ( q+2 \right )e^{\textbf{T}\left ( q+2 \right )\boldsymbol{\eta}^T}}{h\left ( q+1 \right )e^{\textbf{T}\left ( q+1 \right )\boldsymbol{\eta}^T}}$};
				
				\renewcommand\xx{6*3-1}
				\draw[black, very thick] (\xx,0) node {$\scriptstyle q+2$} circle (\xr);
				\draw[black, very thick] (\xx+1,0) node {\textbf{...}};
			\end{tikzpicture}
		}
	\end{subfigure}
	\caption{State-transition-rate ratio diagram for the exponential-family distribution}. 
\end{figure}
We know from Section 2.1 and Definition 1 of the main manuscript that 
\begin{equation*}
	\lambda_n=\frac{q(n+1,\boldsymbol{\eta})}{q(n,\boldsymbol{\eta})}=
	\frac{h\left ( n+1 \right )e^{\textbf{T}\left ( n+1 \right )\boldsymbol{\eta}^T-A(\boldsymbol{\eta})}}{h\left ( n \right )e^{\textbf{T}\left ( n \right )\boldsymbol{\eta}^T-A(\boldsymbol{\eta})}}
	=\frac{h\left ( n+1 \right )e^{\textbf{T}\left ( n+1 \right )\boldsymbol{\eta}^T}}{h\left ( n \right )e^{\textbf{T}\left ( n \right )\boldsymbol{\eta}^T}}.
\end{equation*}

\vspace{0.3cm}
\noindent Substituting $\lambda_n$ into $1+\lambda_0 + \lambda_0\lambda_1 + \lambda_0\lambda_1\lambda_2 + ...$ gives

\begin{align*}
	1+\lambda_0 + \lambda_0\lambda_1 + \lambda_0\lambda_1\lambda_2 + ...
	&=1+ \frac{h\left ( 1 \right )e^{\textbf{T}\left ( 1 \right )\boldsymbol{\eta}^T}}{h\left ( 0 \right )e^{\textbf{T}\left ( 0 \right )\boldsymbol{\eta}^T}}+ \frac{h\left ( 2 \right )e^{\textbf{T}\left ( 2 \right )\boldsymbol{\eta}^T}}{h\left ( 0 \right )e^{\textbf{T}\left ( 0 \right )\boldsymbol{\eta}^T}}+ \frac{h\left ( 3 \right )e^{\textbf{T}\left ( 3 \right )\boldsymbol{\eta}^T}}{h\left ( 0 \right )e^{\textbf{T}\left ( 0 \right )\boldsymbol{\eta}^T}} +...\\
	&=\frac{1}{h\left ( 0 \right ) e^{\textbf{T}\left ( 0 \right ) \boldsymbol{\eta}^T}} 
	\left ( h\left ( 0 \right )e^{\textbf{T}\left ( 0 \right )\boldsymbol{\eta}^T} + h\left ( 1 \right )e^{\textbf{T}\left ( 1 \right )\boldsymbol{\eta}^T} + h\left ( 2 \right )e^{\textbf{T}\left ( 2 \right )\boldsymbol{\eta}^T}+...\right ) \\
	&=\frac{1}{h\left ( 0 \right ) e^{\textbf{T}\left ( 0 \right ) \boldsymbol{\eta}^T}} e^{A(\boldsymbol{\eta})}.
\end{align*}
\begin{equation*}
	e^{A(\boldsymbol{\eta})}= \left ( 1+\lambda_0 + \lambda_0\lambda_1 + \lambda_0\lambda_1\lambda_2 + ... \right ) h\left ( 0 \right )e^{\textbf{T}\left ( 0 \right ) \boldsymbol{\eta}^T }.
\end{equation*}
Taking the logarithm of both sides yields

\begin{equation*}
	A(\boldsymbol{\eta}) - \log h\left ( 0 \right ) - \textbf{T}\left ( 0 \right ) \boldsymbol{\eta}^T=\log \left( 1+\lambda_0 + \lambda_0\lambda_1 + \lambda_0\lambda_1\lambda_2 + ...\right).
\end{equation*}
This completes the derivation.

\section{Proof of Proposition 4}
By equating the type 1 and mixture probabilities, we obtain
\begin{equation*}
	w_n+\left ( 1-\sum_{i=0}^{m} w_{n_i} \right )b\left ( n,\boldsymbol{\theta} \right )=\frac{\alpha_n b\left ( n,\boldsymbol{\theta} \right )}{z (\boldsymbol{\theta},\boldsymbol{\alpha})},
\end{equation*}
and
\begin{equation*}
	\left ( 1-\sum_{i=0}^{m} w_{n_i} \right )=\frac{1}{z (\boldsymbol{\theta},\boldsymbol{\alpha})}.
\end{equation*}
Substitution of $1-\sum_{i=0}^{m} w_{n_i} $ into the first equation gives
\begin{align*}
	w_n+\frac{b\left ( n,\boldsymbol{\theta} \right )}{z (\boldsymbol{\theta},\boldsymbol{\alpha})}&=\frac{\alpha_n b\left ( n,\boldsymbol{\theta} \right )}{z (\boldsymbol{\theta},\boldsymbol{\alpha})}.\\
	w_n&=\frac{\left ( \alpha_n-1 \right ) b\left ( n,\boldsymbol{\theta} \right )}{z (\boldsymbol{\theta},\boldsymbol{\alpha})}.
\end{align*}
Substitution of $z (\boldsymbol{\theta},\boldsymbol{\alpha})$ into the first equation gives
\begin{align*}
	w_n+\left ( 1-\sum_{i=0}^{m} w_{n_i} \right )b\left ( n,\boldsymbol{\theta} \right )&=\alpha_n \left ( 1-\sum_{i=0}^{m} w_{n_i} \right ) b\left ( n,\boldsymbol{\theta} \right ). \\
	\alpha_n &= \frac{w_n+\left ( 1-\sum_{i=0}^{m} w_{n_i} \right ) b\left ( n,\boldsymbol{\theta} \right )}{\left ( 1-\sum_{i=0}^{m} w_{n_i} \right )b\left ( n,\boldsymbol{\theta} \right )}.
\end{align*}
This completes the proof.

\section{Proof of Proposition 6}
From Proposition 5 and (25) of the main menuscript we know that 
\begin{equation*}
	A (\boldsymbol{\eta})=A_b (\boldsymbol{\eta}_b) + \log z (\boldsymbol{\eta})=A_b (\boldsymbol{\eta}_b) + \log \left ( 1+\sum_{i=0}^{m} \left ( \prod_{j=i}^{m}e^{\eta_{l+j+1}}-1 \right )  \sum_{k=n_{i-1}+1}^{n_i}b\left ( k,\boldsymbol{\eta}_b \right ) \right ),
\end{equation*}
where $b\left ( k,\boldsymbol{\eta}_b \right )=h \left ( k \right ) e^{\sum_{s=1}^{l}T_s\left ( k \right )\eta_s -A_b \left ( \boldsymbol{\eta}_b \right )}$, and $\boldsymbol{\eta}=\left [ \boldsymbol{\eta}_b,\eta_{l+1},...,\eta_{l+m+1} \right ]$.\\

\noindent Applying Corollary 1.6.1 in \citet[p.~59]{sBickel15}, we obtain
\begin{align*}
	E\left [ N \right ]	
	&=\frac{\partial A_b}{\partial \eta_1}\left ( \boldsymbol{\eta}_b \right ) + \frac{\partial \log z}{\partial \eta_1} \left ( \boldsymbol{\eta} \right ) \\
	&=\frac{\partial A_b}{\partial \eta_1}\left ( \boldsymbol{\eta}_b \right ) + \frac{1}{z \left ( \boldsymbol{\eta} \right )} \frac{\partial  z}{\partial \eta_1}\left ( \boldsymbol{\eta}  \right )\\
	&=E_b\left [ N \right ] + \frac{1}{z \left ( \boldsymbol{\eta} \right )}\sum_{i=0}^{m} \left ( \prod_{j=i}^{m}e^{\eta_{l+j+1}}-1 \right ) \sum_{k=n_{i-1}+1}^{n_i} \frac{\partial  b}{\partial \eta_1}\left ( k,\boldsymbol{\eta}_b \right )\\
	&=E_b\left [ N \right ] + \frac{1}{z \left ( \boldsymbol{\eta} \right )} \sum_{i=0}^{m} \left ( \prod_{j=i}^{m}e^{\eta_{l+j+1}}-1 \right ) \sum_{k=n_{i-1}+1}^{n_i} h \left ( k \right ) e^{\sum_{s=1}^{l}T_s\left ( k \right )\eta_s -A_b \left ( \boldsymbol{\eta}_b \right )} \left ( T_{1}\left ( k \right ) - \frac{\partial  A_b}{\partial \eta_1} \left ( \boldsymbol{\eta}_b \right ) \right )  \\
	&=E_b\left [ N \right ] + \frac{1}{z \left ( \boldsymbol{\eta} \right )} \sum_{i=0}^{m} \left ( \prod_{j=i}^{m}e^{\eta_{l+j+1}}-1 \right ) \sum_{k=n_{i-1}+1}^{n_i} b\left ( k,\boldsymbol{\eta}_b \right ) \left ( T_1 \left ( k \right ) - \frac{\partial A_b}{\partial \eta_1} \left ( \boldsymbol{\eta}_b \right ) \right )\\
	&=E_b\left [ N \right ] + \frac{1}{z\left ( \boldsymbol{\theta},\boldsymbol{\varphi} \right )} \sum_{i=0}^{m} \left ( \prod_{j=i}^{m}\varphi_{n_j}-1 \right ) \sum_{k=n_{i-1}+1}^{n_i} b\left ( k,\boldsymbol{\theta} \right ) \left ( k - E_b \left [ N \right ] \right ), 		
\end{align*}
and

\begin{align*}
	V\left [ N \right ]
	&=\frac{\partial^2 A_b}{\partial \eta_1^2}\left ( \boldsymbol{\eta}_b \right ) + \frac{1}{z \left ( \boldsymbol{\eta} \right )^2} \left ( z\left ( \boldsymbol{\eta} \right ) \frac{\partial^2  z}{\partial \eta_1^2}\left ( \boldsymbol{\eta} \right ) -\left ( \frac{\partial  z}{\partial \eta_1}\left ( \boldsymbol{\eta} \right ) \right )^2 \right )\\
	&=V_b\left [ N \right ] + \frac{1}{z \left ( \boldsymbol{\eta} \right )}  \sum_{i=0}^{m} \left ( \prod_{j=i}^{m}e^{\eta_{l+j+1}}-1 \right ) \sum_{k=n_{i-1}+1}^{n_i} \frac{\partial^2  b}{\partial \eta_1^2}\left ( k,\boldsymbol{\eta}_b \right ) - \frac{1}{z \left ( \boldsymbol{\eta} \right )^2} \left ( z \left ( \boldsymbol{\eta} \right )   \left ( E\left [ N \right ] - E_b\left [ N \right ] \right ) \right )^2 \\		
	&=V_b\left [ N \right ] -  \left ( E\left [ N \right ] - E_b\left [ N \right ]  \right )^2 + \frac{1}{z \left ( \boldsymbol{\eta} \right )}  \sum_{i=0}^{m} \left ( \prod_{j=i}^{m}e^{\eta_{l+j+1}}-1 \right ) \sum_{k=n_{i-1}+1}^{n_i} \frac{\partial }{\partial \eta_1} \left ( b\left ( k,\boldsymbol{\eta}_b \right ) \left ( T_1 \left ( k \right )-\frac{\partial A_b}{\partial \eta_1} \left ( \boldsymbol{\eta}_b \right ) \right )  \right )\\
	&=V_b\left [ N \right ] -  \left ( E\left [ N \right ] - E_b\left [ N \right ]  \right )^2 + \frac{1}{z \left ( \boldsymbol{\eta} \right )}  \sum_{i=0}^{m} \left ( \prod_{j=i}^{m}e^{\eta_{l+j+1}}-1 \right ) \sum_{k=n_{i-1}+1}^{n_i} \left ( \left ( T_1 \left ( k \right )-\frac{\partial A_b}{\partial \eta_1} \left ( \boldsymbol{\eta}_b \right ) \right ) \frac{\partial b}{\partial \eta_1} \left ( k, \boldsymbol{\eta}_b \right )  \right. \\
	&\quad \left. + b\left ( k, \boldsymbol{\eta}_b \right ) \frac{\partial }{\partial \eta_1} \left ( T_1 \left ( k \right )-\frac{\partial A_b}{\partial \eta_1} \left ( \boldsymbol{\eta}_b \right ) \right ) \right )\\
	&=V_b\left [ N \right ] -  \left ( E\left [ N \right ] - E_b\left [ N \right ]  \right )^2 + \frac{1}{z \left ( \boldsymbol{\eta} \right )}  \sum_{i=0}^{m} \left ( \prod_{j=i}^{m}e^{\eta_{l+j+1}}-1 \right ) \sum_{k=n_{i-1}+1}^{n_i} \left ( b \left ( k, \boldsymbol{\eta}_b \right ) \left ( T_1 \left ( k \right )-\frac{\partial A_b}{\partial \eta_1} \left ( \boldsymbol{\eta}_b \right ) \right )^2  \right.\\
	&\quad \left. - b\left ( k, \boldsymbol{\eta}_b \right ) \frac{\partial^2 A_b}{\partial \eta_1^2} \left ( \boldsymbol{\eta}_b \right ) \right )\\
	&=V_b\left [ N \right ] -  \left ( E\left [ N \right ] - E_b\left [ N \right ]  \right )^2 + \frac{1}{z \left ( \boldsymbol{\theta},\boldsymbol{\varphi} \right )}  \sum_{i=0}^{m} \left ( \prod_{j=i}^{m}\varphi_{n_j}-1 \right ) \sum_{k=n_{i-1}+1}^{n_i} b \left ( k, \boldsymbol{\theta} \right ) \left ( \left ( k- E_b\left [ N \right ] \right )^2 - V_b\left [ N \right ]  \right ).
\end{align*}
This completes the proof.

\section{Proof of Proposition 7}
From Proposition 6 we know that if the base distribution is Poisson, and $\mathcal{F}= \left\{q \neq0 \right\}$, then
\begin{align*}
	E[N]&=E_b[N]+\frac{1}{z\left ( \boldsymbol{\theta},\boldsymbol{\varphi} \right )}\sum_{i=0}^{m} \left ( \prod_{j=i}^{m}\varphi_{n_j}-1 \right )  \sum_{k=n_{i-1}+1}^{n_i}b\left ( k,\boldsymbol{\theta} \right ) \left ( k-E_b[N] \right )\\
	&= \lambda+\frac{1}{z\left (\lambda,\varphi \right )}  \left ( \varphi-1 \right ) \sum_{k=0}^{q}\left ( k-\lambda \right )\frac{\lambda^k e^{-\lambda}}{k!},\, \text{and}
\end{align*}\vspace{-0.5cm} 
\begin{align*}
	V[N] &=V_b [N]-\left ( E[N]- E_b[N]\right )^2   +\frac{1}{z\left ( \boldsymbol{\theta},\boldsymbol{\varphi} \right )}\sum_{i=0}^{m} \left ( \prod_{j=i}^{m}\varphi_{n_j}-1 \right )  \sum_{k=n_{i-1}+1}^{n_i}b\left ( k,\boldsymbol{\theta} \right ) \left ( \left ( k-E_b[N] \right )^2-V_b[N] \right ) \\
	&= \lambda-\left ( q-\lambda \right )^2+\frac{1}{z\left ( \lambda,\varphi  \right )} \left ( \varphi-1 \right )  \sum_{k=0}^{q} \left (\left ( k-\lambda \right )^2-\lambda  \right ) \frac{\lambda^k e^{-\lambda}}{k!},
\end{align*}
where $\displaystyle z\left ( \lambda,\varphi \right ) =1+\left (\varphi-1 \right )\sum_{k=0}^{q}\frac{\lambda^k e^{-\lambda}}{k!}$.\\
Let
\begin{equation*}
	\varphi =  1+\frac{\lambda-q}{\sum_{k=0}^{q-1}\left ( q-k \right )\frac{\lambda^k e^{-\lambda}}{k!}}.
\end{equation*}
Then
\begin{align*}
	z\left ( \lambda,\varphi \right )&\displaystyle = 1+\left ( \varphi-1 \right )\sum_{k=0}^{q}\frac{\lambda^k e^{-\lambda}}{k!}\\
	&\displaystyle = 1+\frac{\lambda-q}{\sum_{k=0}^{q-1}\left ( q-k \right )\frac{\lambda^k e^{-\lambda}}{k!}}\sum_{k=0}^{q}\frac{\lambda^k e^{-\lambda}}{k!}\\
	&\displaystyle = \frac{\sum_{k=0}^{q-1}\left ( q-k \right )\frac{\lambda^k e^{-\lambda}}{k!}+\left ( \lambda-q \right )\sum_{k=0}^{q}\frac{\lambda^k e^{-\lambda}}{k!}}{\sum_{k=0}^{q-1}\left ( q-k \right )\frac{\lambda^k e^{-\lambda}}{k!}}\\
	&\displaystyle = \frac{\sum_{k=0}^{q}\left ( q-k \right )\frac{\lambda^k e^{-\lambda}}{k!}+\sum_{k=0}^{q}\left ( \lambda-q \right )\frac{\lambda^k e^{-\lambda}}{k!}}{\sum_{k=0}^{q-1}\left ( q-k \right )\frac{\lambda^k e^{-\lambda}}{k!}}\\
	&\displaystyle = \frac{\sum_{k=0}^{q} \left ( \lambda-k \right ) \frac{\lambda^k e^{-\lambda}}{k!}}{\sum_{k=0}^{q-1}\left ( q-k \right )\frac{\lambda^k e^{-\lambda}}{k!}},
\end{align*}

\begin{align*}
	E[N]
	&= \lambda+\left ( \frac{\sum_{k=0}^{q-1}  \left ( q-k \right )\frac{\lambda^k e^{-\lambda}}{k!}}{\sum_{k=0}^{q} \left ( \lambda-k \right ) \frac{\lambda^k e^{-\lambda}}{k!}} \right ) \left ( \frac{\lambda-q}{\sum_{k=0}^{q-1}\left ( q-k \right )\frac{\lambda^k e^{-\lambda}}{k!}} \right )    \sum_{k=0}^{q}\left ( k-\lambda \right )\frac{\lambda^k e^{-\lambda}}{k!}\\
	&= \lambda-\left ( \lambda-q \right )\frac{\sum_{k=0}^{q}\left ( \lambda-k \right )\frac{\lambda^k e^{-\lambda}}{k!}}{\sum_{k=0}^{q}\left ( \lambda-k \right )\frac{\lambda^k e^{-\lambda}}{k!}} \\
	&=q,
\end{align*}
and 
\begin{align*}
	V[N] 
	&= \lambda-\left ( q-\lambda \right )^2 +\left ( \frac{\sum_{k=0}^{q-1}\left ( q-k \right )\frac{\lambda^k e^{-\lambda}}{k!}}{\sum_{k=0}^{q}\left ( \lambda-k \right ) \frac{\lambda^k e^{-\lambda}}{k!}} \right ) \left ( \frac{\lambda-q}{\sum_{k=0}^{q-1}\left ( q-k \right )\frac{\lambda^k e^{-\lambda}}{k!}} \right )\sum_{k=0}^{q} \left ( \left ( k-\lambda \right )^2-\lambda \right ) \frac{\lambda^k e^{-\lambda}}{k!}\\	
	&= \lambda-\left ( q-\lambda \right )^2 +\left ( \lambda-q \right )\frac{\sum_{k=0}^{q} \left ( \left ( k-\lambda \right )^2-\lambda \right ) \frac{\lambda^k e^{-\lambda}}{k!}}{\sum_{k=0}^{q}\left ( \lambda-k \right )\frac{\lambda^k e^{-\lambda}}{k!}}\\
	&= \lambda-\left ( q-\lambda \right )^2 +\left ( \lambda-q \right )\frac{\sum_{k=0}^{q} \left ( k^2-2k \lambda+ \lambda^2-\lambda \right ) \frac{\lambda^k e^{-\lambda}}{k!}}{\sum_{k=0}^{q}\left ( \lambda-k \right )\frac{\lambda^k e^{-\lambda}}{k!}}\\
	&= \lambda-\left ( q-\lambda \right )^2 +\left ( \lambda-q \right )\frac{\sum_{k=0}^{q}k^2 \frac{\lambda^k e^{-\lambda}}{k!}-2\sum_{k=0}^{q} k \lambda \frac{\lambda^k e^{-\lambda}}{k!} +\sum_{k=0}^{q} \lambda^2 \frac{\lambda^k e^{-\lambda}}{k!}-\sum_{k=0}^{q} \lambda \frac{\lambda^k e^{-\lambda}}{k!} }{\sum_{k=0}^{q} \lambda \frac{\lambda^k e^{-\lambda}}{k!} - \sum_{k=0}^{q} k \frac{\lambda^k e^{-\lambda}}{k!}}\\
	&= \lambda-\left ( q-\lambda \right )^2 +\left ( \lambda-q \right )\frac{\sum_{k=1}^{q} k \frac{\lambda^k e^{-\lambda}}{\left ( k-1 \right )!} - \sum_{k=0}^{q} \frac{\lambda^{k+1} e^{-\lambda}}{k!} - 2\sum_{k=1}^{q} \frac{\lambda^{k+1} e^{-\lambda}}{\left ( k-1 \right )!} + \sum_{k=0}^{q} \frac{\lambda^{k+2} e^{-\lambda}}{k!}}{\sum_{k=0}^{q} \frac{\lambda^{k+1} e^{-\lambda}}{k!} - \sum_{k=1}^{q} \frac{\lambda^k e^{-\lambda}}{\left ( k-1 \right )!}}\\
	&= \lambda-\left ( q-\lambda \right )^2 +\left ( \lambda-q \right )\frac{\sum_{k=0}^{q-1} \left ( k+1 \right ) \frac{\lambda^{k+1} e^{-\lambda}}{k!}-\sum_{k=0}^{q} \frac{\lambda^{k+1} e^{-\lambda}}{k!}-2\sum_{k=0}^{q-1} \frac{\lambda^{k+2} e^{-\lambda}}{k!} +\sum_{k=0}^{q} \frac{\lambda^{k+2} e^{-\lambda}}{k!}}{\sum_{k=0}^{q-1} \frac{\lambda^{k+1} e^{-\lambda}}{k!} +\frac{\lambda^{q+1} e^{-\lambda}}{q!} - \sum_{k=0}^{q-1} \frac{\lambda^{k+1} e^{-\lambda}}{k!}}\\
	&= \lambda-\left ( q-\lambda \right )^2 +\left ( \lambda-q \right )\frac{\sum_{k=0}^{q-1} k \frac{\lambda^{k+1} e^{-\lambda}}{k!}-\frac{\lambda^{q+1} e^{-\lambda}}{q!} -\sum_{k=0}^{q-1} \frac{\lambda^{k+2} e^{-\lambda}}{k!} +\frac{\lambda^{q+2} e^{-\lambda}}{q!}}{\frac{\lambda^{q+1} e^{-\lambda}}{q!}}\\
	&= \lambda-\left ( q-\lambda \right )^2 +\left ( \lambda-q \right )\frac{\sum_{k=1}^{q-1} \frac{\lambda^{k+1} e^{-\lambda}}{\left ( k-1 \right )!} -\sum_{k=0}^{q-1} \frac{\lambda^{k+2} e^{-\lambda}}{k!} -\frac{\lambda^{q+1} e^{-\lambda}}{q!} +\frac{\lambda^{q+2} e^{-\lambda}}{q!}}{\frac{\lambda^{q+1} e^{-\lambda}}{q!}}\\
	&= \lambda-\left ( q-\lambda \right )^2 +\left ( \lambda-q \right )\frac{\sum_{k=0}^{q-2} \frac{\lambda^{k+2} e^{-\lambda}}{k!} -\sum_{k=0}^{q-1} \frac{\lambda^{k+2} e^{-\lambda}}{k!} -\frac{\lambda^{q+1} e^{-\lambda}}{q!} +\frac{\lambda^{q+2} e^{-\lambda}}{q!}}{\frac{\lambda^{q+1} e^{-\lambda}}{q!}}\\
\end{align*}

\begin{align*}
	&= \lambda-\left ( q-\lambda \right )^2 +\left ( \lambda-q \right )\frac{-\frac{\lambda^{q+1} e^{-\lambda}}{\left ( q-1 \right )!} -\frac{\lambda^{q+1} e^{-\lambda}}{q!} +\frac{\lambda^{q+2} e^{-\lambda}}{q!}}{\frac{\lambda^{q+1} e^{-\lambda}}{q!}}\\
	&= \lambda-\left ( q-\lambda \right )^2 +\left ( \lambda-q \right )\left ( -q-1+\lambda \right )\\
	&= \lambda-\left ( q-\lambda \right )^2 +\left ( \lambda-q \right )^2 -\left ( \lambda-q \right )\\
	&= q.
\end{align*}
This completes the proof.

\section{Mean and variance of the type 1 model}

\begin{proposition}
	\label{type1_meanvariance}
	Suppose $N$ is a type 1 stationary distribution (13), and $b(n,\boldsymbol{\theta})$ is a member of the exponential family with $T_1(n)=n$. Then
	\begin{equation*} 
		E[N]=E_b[N]+\frac{1}{z\left ( \boldsymbol{\theta},\boldsymbol{\alpha} \right )}\sum_{i=0}^{m} \left ( \alpha_{n_i}-1 \right ) b\left ( n_i,\boldsymbol{\theta} \right ) \left ( n_i-E_b[N] \right ),
	\end{equation*}
	and
	\begin{align*}
		&V[N] =V_b [N]-\left ( E[N]- E_b[N]\right )^2 +\frac{1}{z\left ( \boldsymbol{\theta},\boldsymbol{\alpha} \right )}\sum_{i=0}^{m} \left ( \alpha_{n_i}-1 \right ) b\left ( n_i,\boldsymbol{\theta} \right ) \left ( \left ( n_i-E_b[N] \right )^2-V_b[N] \right ),
	\end{align*}
	$E_b[N]$ and $V_b[N]$ denote the mean and variance of $b(n,\boldsymbol{\theta})$, respectively.
\end{proposition}

\noindent \textit{Proof.} From Proposition 5 and (15) of the main menuscript we know that 
\begin{equation*}
	A (\boldsymbol{\eta})=A_b (\boldsymbol{\eta}_b) + \log z (\boldsymbol{\eta})=A_b (\boldsymbol{\eta}_b) + \log \left ( 1+\sum_{i=0}^{m} \left ( e^{\eta_{l+i+1}}-1 \right ) b\left ( n_i,\boldsymbol{\eta}_b \right ) \right ),
\end{equation*}
where $b\left ( n_i,\boldsymbol{\eta}_b \right )=h \left ( n_i \right ) e^{\sum_{s=1}^{l}T_s\left ( n_i \right )\eta_s -A_b \left ( \boldsymbol{\eta}_b \right )}$, and $\boldsymbol{\eta}=\left [ \boldsymbol{\eta}_b,\eta_{l+1},...,\eta_{l+m+1} \right ]$.\\

\noindent Applying Corollary 1.6.1 in \citet[p.~59]{sBickel15}, we obtain
\begin{align*}
	E\left [ N \right ]	
	&=\frac{\partial A_b}{\partial \eta_1}\left ( \boldsymbol{\eta}_b \right ) + \frac{\partial \log z}{\partial \eta_1} \left ( \boldsymbol{\eta} \right ) \\
	&=\frac{\partial A_b}{\partial \eta_1}\left ( \boldsymbol{\eta}_b \right ) + \frac{1}{z \left ( \boldsymbol{\eta} \right )} \frac{\partial  z}{\partial \eta_1}\left ( \boldsymbol{\eta}  \right )\\
	&=E_b\left [ N \right ] + \frac{1}{z \left ( \boldsymbol{\eta} \right )}\sum_{i=0}^{m} \left ( e^{\eta_{l+i+1}}-1 \right ) \frac{\partial  b}{\partial \eta_1}\left ( n_i,\boldsymbol{\eta}_b \right )\\
	&=E_b\left [ N \right ] + \frac{1}{z \left ( \boldsymbol{\eta} \right )} \sum_{i=0}^{m} \left ( e^{\eta_{l+i+1}}-1 \right ) h \left ( n_i \right ) e^{\sum_{s=1}^{l}T_s\left ( n_i \right )\eta_s -A_b \left ( \boldsymbol{\eta}_b \right )} \left ( T_{1}\left ( n_i \right ) - \frac{\partial  A_b}{\partial \eta_1} \left ( \boldsymbol{\eta}_b \right ) \right )  \\
	&=E_b\left [ N \right ] + \frac{1}{z \left ( \boldsymbol{\eta} \right )} \sum_{i=0}^{m} \left ( e^{\eta_{l+i+1}}-1 \right ) b\left ( n_i,\boldsymbol{\eta}_b \right ) \left ( T_1 \left ( n_i \right ) - \frac{\partial A_b}{\partial \eta_1} \left ( \boldsymbol{\eta}_b \right ) \right )\\
	&=E_b\left [ N \right ] + \frac{1}{z\left ( \boldsymbol{\theta},\boldsymbol{\alpha} \right )} \sum_{i=0}^{m} \left ( \alpha_{n_i}-1 \right ) b\left ( n_i,\boldsymbol{\theta} \right ) \left ( n_i - E_b \left [ N \right ] \right ), 		
\end{align*}
and

\begin{align*}
	V\left [ N \right ]
	&=\frac{\partial^2 A_b}{\partial \eta_1^2}\left ( \boldsymbol{\eta}_b \right ) + \frac{1}{z \left ( \boldsymbol{\eta} \right )^2} \left ( z\left ( \boldsymbol{\eta} \right ) \frac{\partial^2  z}{\partial \eta_1^2}\left ( \boldsymbol{\eta} \right ) -\left ( \frac{\partial  z}{\partial \eta_1}\left ( \boldsymbol{\eta} \right ) \right )^2 \right )\\
	&=V_b\left [ N \right ] + \frac{1}{z \left ( \boldsymbol{\eta} \right )}  \sum_{i=0}^{m} \left ( e^{\eta_{l+i+1}}-1 \right ) \frac{\partial^2  b}{\partial \eta_1^2}\left ( n_i,\boldsymbol{\eta}_b \right ) - \frac{1}{z \left ( \boldsymbol{\eta} \right )^2} \left ( z \left ( \boldsymbol{\eta} \right )   \left ( E\left [ N \right ] - E_b\left [ N \right ] \right ) \right )^2 \\		
	&=V_b\left [ N \right ] -  \left ( E\left [ N \right ] - E_b\left [ N \right ]  \right )^2 + \frac{1}{z \left ( \boldsymbol{\eta} \right )}  \sum_{i=0}^{m} \left ( e^{\eta_{l+i+1}}-1 \right ) \frac{\partial }{\partial \eta_1} \left ( b\left ( n_i,\boldsymbol{\eta}_b \right ) \left ( T_1 \left ( n_i \right )-\frac{\partial A_b}{\partial \eta_1} \left ( \boldsymbol{\eta}_b \right ) \right )  \right )\\
	&=V_b\left [ N \right ] -  \left ( E\left [ N \right ] - E_b\left [ N \right ]  \right )^2 + \frac{1}{z \left ( \boldsymbol{\eta} \right )}  \sum_{i=0}^{m} \left ( e^{\eta_{l+i+1}}-1 \right ) \left ( \left ( T_1 \left ( n_i \right )-\frac{\partial A_b}{\partial \eta_1} \left ( \boldsymbol{\eta}_b \right ) \right ) \frac{\partial b}{\partial \eta_1} \left ( n_i, \boldsymbol{\eta}_b \right )  \right. \\
	&\quad \left. + b\left ( n_i, \boldsymbol{\eta}_b \right ) \frac{\partial }{\partial \eta_1} \left ( T_1 \left ( n_i \right )-\frac{\partial A_b}{\partial \eta_1} \left ( \boldsymbol{\eta}_b \right ) \right ) \right )\\
	&=V_b\left [ N \right ] -  \left ( E\left [ N \right ] - E_b\left [ N \right ]  \right )^2 + \frac{1}{z \left ( \boldsymbol{\eta} \right )}  \sum_{i=0}^{m} \left ( e^{\eta_{l+i+1}}-1 \right ) \left ( b \left ( n_i, \boldsymbol{\eta}_b \right ) \left ( T_1 \left ( n_i \right )-\frac{\partial A_b}{\partial \eta_1} \left ( \boldsymbol{\eta}_b \right ) \right )^2  \right.\\
	&\quad \left. - b\left ( n_i, \boldsymbol{\eta}_b \right ) \frac{\partial^2 A_b}{\partial \eta_1^2} \left ( \boldsymbol{\eta}_b \right ) \right )\\
	&=V_b\left [ N \right ] -  \left ( E\left [ N \right ] - E_b\left [ N \right ]  \right )^2 + \frac{1}{z \left ( \boldsymbol{\theta},\boldsymbol{\alpha} \right )}  \sum_{i=0}^{m} \left ( \alpha_{n_i}-1 \right ) b \left ( n_i, \boldsymbol{\theta} \right ) \left ( \left ( n_i- E_b\left [ N \right ] \right )^2 - V_b\left [ N \right ]  \right ).
\end{align*}
This completes the proof.

\afterpage{
	\begin{figure}[H]
		\begin{subfigure}[b]{0.245\linewidth}
			\centering
			\includegraphics[height=4.4cm]{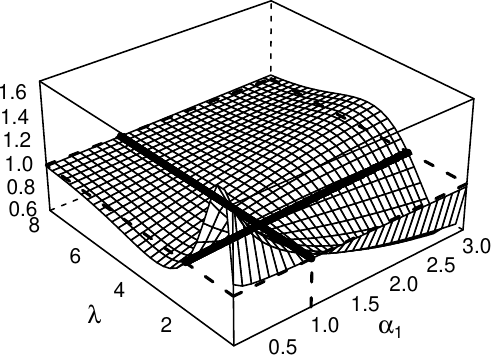}
			\caption{$q=1$}
			\label{fig:ssfig1}
		\end{subfigure}
		\begin{subfigure}[b]{0.245\linewidth}
			\includegraphics[height=4.4cm]{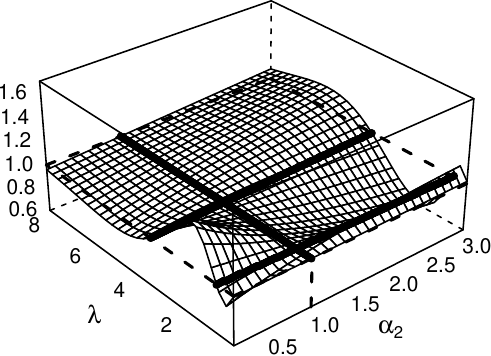}
			\caption{$q=2$}
			\label{fig:ssfig2}
		\end{subfigure}
		\begin{subfigure}[b]{0.245\linewidth}
			\includegraphics[height=4.4cm]{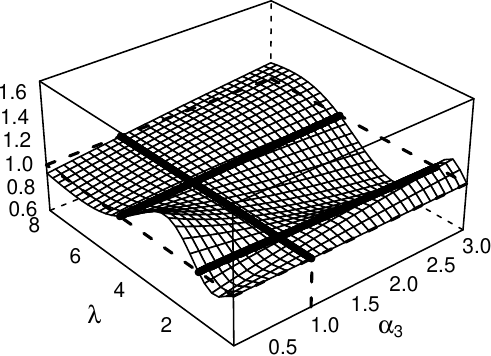}
			\caption{$q=3$}
			\label{fig:ssfig3}
		\end{subfigure}
		\begin{subfigure}[b]{0.245\linewidth}
			\includegraphics[height=4.4cm]{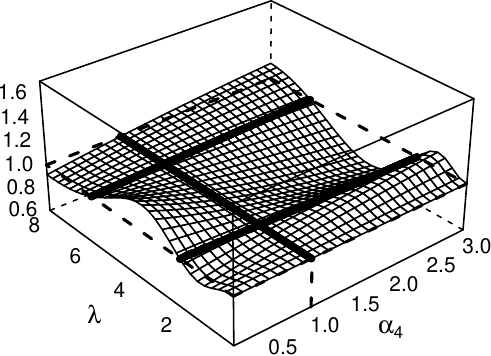}
			\caption{$q=4$}
			\label{fig:ssfig4}
		\end{subfigure} \vspace{0.1em}
		
		\begin{subfigure}[b]{0.245\linewidth}
			\centering
			\includegraphics[height=4.4cm]{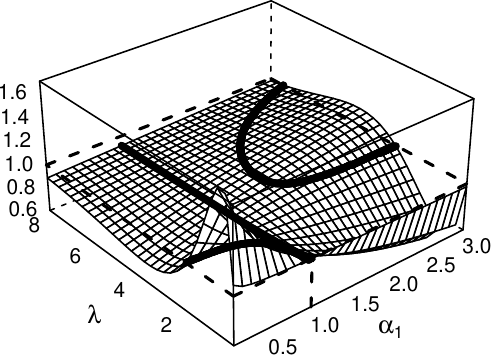}
			\caption{$q=1$, $\nu=1.1$}
			\label{fig:ssfig5}
		\end{subfigure}
		\begin{subfigure}[b]{0.245\linewidth}
			\includegraphics[height=4.4cm]{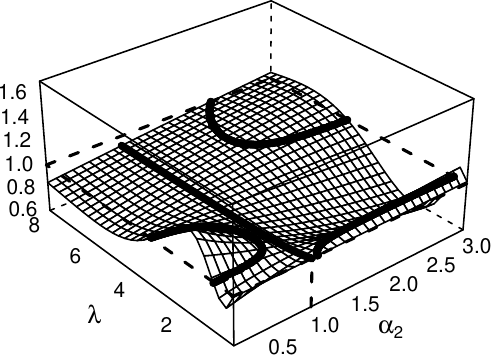}
			\caption{$q=2$, $\nu=1.1$}
			\label{fig:ssfig6}
		\end{subfigure}
		\begin{subfigure}[b]{0.245\linewidth}
			\includegraphics[height=4.4cm]{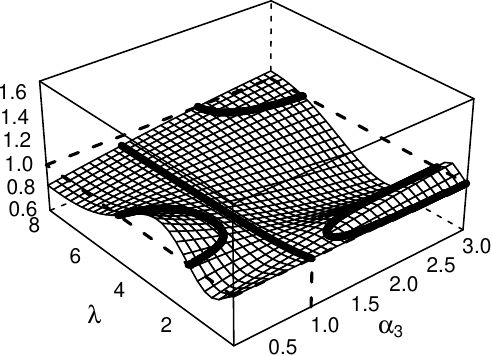}
			\caption{$q=3$, $\nu=1.1$ }
			\label{fig:ssfig7}
		\end{subfigure}
		\begin{subfigure}[b]{0.245\linewidth}
			\includegraphics[height=4.4cm]{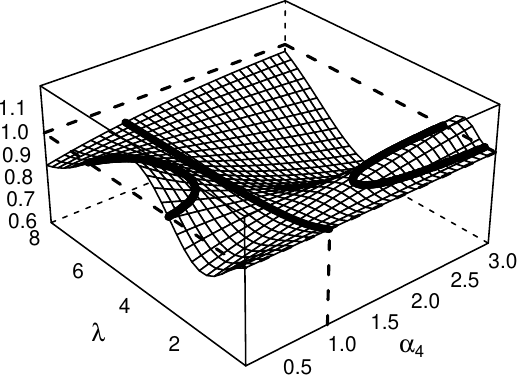}
			\caption{$q=4$, $\nu=1.1$}
			\label{fig:ssfig8}
		\end{subfigure} \vspace{0.1em}
		\caption{Dispersion surfaces of the type 1 Poisson (top) and CMP (bottom) models with $\mathcal{F}=\left\{ q \right\}$. The straight and curved (Poisson and non-Poisson) lines display equidispersion , that is, the level lines are equal to one. An exception is the CMP lines with $\alpha_q=1$, which show underdispersion.}
		\label{fig:vmr_type1}
		
	\end{figure}
}

\end{landscape}
\end{document}